\documentclass[twoside,11pt]{article}

\usepackage{blindtext}

\usepackage{amsmath,amsthm,amssymb}
\usepackage{float}
\usepackage[abbrvbib, preprint]{jmlr2e}
\usepackage{breqn}
\usepackage{multirow}
\usepackage{makecell}
\usepackage{booktabs}
\usepackage{mathtools}
\usepackage{url}
\usepackage{bm}
\usepackage{graphicx}
\usepackage{setspace}
\usepackage{enumitem}
\usepackage{tabularx}
\usepackage{pifont}
\usepackage{algorithm}
\usepackage{algpseudocode}
\usepackage{mathrsfs}
\usepackage[scr=boondox]{mathalpha}
\usepackage{array}
\newcolumntype{L}[1]{>{\raggedright\arraybackslash}m{#1}}
\newcolumntype{C}[1]{>{\centering\arraybackslash}m{#1}}
\hypersetup{hidelinks}

\usepackage{aliascnt}
\usepackage[capitalize,nameinlink]{cleveref}

\crefname{appendix}{appendix}{appendices}
\Crefname{appendix}{Appendix}{Appendices}
\crefname{subappendix}{appendix}{appendices}
\Crefname{subappendix}{Appendix}{Appendices}
\crefname{subsubappendix}{appendix}{appendices}
\Crefname{subsubappendix}{Appendix}{Appendices}

\crefname{onlineappendix}{Online Appendix}{Online Appendices}
\Crefname{onlineappendix}{Online Appendix}{Online Appendices}

\theoremstyle{plain}

\newaliascnt{lemma}{theorem}
\newtheorem{lemma}[lemma]{Lemma}
\aliascntresetthe{lemma}

\newaliascnt{proposition}{theorem}
\newtheorem{proposition}[proposition]{Proposition}
\aliascntresetthe{proposition}

\newaliascnt{corollary}{theorem}
\newtheorem{corollary}[corollary]{Corollary}
\aliascntresetthe{corollary}

\newaliascnt{definition}{theorem}
\newtheorem{definition}[definition]{Definition}
\aliascntresetthe{definition}

\newaliascnt{assumption}{theorem}
\newtheorem{assumption}[assumption]{Assumption}
\aliascntresetthe{assumption}

\theoremstyle{definition}

\newaliascnt{remark}{theorem}
\newtheorem{remark}[remark]{Remark}
\aliascntresetthe{remark}

\theoremstyle{plain}

\crefname{theorem}{Theorem}{Theorems}
\crefname{lemma}{Lemma}{Lemmas}
\crefname{proposition}{Proposition}{Propositions}
\crefname{corollary}{Corollary}{Corollaries}
\crefname{definition}{Definition}{Definitions}
\crefname{assumption}{Assumption}{Assumptions}
\crefname{remark}{Remark}{Remarks}

\Crefname{theorem}{Theorem}{Theorems}
\Crefname{lemma}{Lemma}{Lemmas}
\Crefname{proposition}{Proposition}{Propositions}
\Crefname{corollary}{Corollary}{Corollaries}
\Crefname{definition}{Definition}{Definitions}
\Crefname{assumption}{Assumption}{Assumptions}
\Crefname{remark}{Remark}{Remarks}
\crefname{equation}{eq.}{eqs.}
\Crefname{equation}{Eq.}{Eqs.}

\usepackage{lastpage}
\jmlrheading{27}{2026}{1-\pageref{LastPage}}{7/26}{}{26-0000}{Dayi Li}

\ShortHeadings{Sampling-Based Adaptive Active Learning}{Li}
\firstpageno{1}

\newcommand{\para}{\bm \theta}
\newcommand{\dataseq}{\mathbf X}

\newcommand{\eps}{\varepsilon}

\newcommand{\dee}{\mathrm{d}}

\newcommand{\TV}{\mathrm{TV}}

\DeclareMathOperator{\EE}{\mathbb{E}}
\DeclareMathOperator{\PP}{\mathbb{P}}

\DeclareMathOperator{\argmax}{\arg\max}

\newcommand{\Reals}{\mathbb{R}}
\newcommand{\Nats}{\mathbb{N}}

\def\[#1\]{\begin{equation}\begin{aligned}#1\end{aligned}\end{equation}}
\def\*[#1\]{\begin{align*}#1\end{align*}}

\begin{document}

\title{Robust Surrogate-Based Bayesian Inference via Sampling-Based Adaptive Active Learning (SALE)}

\author{\name Dayi Li \email dayi.li@mail.utoronto.ca \\
       \addr Department of Statistical Sciences,\\
       David A. Dunlap Department of Astronomy \& Astrophysics \\
        University of Toronto\\
        700 University Avenue, 9th Floor\\
        Toronto, ON M5G 1Z5, Canada}

\editor{}

\maketitle

\begin{abstract}
Bayesian inference is difficult when likelihood evaluations are expensive and
budgets are limited. We propose sampling-based adaptive active learning (SALE),
a Gaussian-process (GP) framework for surrogate-based Bayesian inference. SALE
uses the expected posterior (EP) induced by normalised GP sample paths as a
common sequential-design measure: it defines a posterior-guided search region
and weights uncertainty reduction (UR). A state-dependent rule allocates
evaluations between Bayesian optimisation (BO) for localisation and UR for
calibration. For BO, an annealed objective interpolates between the EP and
Thompson sampling while regularising the query law against surrogate-path
perturbations. For UR, we introduce an ideal EP-weighted rule and a
computationally feasible proxy. Under a Bayesian GP framework, we characterise
the annealed objective's stability--bias trade-off through perturbation and
Bayesian regret bounds, derive explicit budget-dependent expected
total-variation control for the ideal EP-weighted UR, and establish an expected
total-variation rate for the implemented proxy. Across analytic benchmarks and
simulated likelihoods, SALE reduces total-variation error across all considered settings while avoiding severe failures seen under several external baselines. Econometric and astrophysical examples demonstrate its practical value.
\end{abstract}


\section{Introduction}
\label{sec:intro}

Standard Bayesian inference methods, such as Markov chain Monte Carlo (MCMC), become difficult when the likelihood or log-posterior evaluation is expensive. This occurs when evaluating the likelihood requires a simulator
\citep{pellejeroibanez2020cosmological,zeghal2025sbiwl,lovell2025learning},
a parameter-dependent normalising constant in the likelihood
\citep{park2018bayesian,park2020function,matsubara2024generalized}, repeated
integration over latent structure \citep{bivand2014approximate,gomez2018markov},
or the numerical solution of an inverse problem
\citep{lan2023spatiotemporal,bingham2024inverse,kim2025enhancing}. Standard
MCMC may require many more evaluations than the available
budget permits.

Let \(f(\para)\), \(\para\in\Omega\subseteq\Reals^d\), denote the expensive
log-unnormalised posterior. Surrogate-based Bayesian inference learns a cheap
approximation to \(f\) from a limited evaluation budget, commonly using a
Gaussian process \citep[GP;][]{rasmussen2006gaussian} as the surrogate model. Active learning \citep[AL;][]{settles2012active} then uses the evolving GP posterior to sequentially select where \(f\) is evaluated. Under this framework, inference has been pursued through direct posterior estimation \citep{gutmann2016bolfi,kandasamy2017query,wang2018adaptive,
elgammal2023gpry,li2026boss}, variational posterior approximation
\citep{acerbi2018vbmc}, and GP-emulated MCMC
\citep{jarvenpaa2024approximate}.

In practice, computational budgets typically permit only one sequential design run, with little opportunity for restarts or problem-specific tuning. The key is that accurate posterior inference generally does not require uniform convergence of the surrogate on $\Omega$, instead it should only evaluate $f$ at inference-relevant regions. Robust finite-budget inference therefore needs to solve two learning problems. \emph{Localisation} identifies regions carrying appreciable posterior mass. \emph{Calibration} ensures approximation accuracy throughout those regions. 

Within direct posterior estimation, one line of work uses Bayesian optimisation \citep[BO;][]{shahriari2016review} to locate high log-posterior values and the posterior mode \citep{gutmann2016bolfi,kim2025enhancing,li2026boss}. These methods can be effective when posterior mass is concentrated around a dominant mode. However, they typically inherit generic BO strategies that treat $f$ as a black-box objective, so they do not exploit the log-density structure of $f$ and may explore inference-irrelevant regions. Standard optimisation guarantees also do not control the resulting surrogate posterior approximation.

A second line of work uses uncertainty reduction (UR) for calibration.
Posterior-aware UR criteria either weight GP uncertainty by the current
surrogate posterior
\citep{kandasamy2017query,wang2018adaptive,jarvenpaa2019efficient,
jarvenpaa2021parallel,llorente2020adaptive,surer2024sequential,
villani2024posterior,lartaud2024sequential}
or restrict UR to an estimated high-posterior region
\citep{picheny2010adaptive,bect2019supermartingale,lartaud2024sequential}. However, these strategies provide no common uncertainty-propagated sequential design measure for coordinating localisation and calibration. Tension also exists between theory and finite-budget practice. Existing budget-dependent analyses typically rely on space-filling or prescribed random design \citep{llorente2020adaptive,helin2024introduction,kim2025enhancing}.
Since coverage is controlled independently of the evolving surrogate,
convergence rates are tractable, but empirical performance may be suboptimal. Posterior-weighted UR is generally more efficient, but its decision depends on the evolving surrogate, creating a feedback loop that is difficult to control theoretically. To our knowledge, no previous analysis gives an explicit
budget-dependent bound on the approximation error of the final surrogate
posterior under a sequential posterior-weighted UR design.

The exclusive use of BO or UR causes a further limitation. Before localisation, posterior-weighted UR can be misled by a highly uncertain weighting measure and spend evaluations in regions with negligible true posterior mass. BO is therefore needed for localisation. After localisation, UR is needed to ensure calibration in other inference-relevant regions. A reliable method must adapt its allocation between BO and UR as the surrogate evolves.

We propose \emph{sampling-based adaptive active learning} (SALE), an AL
framework for GP-surrogate posterior approximation. At iteration \(t\), each GP posterior sample path induces a posterior law on \(\Omega\), and averaging these laws gives the \emph{expected posterior} (EP), denoted by \(\bar\pi_t\). The EP has previously been studied as an uncertainty-propagated posterior estimator and, under the Bayesian GP
setup, minimises conditional expected Kullback--Leibler (KL) loss
\citep{reiser2025uncertainty,roberts2026propagating,
roberts2026surrogatebased}. SALE instead uses the EP as a common
sequential design measure for localisation and calibration. To our knowledge, SALE is the first framework to use the EP as a sequential design law.

\subsection{Contributions}

The main contributions of this paper are as follows.
\begin{enumerate}[itemsep=2pt,topsep=2pt]
\item
We use the EP as a sequential design measure for AL. Samples from \(\bar\pi_t\) define a posterior-guided search region for BO and supply the candidate set and weighting measure for UR. EP therefore guides both localisation and calibration.

\item
We introduce a state-dependent BO--UR allocation rule based on local design
resolution around the incumbent maximiser. It favours BO while resolution is
low and shifts toward UR as it improves. For the allocation rule, we show that UR receives a positive asymptotic fraction of evaluations.

\item
For BO, we introduce the annealed objective (AO), which averages the
\(\tau\)-tempered posterior laws induced by GP sample paths. AO equals
\(\bar\pi_t\) at \(\tau=1\) and approaches Thompson sampling
\citep[TS;][]{russo2018tutorial} as \(\tau\downarrow0\). We show that positive
$\tau$-tempering regularises the AO query law against surrogate-path perturbations,
while the corresponding TV continuity need not hold for TS. We develop
exchange simulated annealing to target AO and derive a Bayesian regret bound
quantifying the associated optimisation penalty of AO.

\item
We define an ideal EP-weighted UR rule and show that its one-step
objective is the expected decrease of a variance functional controlling
the conditional expected total-variation (TV) error between the EP and the
true posterior. This yields an explicit budget-dependent bound on the
expected TV error of the final EP produced from a sequential
posterior-weighted UR design. The ideal rule motivates a practical UR proxy, and we establish an expected-TV rate with the same polynomial budget order as the ideal UR rule, up to additional logarithmic factors.
\end{enumerate}

For empirical assessment, we focus on the robustness of a method, that is, reliable finite-budget posterior accuracy across different problems. We assess it through joint-TV distance reduction, performance variations across runs, and the frequency of severe failures.

We evaluate SALE on analytic benchmarks, simulated likelihoods, and applications in econometrics and astrophysics. Among the methods considered, none is best on every problem in terms of median convergence rate. SALE's main empirical advantage is robustness: it reduces TV across all settings and avoids several severe failures seen under the considered external baselines. The ablation study identifies EP-guided region restriction as the main stabilising component, with state-dependent allocation, AO, and the practical UR proxy providing further gains.

The paper is organised as follows. \Cref{sec:preliminaries} defines the
inferential objects and the standing theoretical framework.
\Cref{sec:related_work} reviews related methods. \Cref{sec:methods} presents
SALE and its main theoretical results. \Cref{sec: empirical} reports the
benchmark and simulation studies, and \Cref{sec: applications} presents the
two applications. \Cref{sec:discussion} concludes. \Cref{app: impl} contains core
implementation and metric details, while \Cref{app: proofs} contains proofs and
auxiliary theoretical results. Extended implementation and reproducibility
details and additional empirical results are provided in the \hyperref[oa:start]{Online Appendix}.

The source code, scripts, data, and result files required to reproduce the numerical results in this paper are archived at \url{https://doi.org/10.5281/zenodo.20319901}. An actively maintained \texttt{R} package is at \url{https://github.com/davidolohowski/SALE}.
\section{Problem Setup and Theoretical Framework}
\label{sec:preliminaries}
\label{sec:setup}

Let \(\dataseq=(\mathbf x_1,\ldots,\mathbf x_n)\in\mathscr X^n\) denote the
fixed observed data. Consider a parametric model
\(
\mathcal P=\{P_{\para}:\para\in\Omega\}
\), where \(\Omega\subset\mathbb R^d\), and suppose the model admits a
likelihood \(\mathcal L(\para\mid\dataseq)\) with respect to a common
dominating measure on \(\mathscr X^n\). Let $\mathscr{B}(\Omega)$ denote the Borel $\sigma$-algebra on $\Omega$. Let \(p_\circ\) be the prior density with respect to the
Lebesgue measure \(\lambda\) on \(\Omega\). Define
\begin{equation*}
    f(\para):=
\log\mathcal L(\para\mid\dataseq)+\log p_\circ(\para), \qquad
\mathcal Z(f)
:=
\int_\Omega \exp\{f(\bm u)\}\,\lambda(\dee\bm u).
\end{equation*}
We refer to $f$ as the \emph{objective function}. The posterior law and
its Lebesgue density are
\begin{equation*}
    \pi_\dataseq(\dee\para)
=
p_\dataseq(\para)\lambda(\dee\para),
\qquad
p_\dataseq(\para)
:= \exp\{f(\para)\}/\mathcal Z(f).
\end{equation*}

We consider cases where \(f(\para)\) can be evaluated, possibly with
noise, but each evaluation is expensive. The goal is to construct a cheap surrogate of $f$ that can be used for subsequent inference by evaluating $f$ at only a small set of design points.

For probability laws \(\pi\) and \(\rho\) on \(\Omega\), write \(
{\mathrm{TV}}(\pi, \rho)
:=
\sup_{A\in\mathscr{B}(\Omega)}|\pi(A)-\rho(A)|\). Additionally, write $[n] = \{1, \dots, n\}$ for $n \in \Nats^+$ and $[n]_0 = \{0, 1, \dots, n\}$ for $n \in \Nats$.

\subsection{GP Surrogate and Sequential Design}
\label{subsec: GP}

We place a GP prior on \(f\),
\(
f\sim\mathcal{GP}(0,\kappa),
\)
where \(\kappa:\Omega\times\Omega\to\mathbb R\) is positive definite. Suppose the observations on $f$ are made sequentially, possibly with $n_0\ge 0$ initial observations. After \(t\) active iterations, \(n_t=n_0+t\) and the design is
\(
\mathcal D_t
=
\{(\para_i,y_i)\}_{i=1}^{n_t}\), with
\(y_i=f(\para_i)+\varepsilon_i\),
\(\varepsilon_i\stackrel{\mathrm{iid}}{\sim}\mathcal N(0,\varsigma^2),
\)
and \(\varsigma^2\ge0\). Write
\(
\bm y_t=(y_1,\ldots,y_{n_t})^\top\),
\(\bm K_t=\{\kappa(\para_i,\para_j)\}_{i,j=1}^{n_t}\),
\(\bm k_t(\para)
=
[\kappa(\para,\para_1),\ldots,\kappa(\para,\para_{n_t})]^\top.
\)
The GP posterior given \(\mathcal D_t\) has mean, variance, and covariance
\*[
\mu_t(\para)
&=
\bm k_t(\para)^\top
(\bm K_t+\varsigma^2 \bm I_{n_t})^{-1}\bm y_t,
\\
\sigma_t^2(\para)
&=
\kappa(\para,\para)
-
\bm k_t(\para)^\top
(\bm K_t+\varsigma^2 \bm I_{n_t})^{-1}\bm k_t(\para),
\\
\kappa_t(\para,\para')
&=
\kappa(\para,\para')
-
\bm k_t(\para)^\top
(\bm K_t+\varsigma^2 \bm I_{n_t})^{-1}\bm k_t(\para').
\]

In this paper, we use squared-exponential (SE) kernels with
coordinate-specific length scales, yielding automatic relevance determination
\citep{rasmussen2006gaussian}. The method also applies to sufficiently smooth
Mat\'ern kernels.

GP-based AL uses the current GP posterior to select the next
query \(\para_{n_t+1}\) through an acquisition rule
\citep{settles2012active}. After observing \(y_{n_t+1}\), \((\para_{n_t+1},y_{n_t+1})\) is added to \(\mathcal D_t\), the GP posterior is updated, and the process repeats until the evaluation budget is exhausted. For concise theoretical notation, after conditioning on \(\mathcal D_0\) we relabel the active query \(\para_{n_t+1}\) as \(\para_{t+1}\); algorithmic
statements retain the original data indexing.

Let
\(
\mathcal H_t:=\sigma(\mathcal D_t)
\)
be the information available before the next query is generated. Conditional
on \(\mathcal H_t\), write
\(
\PP_t(\dee g):=\PP(\dee g\mid\mathcal H_t)
\)
for the GP posterior law and let \(f_t\sim\PP_t\) denote a sample
path drawn from \(\PP_t\).

\subsection{Standing GP Assumptions}
We use the following common framework for theoretical analyses. Assumptions and conditions pertaining only to particular components of the analyses are stated in \Cref{app: proofs}.

\begin{assumption}
\label{ass:main_gp_setup}
Unless stated otherwise, the following conditions hold.
\begin{enumerate}[label=(A\arabic*),leftmargin=*,itemsep=1pt]
\item
The domain \(\Omega\subset\mathbb R^d\) is
compact and convex and has nonempty interior. \(d\) is fixed. There are constants
\(c_\Omega,r_\Omega>0\) such that, with
\(
\mathbb B(\para,r)=\{\bm u:\|\bm u-\para\|<r\}
\),
\begin{equation*}
\lambda\{\mathbb B(\para,r)\cap\Omega\}
\ge c_\Omega r^d,
\quad
\para\in\Omega,\quad 0<r\le r_\Omega.
\end{equation*}

\item
$f$ is generated from the GP prior \(\mathcal{GP}(0,\kappa)\) used by the algorithm, and 
\begin{equation*}
y_i=f(\para_i)+\varepsilon_i,
\quad
\varepsilon_i\stackrel{\mathrm{iid}}{\sim}\mathcal N(0,\varsigma^2),
\end{equation*}
where the $\varepsilon_i$ are independent of \(f\) and of the algorithmic
randomness. The kernel hyperparameters and \(\varsigma^2\) are fixed, and GP
updating is exact. If \(\varsigma^2=0\), redundant observations are
removed before GP conditioning.

\item
At iteration \(t\), the next query is measurable with respect to
\(\mathcal H_t\) and fresh algorithmic randomness that is independent of
\(f\), $\varepsilon_i$, and earlier algorithmic randomness. Thus the
design is non-anticipating. Moreover, conditional on \(\mathcal H_t\), \(f, f_t \stackrel{\mathrm{iid}}{\sim} \PP_t\).

\item 
For every \(t\), the posterior path \(f_t\) has a unique maximiser over
\(\Omega\), \(\PP_t\)-almost surely.

\item
For an open set \(U\supset\Omega\), \(\kappa\) extends to a covariance
kernel on \(U\times U\), satisfies
\(
B_\kappa := \sup_{\para\in\Omega}\kappa(\para,\para)\in(0,\infty)
\),
and the corresponding prior GP on \(U\) has a version with \(C^1(U)\)
sample paths a.s.
Let
\(
\mathscr d_0^2(\para,\para')
:=
\operatorname{Var}\{f(\para)-f(\para')\}
\)
and, for \(j\in[d]\), let
\(
\mathscr d_{0,j}^2(\para,\para')
:=
\operatorname{Var}\{\partial_jf(\para)-\partial_jf(\para')\}.
\)
Writing \(N(\epsilon,\Omega,\mathscr d)\) for the covering number under a
pseudometric \(\mathscr d\), assume
\begin{align*}
J_0 &:= \int_0^{\operatorname{diam}(\Omega,\mathscr d_0)}
\sqrt{\log N(\epsilon,\Omega,\mathscr d_0)}\,\dee\epsilon
<\infty,
\\
J_0^{(j)} &:= \int_0^{\operatorname{diam}(\Omega,\mathscr d_{0,j})}
\sqrt{\log N(\epsilon,\Omega,\mathscr d_{0,j})}\,\dee\epsilon
<\infty,
\qquad j\in [d],
\end{align*}
and
\(
B_{\partial,j}:= \sup_{\para\in\Omega}\operatorname{Var}\{\partial_jf(\para)\}<\infty
\)
for every \(j\). Moreover, for every \(t\), almost surely, \(\mu_t\in C^1(U)\) and
\(
\partial_j\mu_t(\para)
=
\EE\{\partial_jf(\para)\mid\mathcal H_t\}\), \(\forall\para\in\Omega\), \(j\in[d].
\)
\end{enumerate}
\end{assumption}

\Cref{ass:main_gp_setup} describes the theoretical Bayesian experiment. In
practice, GP hyperparameters are estimated and several posterior
calculations are approximated. The theory may also be read conditional on the initial design \(\mathcal D_0\), with active time \(t\) starting after \(\mathcal D_0\). Standard SE kernels satisfy \Cref{ass:main_gp_setup}, as do Mat\'ern kernels with smoothness \(\nu > 1\).

\subsection{Posterior Laws Induced by the Surrogate}
\label{subsec:surrogate_posterior_laws}

\begin{definition}
\label{def: expected posterior}
For a continuous function \(g:\Omega\to\mathbb R\), define
\begin{equation*}
\pi(\dee\para\mid g)
=
\frac{\exp\{g(\para)\}}{\mathcal Z(g)}\,\lambda(\dee\para),
\quad
\mathcal Z(g)
:=
\int_\Omega\exp\{g(\bm u)\}\,\lambda(\dee\bm u).
\end{equation*}
The posterior law based on the GP posterior mean is
\(
\pi_{\mu_t}(\dee\para):=\pi(\dee\para\mid\mu_t).
\)
The \emph{expected posterior (EP)} is
\begin{equation*}
    \bar\pi_t(A)
:=
\int \pi(A\mid g)\,\PP_t(\dee g),
\quad
A\in\mathscr B(\Omega).
\end{equation*}
Equivalently,
\(
\bar\pi_t(\dee\para)
=
\EE_{f_t\sim\PP_t}\{\pi(\dee\para\mid f_t)\}
\). Its Lebesgue density is
\begin{equation*}
    \bar p_t(\para)
:=
\frac{\dee\bar\pi_t}{\dee\lambda}(\para)
=
\int
\frac{\exp\{g(\para)\}}{\mathcal Z(g)}
\,\PP_t(\dee g)
=
\EE\!\left\{
p_\dataseq(\para)
\,\middle|\,
\mathcal H_t
\right\}.
\end{equation*}
\end{definition}

\(\bar\pi_t\) was studied as an uncertainty-propagated inferential object by \citet{reiser2025uncertainty} and \citet{roberts2026propagating,roberts2026surrogatebased}. Its use in SALE is supported by the following result, which is Proposition~1 of \citet{roberts2026propagating} applied conditionally on \(\mathcal H_t\).

\begin{proposition}
\label{prop:pi1_kl_optimal_main}
Under \Cref{ass:main_gp_setup}, let \(\widetilde\pi_t\) be any
\(\mathcal H_t\)-measurable probability law on \(\Omega\). Then, provided the KL loss is well defined,
\begin{equation*}
\EE\!\left\{
\operatorname{KL}(\pi_\dataseq\|\widetilde\pi_t)
\,\middle|\,
\mathcal H_t
\right\}
\ge
\EE\!\left\{
\operatorname{KL}(\pi_\dataseq\|\bar\pi_t)
\,\middle|\,
\mathcal H_t
\right\}
\quad\text{a.s.}
\end{equation*}
\end{proposition}

Thus, among estimators of $\pi_\dataseq$ given \(\mathcal H_t\), \(\bar\pi_t\)
minimises the conditional expected forward-KL loss. Instead of using \(\bar\pi_t\) only as an inferential output, SALE uses it to guide sequential design, and the main calibration theory concerns \(\bar\pi_t\). To our knowledge, the use of \(\bar\pi_t\) as a common sequential design measure for region restriction, BO, and UR is unexplored. For empirical comparisons, \Cref{sec: empirical} uses \(\pi_{\mu_t}\) as a cheap deterministic summary for all methods.

\(\bar\pi_t\) is also dynamically coherent as a weighting law, as seen in the following proposition.

\begin{proposition}
\label{prop:barpi_martingale_main}
Under \Cref{ass:main_gp_setup},
\(\{\bar\pi_t\}_{t\geq0}\) is a measure-valued martingale. Equivalently, for
every bounded measurable \(h:\Omega\to\mathbb R\),
\begin{equation*}
M_t(h)
:=
\int_\Omega h(\para)\,\bar\pi_t(\dee\para)
=
\EE\!\left\{
\int_\Omega h(\para)\,\pi_\dataseq(\dee\para)
\,\middle|\,
\mathcal H_t
\right\},
\end{equation*}
is a bounded martingale. Thus, there exists a probability law
\(\bar\pi_\infty\) such that
\(
\bar\pi_t\Longrightarrow\bar\pi_\infty
\)
a.s.
\end{proposition}

The proof is given in \Cref{appdx:ur_activation}. Conditional on \(\mathcal H_t\), the expected value of \(\bar\pi_t\) after the next query is its
current value. The realised law can change when a new response is observed,
but it has no conditional drift under the ideal Bayesian model. This
property underlies the one-step $\bar\pi_t$-weighted variance identity in
\Cref{prop:ur_calibration_link_main}. The martingale property does not
assert \(\bar\pi_\infty\overset{d}{=}\pi_\dataseq\); error control of \(\bar\pi_t\) is
provided by \Cref{thm:sale_main_calibration,thm:sale_practical_score_calibration}.
\section{Related Work}
\label{sec:related_work}

This section positions SALE within the literature on surrogate-based Bayesian inference by comparing what existing approaches target, how they select evaluations, and the inferential objects and guarantees they provide.

\subsection{Bayesian Optimisation for Surrogate Inference}
\label{subsec:rw_bo}

Among pointwise BO methods for surrogate inference, \citet{gutmann2016bolfi} model discrepancies in summary statistics for synthetic likelihoods, while \citet{li2026boss} apply a GP-upper confidence bound (GP-UCB) to the log-posterior and use the resulting GP to construct a surrogate posterior. These are useful when the main task is locating a small region of high posterior mass. Standard BO regret results, however, do not by themselves control surrogate-posterior error.

\citet{kim2025enhancing} augment GP-UCB with random exploration from a fixed external measure \(P_0\). Under noiseless evaluations and a reproducing kernel Hilbert space (RKHS) assumption, the random exploration yields expected-Hellinger convergence rates for \(\pi_{\mu_t}\). However, \(P_0\) does not adapt to the evolving surrogate posterior, and its random evaluations need not focus on inference-relevant regions at finite budgets.

A distributional BO formulation is given by \citet{oliveira2021noregret}. Their KL-UCB algorithm draws evaluations from an evolving posterior induced by a UCB surface. Under an RKHS assumption, they establish a finite-budget KL bound and show that the minimum KL error attained among the first \(T\) UCB-induced posteriors converges to zero. However, this does not establish convergence of the final surrogate posterior, and the KL-minimising iterate cannot be identified without access to the true posterior.

\subsection{Posterior-Focused Uncertainty Reduction}
\label{subsec:rw_ur}

\citet{kandasamy2017query} model the log-likelihood with a GP and select
queries using pointwise uncertainty in the exponentiated GP surrogate.
\citet{wang2018adaptive} use a log-GP entropy criterion, while
\citet{jarvenpaa2019efficient,jarvenpaa2021parallel} develop expected-loss
and variance-based rules for noisy likelihoods, including batch designs. These
methods focus evaluations on regions judged important by the current surrogate
posterior and are generally more efficient than reducing uncertainty over the
full domain \(\Omega\). Theoretically, these methods do not establish convergence of the final surrogate posterior.

\citet{llorente2020adaptive} sequentially interpolate the unnormalised posterior and select evaluations using its current interpolant and a diversity term. However, their theoretical control only applies to the pure space-filling specialisation, not the posterior-weighted acquisition.

For Bayesian inverse problems, \citet{villani2024posterior} use the current
surrogate posterior to weight an error measure and greedily refine the surrogate. Their analysis bounds surrogate posterior error through this measure but derives no budget-dependent bound. \citet{lartaud2024sequential} either restrict evaluations to a Mahalanobis region around the estimated posterior mode or use stepwise uncertainty reduction (SUR). They prove almost-sure convergence of a weighted surrogate uncertainty functional to zero under their SUR design. Both inverse-problem methods emulate the forward map under Gaussian observation noise, and do not directly target a generic log-posterior. Related region-restriction and SUR ideas appear in level-set estimation and general sequential design \citep{picheny2010adaptive,bect2019supermartingale}.

\subsection{Variational Surrogate Inference}
\label{subsec:rw_uncertainty_propagation}

Variational Bayesian Monte Carlo
\citep[VBMC;][]{acerbi2018vbmc,acerbi2019exploration} combines variational
inference, GP modelling, and active-sampling Bayesian quadrature to return a
variational posterior and model-evidence estimate. The aims of VBMC are closely related to those of SALE. However, its design is built around a variational and quadrature objective, instead of direct approximation of the posterior surface.

\subsection{GP-Emulated Markov Chain Monte Carlo}
\label{subsec:rw_gp_mcmc}

\citet{jarvenpaa2024approximate} place a GP on the noisy log-likelihood
and use it to approximate the acceptance decisions along a Metropolis--Hastings (MH) path. They request additional log-likelihood evaluations until the GP-based accept/reject decision is sufficiently reliable. Their design is therefore transition-centred and returns an approximate MCMC sample; its analysis controls individual accept/reject decisions and the evaluations required to resolve them, rather than the accuracy of the resulting MCMC sample or the resulting surrogate posterior.

\subsection{Position of SALE}
\label{subsec:rw_positioning}

SALE's main distinction is its adaptive design architecture centred on \(\bar\pi_t\). At iteration \(t\), \(\bar\pi_t\) is the common design measure for BO and UR; a state-dependent rule allocates the next evaluation between the branches; and AO or the practical UR criterion selects the query. AO exploits the log-density structure by averaging normalised tempered laws induced by GP paths, while an ideal \(\bar\pi_t\)-weighted UR targets posterior calibration. For AO, the theory provides perturbation-stability and Bayesian regret bounds. For UR, the theory provides an explicit budget-dependent bound on the expected TV error of the final EP under the ideal rule and an explicit expected-TV rate for the implemented UR score.
\section{Methods and Theory}\label{sec:methods}

We now develop the architecture of SALE, beginning with the EP-guided localisation.

\subsection{EP-Guided Localisation}
\label{subsec:posterior_relevance}

To exploit the information in \(\bar\pi_t\), one only requires a sample from it. SALE therefore relies on an approximate sample
\(
\mathcal X_t=\{\para_s\}_{s=1}^S
\) from \(\bar\pi_t\). We generate this sample using the stochastic surrogate MCMC (ssMCMC) motivated by \citet{roberts2026surrogatebased}. \Cref{alg:ssmcmc} summarises the algorithm; implementation details are in \Cref{app:ssMCMC}.

\begin{algorithm}[t]
\begingroup
\footnotesize
\caption{Stochastic surrogate MCMC (ssMCMC) for approximating
\(\bar\pi_t\)}
\label{alg:ssmcmc}
\begin{algorithmic}[1]
\State \textbf{Input:} Sample size \(S\); number of inner
Metropolis--Hastings (MH) steps \(L\); GP posterior law \(\PP_t\);
proposal kernel \(q(\para'\mid\para)\); initial state \(\para_0\).
\vspace{0.25em}
\For{\(s=1,\dots,S\)}
  \State Draw a GP posterior sample path \(f_t^{(s)}\sim\PP_t\).
  \State Set \(\para_s^{(0)}\gets\para_{s-1}\).
  \For{\(\ell=1,\dots,L\)}
    \State Propose
    \(\para'\sim q\{\cdot\mid\para_s^{(\ell-1)}\}\).
    \State Accept or reject \(\para'\) using the MH ratio with target
    \(\pi(\cdot\mid f_t^{(s)})\).
    \State If accepted, set \(\para_s^{(\ell)}\gets\para'\); otherwise,
    set \(\para_s^{(\ell)}\gets\para_s^{(\ell-1)}\).
  \EndFor
  \State Set \(\para_s\gets\para_s^{(L)}\).
\EndFor
\vspace{0.25em}
\State \textbf{Output:} Sample
\(\mathcal X_t=\{\para_s\}_{s=1}^S\) approximating \(\bar\pi_t\).
\end{algorithmic}
\endgroup
\end{algorithm}

The first use of \(\mathcal X_t\) is to construct a BO
search region \(\Omega_t\). In practice, \(\Omega_t\) is obtained from \(\mathcal X_t\) by covariance adjustment and marginal tail trimming; see \Cref{oa:search_region}. For analysis, let \(\mathscr R\) denote the
population analogue of this construction and define
\(
\Omega_t^\circ:=\mathscr R(\bar\pi_t)\subseteq\Omega.
\)
BO is conducted over \(\Omega_t\) in the implementation and over
\(\Omega_t^\circ\) in analysis.

To reduce computational overhead, \(\mathcal X_t\) and \(\Omega_t\) are
refreshed intermittently according to a dynamic schedule in
\Cref{subsec:integrated_workflow}. Let \(r(t)\leq t\) denote the most recent
refresh time. \(\mathcal X_{r(t)}\) and \(\Omega_{r(t)}\) are the
current working sample and BO search region.

\Cref{prop:localisation_tv_transfer} quantifies how localisation under
\(\bar\pi_t\) transfers to the true posterior. It follows immediately from
the variational characterisation of TV distance.

\begin{proposition}
\label{prop:localisation_tv_transfer}
Let \(\Omega_t^\circ\subseteq\Omega\) be measurable. Suppose that, for some
\(\eta_t,\delta_t\in[0,1]\),
\(
\bar\pi_t(\Omega_t^\circ)\geq 1-\eta_t\) and
\({\mathrm{TV}}(\bar\pi_t,\pi_\dataseq)
\leq\delta_t.
\)
Then
\(
\pi_\dataseq(\Omega_t^\circ)
\geq 1-\eta_t-\delta_t.
\)
\end{proposition}

Thus, a region omitting at most \(\eta_t\) posterior mass under
\(\bar\pi_t\) omits at most \(\eta_t+\delta_t\) under the true posterior.
This makes the connection between localisation and calibration explicit:
the UR branch in \Cref{subsec:ur_calibration} targets the TV discrepancy
\(\delta_t\).

\subsection{State-Dependent BO--UR Allocation}\label{subsec: state}

At each iteration, SALE decides whether the next evaluation is for
localisation or calibration. BO is favoured while local design resolution
around the incumbent is low; as it improves, UR is favoured to calibrate the
surrogate over regions relevant to inference.

We formalise this decision through a latent binary state
\(\mathbf S_t\in\{0,1\}\), where \(\mathbf S_t=0\) denotes unresolved value
discovery and \(\mathbf S_t=1\) denotes resolution. Fix an integer
\(k\geq1\), and suppose that the initial design $\mathcal D_0$ contains at least \(k+1\) distinct
input locations. Define
\(
\para_t^\dagger
\in
\arg\max_{\para_i:\,(\para_i,y_i)\in\mathcal D_t}
\mu_t(\para_i),
\)
using deterministic tie-breaking. For a GP posterior draw
\(f_t\sim\PP_t\), define
\(f_t^+:=f_t(\para_t^\dagger)\). Given an improvement threshold
\(\varepsilon_{\rm imp}>0\), define \(\mathbf S_t=0\) if
\(
\sup_{\para\in\Omega}f_t(\para)
>
f_t^++\varepsilon_{\rm imp},
\)
and set \(\mathbf S_t=1\) otherwise. The ideal probability for conducting BO
at iteration \(t\) is therefore
\begin{equation*}
p_t^\star
:=
\PP(\mathbf S_t=0\mid\mathcal H_t)
=
\PP_t\left\{
\sup_{\para\in\Omega}f_t(\para)
>
f_t^++\varepsilon_{\rm imp}
\right\}.
\end{equation*}

\citet{wilson2024stopping} uses \(p_t^\star\) to terminate a BO algorithm; SALE instead uses the same event to motivate BO--UR allocation. Computing \(p_t^\star\), however, is impractical: as in \citet{wilson2024stopping}, estimating \(p_t^\star\) requires repeatedly optimising GP posterior draws. Moreover, \(\varepsilon_{\rm imp}\) is on the scale of \(f\) and is therefore problem-dependent.

Rather than directly estimate \(p_t^\star\), SALE uses a cheap proxy for \(p_t^\star\). The idea is to monitor local design resolution around \(\para_t^\dagger\), treating further improvement as more likely when the design remains sparse around \(\para_t^\dagger\) and less likely once its neighbourhood is densely sampled. Let
\(
\mathbf H_t
=
\nabla^2\mu_t(\para_t^\dagger)\),
\(\mathbf C_t
=
\mathcal M\{-(\mathbf H_t+\mathbf H_t^\top)/2\},
\)
where \(\mathcal M\) denotes the numerical positive-definite stabilisation map
used in computation. For each
\((\para_i,y_i)\in\mathcal D_t\) with
\(\para_i\neq\para_t^\dagger\), define the curvature-adjusted distance
\(
d_{t,i}
=
\{
(\para_i-\para_t^\dagger)^\top
\mathbf C_t
(\para_i-\para_t^\dagger)
\}^{1/2}.
\)
Let \(\Delta_{t,k}^{\mathrm{loc}}\) be the mean of the \(k\) smallest positive
\(d_{t,i}\). The initial design $\mathcal{D}_0$ ensures that these \(k\) distances
are available at every $t$. We set
\begin{equation}
\label{eq:p_t_proxy}
\widehat p_{t,k}
=
\tanh\left(
C\Delta_{t,k}^{\mathrm{loc}}/a
\right)
\in[0,1],
\end{equation}
with \(C=\log(3)/2\), so that \(\widehat p_{t,k}=1/2\) when
\(\Delta_{t,k}^{\mathrm{loc}}=a\). The default for $k$ is $k=5$. The parameter \(a\) determines how dense the design must be around \(\para_t^\dagger\) before SALE transitions from BO to UR. The curvature adjustment by \(\mathbf C_t\) makes \(\Delta_{t,k}^{\mathrm{loc}}\), and hence \(a\), dimensionless and gives \(a\) a stable interpretation across problems. SALE draws \(B_t\sim\operatorname{Bernoulli}(\widehat p_{t,k})\), conducting BO when \(B_t=1\) and UR when \(B_t=0\).

More detailed construction of \(\widehat p_{t,k}\) is in \Cref{appdx:pt_proxy}. The
connection between \(p_t^\star\) and \(\widehat p_{t,k}\) is recorded in \Cref{oa:allocation_proxy}.

\subsection{Annealed Objective for Optimisation}\label{sec: AF for BO and UR}

The BO branch of SALE uses an annealed version of
\(\bar\pi_t\).

\begin{definition}
\label{def:annealed_objective}
For \(g\in C(\Omega)\) and \(\tau \in (0,1]\), define
\begin{equation*}
    \pi_\tau(\dee\para\mid g)
=
\exp\{g(\para)/\tau\}\lambda(\dee\para)
/
{\mathcal Z_\tau(g)},
\quad
\mathcal Z_\tau(g)
:=
\int_\Omega
\exp\{g(\bm u)/\tau\}\lambda(\dee\bm u).
\end{equation*}
The \emph{annealed objective (AO)} induced by a Borel
probability law \(\mathbb Q\) on \(C(\Omega)\) is
\begin{equation*}
    \mathcal A_\tau(\mathbb Q)(A)
:=
\int_{C(\Omega)}
\pi_\tau(A\mid g)\,\mathbb Q(\dee g),
\quad
A\in\mathscr{B}(\Omega).
\end{equation*}
Given \(\mathcal H_t\) and \(\tau\in(0,1]\), AO at iteration \(t\)
is \(
    \pi_{\tau,t}(\dee\para)
:=
\mathcal A_\tau(\PP_t)(\dee\para)
=
\EE_{f_t\sim\PP_t}
\left\{
\pi_\tau(\dee\para\mid f_t)
\right\}.\)
\end{definition}

At \(\tau=1\), \(\pi_{1,t}=\bar\pi_t\). Decreasing \(\tau\) concentrates each conditional law \(\pi_\tau(\cdot\mid f_t)\) around high values of \(f_t\). Averaging
these laws gives an acquisition distribution that is tilted toward high posterior values. The next active query is drawn as
\(
\para_{n_t+1}\sim\pi_{\tau,t}.
\)

As \(\tau\downarrow0\), AO approaches TS. TS draws
\(f_t\sim\PP_t\) and selects
\(
\para_{f_t}^\star
:=
\argmax_{\para\in\Omega}f_t(\para).
\)
Its selection distribution is therefore
\(
\mathrm{Law}(\para_{f_t}^\star\mid f_t\sim\PP_t).
\)
The following proposition makes this connection precise; the proof is in
\Cref{appdx:proof_opt}.

\begin{proposition}
\label{prop:ao_ts_limit_main}
Under Assumption~\ref{ass:main_gp_setup}, given
\(\mathcal H_t\),
\(
\pi_{\tau,t}
\Longrightarrow
\operatorname{Law}
(
\para_{f_t}^\star
\mid
f_t\sim\PP_t
) \) as \(\tau\downarrow0. \)
\end{proposition}

The purpose of AO is not, however, to approximate TS by taking \(\tau\)
arbitrarily close to zero. Its practical advantage arises at positive
\(\tau\). For fixed \(\tau>0\), the map
\(
g\mapsto\pi_\tau(\cdot\mid g)
\)
is stable under perturbations of \(g\), while the argmax map need not
be.

To see this, equip \(C(\Omega)\) with the supremum norm and its Borel
\(\sigma\)-algebra \(\mathscr B\{C(\Omega)\}\). Consider path laws \(\mathbb Q_1,\mathbb Q_2\) on $C(\Omega)$ with finite first moment in sup-norm:
\begin{equation}
\label{eq: path_law_moment}
    \int_{C(\Omega)}
\|g\|_\infty\,\mathbb Q_j(\dee g)
<\infty,
\quad j=1,2.
\end{equation}
Let \(\Gamma(\mathbb Q_1,\mathbb Q_2)\) be the set
of all probability laws \(\gamma\) on \(C(\Omega)\times C(\Omega)\) with marginals \(\mathbb Q_1\) and \(\mathbb Q_2\), that is, for \(A,B\in\mathscr B\{C(\Omega)\}\), \(
    \gamma\{A\times C(\Omega)\}=\mathbb Q_1(A),\) and
\(\gamma\{C(\Omega)\times B\}=\mathbb Q_2(B).
\) Write
\begin{equation*}
    W_{1,\infty}(\mathbb Q_1,\mathbb Q_2)
:=
\inf_{\gamma\in\Gamma(\mathbb Q_1,\mathbb Q_2)}
\int
\|g_1-g_2\|_\infty\,
\gamma(\dee g_1,\dee g_2).
\end{equation*}
\Cref{ass:main_gp_setup} \textup{(A5)} ensures \Cref{eq: path_law_moment} for \(\PP_t\), almost surely. The following lemma quantifies the stability of the AO law; the proof
is in \Cref{appdx:proof_opt}.

\begin{lemma}
\label{lem:ao_perturb}
Let \(g_1,g_2\in C(\Omega)\), $\varrho(g_1, g_2) := \|g_1-g_2\|_\infty$, and denote \(p_{\tau,j}\) the Lebesgue
density of \(\pi_\tau(\cdot\mid g_j)\), \(j=1,2\). For every
\(\para\in\Omega\) and \(\tau>0\),
\begin{equation*}
    \left|
\log p_{\tau,1}(\para)-\log p_{\tau,2}(\para)
\right|
\leq 2\varrho(g_1, g_2)/\tau,
\quad
\TV\left\{\pi_\tau(\cdot\mid g_1),\pi_\tau(\cdot\mid g_2)\right\}
\leq
\tanh\!\left\{\varrho(g_1, g_2)/\tau
\right\}.
\end{equation*}
Consequently, for any path laws \(\mathbb Q_1,\mathbb Q_2\) satisfying \Cref{eq: path_law_moment},
\begin{equation*}
\TV\left\{
\mathcal A_\tau(\mathbb Q_1),
\mathcal A_\tau(\mathbb Q_2)\right\}
\leq
W_{1,\infty}(\mathbb Q_1,\mathbb Q_2)/\tau.
\end{equation*}
\end{lemma}

In particular, if \(\PP^\prime_t\) is an alternative law to $\PP_t$ satisfying \Cref{eq: path_law_moment} and
\(\pi^\prime_{\tau,t}
:=
\mathcal A_\tau(\PP^\prime_t),
\)
then
\begin{equation}\label{eq: AO_Pt_bound}
\TV(
\pi_{\tau,t},\pi^\prime_{\tau,t})
\leq 
W_{1,\infty}(\PP_t,\PP^\prime_t)/\tau.
\end{equation}

By definition of the TV distance, \Cref{eq: AO_Pt_bound} shows that AO assigns similar
next-query probabilities to a region whenever the surrogate path laws can be coupled
so that their paths are uniformly close on average. Thus, positive $\tau$ stabilises the query distribution even when the maximiser of an individual surrogate path is sensitive to perturbation.

No analogous \(W_{1,\infty}\)-to-TV continuity holds for TS in general. This is particularly relevant when the surrogate is highly uncertain since individual GP posterior paths may contain artificially high, narrow spikes: TS switches to such a spike if it is the pathwise maximiser, while positive-$\tau$ AO weights both its height and local volume. If the spike is narrow in different directions, its local volume shrinks, allowing AO to tolerate a larger spurious height difference before moving its query mass. \Cref{oa:wide_spike} gives an explicit illustration of this mechanism and its possible dependence on dimension.

Positive \(\tau\) therefore trades concentration around the sampled-path maximiser for robustness to surrogate-law perturbations. The associated optimisation cost of AO is quantified through Bayesian regret \citep{russo2014learning} in \Cref{thm:ao_main_text}.

\begin{definition}
\label{def:ao_regret}
Let \(\para^\star\) be a measurable maximiser of \(f\) over \(\Omega\), and
let \(\mathcal D_0=\emptyset\). After \(T\ge1\) evaluations, the simple regret
is
\(
s_T
=
f(\para^\star)
-
\max_{0\le t\le T-1}f(\para_{t+1}).
\)
The cumulative regret is
\(
R_T
:=
\sum_{t=0}^{T-1}
\left\{
f(\para^\star)-f(\para_{t+1})
\right\}.
\)
The Bayesian regret is
\(
\mathfrak R_T:=\EE[R_T],
\)
where the expectation is over the GP prior, observation noise, and
algorithmic randomness.
\end{definition}

\begin{theorem}
\label{thm:ao_main_text}
Under Assumption~\ref{ass:main_gp_setup}, let
\(\mathcal D_0=\emptyset\), \(\varsigma^2\geq0\), and fix
\(\tau\in(0,1]\). Fix any \(\upsilon^2>0\) satisfying
\(\upsilon^2\geq\varsigma^2\). Suppose that
\(
\para_{t+1}\sim\pi_{\tau,t}
\)
over \(\Omega\). Then, for every \(T\geq1\),
\begin{equation}
\label{eq:ao_main_text_generic}
\mathfrak R_T
=
\mathcal O\left(
\sqrt{T\gamma_T(\upsilon)\log T}
\right)
+
\mathcal O\left(\Delta_T^\tau\right),
\end{equation}
where \(\gamma_T(\upsilon)\) is the GP maximal information gain with regularisation $\upsilon^2$
\citep{whitehouse2023sublinear} and
\(
\Delta_T^\tau
:=
T\tau
\left\{
1+\log_+(1/\tau)
\right\}\), where \(\log_+(x) = \max\{\log x, 0\}\).
\end{theorem}

\begin{algorithm}[t]
\begingroup
\footnotesize
\caption{Exchange simulated annealing (ESA) for AO}
\label{alg:ESA}
\begin{algorithmic}[1]

\Require Design \(\mathcal D_t\); search domain \(\Omega\); GP posterior
law \(\PP_t\); function proposal kernel
\(
Q_f(f_t,\dee f_t')
=
q_f(f_t'\mid f_t)\PP_t(\dee f_t');
\)
parameter proposal density \(q_{\para}(\para'\mid\para)\); cooling schedule
\(\{\tau_\ell\}_{\ell=1}^L\) satisfying
\(1=\tau_1>\cdots>\tau_L=\tau>0\); iteration counts
\(L_{\para},L_{\bm u},L_f\); levelwise sweep counts
\(\{s_\ell\}_{\ell=1}^L\).

\State Initialise \(\para\in\Omega\).
\State Draw \(f_t\sim\PP_t\).

\For{\(\ell=1,\dots,L\)}
    \Comment{Temperature \(\tau_\ell\)}
    \For{\(j=1,\dots,s_\ell\)}
        \For{\(j_{\para}=1,\dots,L_{\para}\)}
            \Comment{Update \(\para\) conditional on \(f_t\)}
            \State Propose
            \(\para'\sim q_{\para}(\cdot\mid\para)\).
            \State Set
            \*[
            \alpha_{\para}
            =
            \min\left\{
            1,\,
            \frac{
                \exp\{f_t(\para')/\tau_\ell\}
                q_{\para}(\para\mid\para')
            }{
                \exp\{f_t(\para)/\tau_\ell\}
                q_{\para}(\para'\mid\para)
            }
            \right\}.
            \]
            \State With probability \(\alpha_{\para}\), set
            \(\para\gets\para'\).
        \EndFor

        \For{\(j_f=1,\dots,L_f\)}
            \Comment{Exchange update of \(f_t\)}
            \State Propose \(f_t'\sim Q_f(f_t,\cdot)\).
            \State Approximately draw
            \(
            \bm u\sim\pi_{\tau_\ell}(\cdot\mid f_t')
            \propto\exp\{f_t'(\cdot)/\tau_\ell\}
            \)
            using \(L_{\bm u}\) MH updates.
            \State Set
            \*[
            \alpha_f
            =
            \min\left\{
            1,\,
            \frac{
                q_f(f_t\mid f_t')
                \exp\{f_t'(\para)/\tau_\ell\}
                \exp\{f_t(\bm u)/\tau_\ell\}
            }{
                q_f(f_t'\mid f_t)
                \exp\{f_t(\para)/\tau_\ell\}
                \exp\{f_t'(\bm u)/\tau_\ell\}
            }
            \right\}.
            \]
            \State With probability \(\alpha_f\), set
            \(f_t\gets f_t'\).
        \EndFor
    \EndFor
\EndFor

\State Set
\((\widetilde\para_t,\widetilde f_t)\gets(\para,f_t)\).
\State \Return
\((\widetilde\para_t,\widetilde f_t)\).

\end{algorithmic}
\endgroup
\end{algorithm}

\begin{remark}
The first term in \Cref{eq:ao_main_text_generic} is the standard GP-BO complexity term. When \(\varsigma^2>0\), one may take \(\upsilon^2=\varsigma^2\); when \(\varsigma^2=0\), \(\upsilon^2\) is an arbitrary fixed analysis parameter. For fixed dimension and any fixed \(\upsilon^2>0\), the information-gain orders for SE and Mat\'ern kernels with smoothness \(\nu > 1/2\) established by \citet[Theorem~7 and Corollary~8]{iwazaki2025improvedregret} imply that \(\sqrt{T\gamma_T(\upsilon)\log T}=o(T)\). The second term is the optimisation penalty induced by a fixed \(\tau>0\), and it is \(\mathcal O(T)\). Thus \Cref{thm:ao_main_text} is not a no-regret claim for fixed $\tau$; together with \Cref{lem:ao_perturb}, it quantifies the trade-off between AO law stability and optimisation bias.
\end{remark}

\Cref{thm:ao_main_text} concerns exact AO sampling over the fixed domain \(\Omega\).
In SALE, AO is instead restricted to the current BO search region
\(\Omega_{r(t)}\). Unless otherwise stated, we continue to write
\(\pi_{\tau,t}\) for this restricted acquisition law. 

Sampling from \(\pi_{\tau,t}\) is difficult when \(\tau\) is small because
the target can be highly concentrated and multimodal. We thus introduce
\emph{exchange simulated annealing} (ESA; \Cref{alg:ESA}), based on simulated annealing
\citep{delahaye2019simulated}. ESA follows a temperature schedule
\(
1=\tau_1>\cdots>\tau_L=\tau>0
\)
and carries its state between successive temperatures.

At a fixed \(\tau_\ell\), the ideal ESA kernel targets the joint
law
\(
\Pi_{\tau_\ell,t}(\dee\para,\dee f_t)
=
\pi_{\tau_\ell}(\dee\para\mid f_t)\PP_t(\dee f_t),
\)
whose \(\para\)-marginal is \(\pi_{\tau_\ell,t}\). ESA alternates between an
MH update of \(\para\) conditional on \(f_t\) and an exchange update of
\(f_t\). The exchange step introduces an auxiliary draw to cancel the
intractable normalising constant \(\mathcal Z_{\tau_\ell}(f_t)\)
\citep{murray2006mcmc}. With an exact auxiliary draw, the exchange kernel leaves \(\Pi_{\tau_\ell,t}\) invariant. The implemented algorithm approximates the auxiliary draw using MH updates and represents GP paths using the random Fourier features (RFF) and Matheron's representation \citep{rahimi2007random, wilson2021pathwise}. Core implementation details are in \Cref{app:esa}; the proposal kernels, temperature schedule, and numerical settings are in \Cref{oa:esa_settings}. \Cref{oa:rff_esa_convergence} analyses the levelwise convergence for the RFF--Matheron representation.

\subsection{EP-Weighted Calibration}\label{subsec:ur_calibration}

The UR branch reduces GP uncertainty in regions that are plausible under $\bar\pi_t$.
\begin{definition}
\label{def:weighted_variance_reduction}
Define the $\bar\pi_t$-weighted integrated mean squared prediction
error by
\(
V_t
:= 
\EE_{\bm u \sim \bar\pi_{t}}\{\sigma_{t}^2(\bm u)\}.
\) The one-step $\bar\pi_t$-weighted integrated variance-reduction is
\begin{equation*}
    \Delta_t(\para)
:=
\int_\Omega
\frac{\kappa_t(\bm u,\para)^2}
{\sigma_t^2(\para)+\varsigma^2}
\bar\pi_t(\dee\bm u), \quad \para\in\Omega.
\end{equation*}
When both
\(\varsigma^2=0\) and \(\sigma_t^2(\para)=0\), the integrand is defined to
be zero.
\end{definition}
\begin{proposition}
\label{prop:ur_calibration_link_main}
Under Assumption~\ref{ass:main_gp_setup}, let
\(V_{t+1}^{(\para)}\) denote \(V_{t+1}\) obtained when the
next query is fixed at \(\para \in \Omega \) without observing its response.
Then
\(
\EE\{
V_{t+1}^{(\para)}
\mid
\mathcal H_t
\}
=
V_t-\Delta_t(\para).
\)
Moreover,
\begin{equation*}
    \EE\!\left\{
{\mathrm{TV}}(\bar\pi_t ,\pi_\dataseq)
\,\middle|\,
\mathcal H_t
\right\}
\leq
\left(V_t/2\right)^{1/2},
\quad\text{and}\quad
\EE\!\left\{
{\mathrm{TV}}(\bar\pi_t ,\pi_\dataseq)
\right\}
\leq
\left\{\EE[V_t]/2\right\}^{1/2}.
\end{equation*}
\end{proposition}

The first identity remains exact even though the weighting law changes from
\(\bar\pi_t\) to \(\bar\pi_{t+1}\): this follows from the martingale property
in \Cref{prop:barpi_martingale_main} and the response-independent GP variance
update. The two claims are proved in \Cref{appdx:ideal_sale_delta_calibration}. 

The ideal UR rule selects
\(
\para_{t+1}
=
\arg\max_{\para\in\Omega}\Delta_t(\para).
\)
Recall \(B_t\) from
\Cref{subsec: state}, and let
\(
N_T^{\rm UR}=\sum_{t=0}^{T-1}(1-B_t)
\)
be the number of UR evaluations in \(T\) iterations. \Cref{thm:sale_main_calibration} achieves finite-budget expected-TV control indexed by realised \(N_T^{\rm UR}\); its proof is in \Cref{appdx:ideal_sale_delta_calibration}.

\begin{theorem}
\label{thm:sale_main_calibration}
Under \Cref{ass:main_gp_setup}, suppose that \(\kappa\) is either a
Mat\'ern-\(\nu\) kernel, \(\nu>0\), or an SE kernel, satisfying,
respectively, condition \textup{(i)} or \textup{(ii)} of
\Cref{ass:prior_spectral_decay}.
The UR branch selects
\(
\para_{t+1}
=
\arg\max_{\para\in\Omega}\Delta_t(\para)
\).
At non-UR iterations, the query may be selected by any rule satisfying
\Cref{ass:main_gp_setup} \textup{(A3)}. For \(T\geq1\) and \(c>0\), define
\begin{align*}
\mathscr O^{+}(T)
&:=
\begin{cases}
(T+1)^{-\nu/(2\nu+d)},
&\text{for a Mat\'ern-\(\nu\) kernel},\\[1mm]
\{1+\log(T+1)\}^{d/2}(T+1)^{-1/2},
&\text{for an SE kernel},
\end{cases}
\\
\mathscr O^{0}_c(T)
&:=
\begin{cases}
(T+1)^{-\nu/d},
&\text{for a Mat\'ern-\(\nu\) kernel},\\[1mm]
\exp\left\{-c(T+1)^{1/(d+1)}\right\},
&\text{for an SE kernel}.
\end{cases}
\end{align*}
Then, for any fixed \(\varsigma^2\geq0\) and every
\(\delta,q\in(0,1)\), there exists \(C_{\delta,q}<\infty\), independent
of \(T\), such that
\begin{equation}
\label{eq:tv_bound_by_ur_count}
    \EE\!\left\{
{\mathrm{TV}}(\bar\pi_T ,\pi_\dataseq)
\right\}
\leq \begin{cases}
C_{\delta,q}\mathscr O^{+}(T)
+
C_{\delta,q}
\PP\!\left(N_T^{\rm UR}<qT\right)^{1/2}
+
\delta,\\
C_{\delta,q}\mathscr O^{0}_{c_{\delta,q}}(T)
+
C_{\delta,q}
\PP\!\left(N_T^{\rm UR}<qT\right)^{1/2}
+
\delta, \qquad \text{if $\varsigma^2 = 0$}.
\end{cases}
\end{equation}
If, in addition, 
\(
\liminf_{T\to\infty}N_T^{\rm UR}/T
>0
\) a.s.,
then, for every \(\delta\in(0,1)\), there exist
\(C_\delta<\infty\), \(c_\delta>0\), and \(T_\delta<\infty\),
independent of \(T\), such that, for every \(T\geq T_\delta\),
\begin{equation}
\label{eq:tv_bound}
\EE\!\left\{
{\mathrm{TV}}(\bar\pi_T,\pi_\dataseq)
\right\}
\leq \begin{cases}
    C_\delta\mathscr O^{+}(T)+\delta, \\
C_\delta\mathscr O^{0}_{c_\delta}(T)+\delta,
\qquad\text{if \(\varsigma^2=0\)}.
\end{cases}
\end{equation}
\end{theorem}

\begin{corollary}
\label{cor:ur_activation_main}
Assume the setup and ideal UR rule of
\Cref{thm:sale_main_calibration}. Suppose additionally that
\Cref{ass:delta_state_proxy_smoothness} holds, that $\mathcal{D}_0$ contains at least \(k+1\) distinct input locations, and that
\(\widehat p_{t,k}\) is computed as in
\Cref{subsec: state}. Then there exists a random variable
\(\rho_\star<1\) such that
\(
\sup_{t\geq0}\widehat p_{t,k}
\leq
\rho_\star\) and
\(\liminf_{T\to\infty}N_T^{\rm UR}/T
\geq
1-\rho_\star
>0
\) a.s.
Consequently, for every \(\delta\in(0,1)\), there exists
\(T_\delta<\infty\) such that \Cref{eq:tv_bound} holds for all
\(T\geq T_\delta\). 
\end{corollary}

\begin{remark}
\label{rem:ideal_ur_calibration_scope}
On bounded domains with the standard boundary regularity,
\citet[Theorem~15]{santin2016approximation} gives
\(\lambda_{0,j}^\lambda\lesssim j^{-(1+2\nu/d)}\) for a
Mat\'ern-\(\nu\) kernel and
\(\lambda_{0,j}^\lambda\lesssim\exp(-c j^{1/d})\) for an SE kernel, where $\lambda_{0,j}^\lambda$ is the eigenvalue of the prior covariance operator in \Cref{ass:prior_spectral_decay}.
Thus \Cref{ass:prior_spectral_decay} \textup{(i)} holds for Mat\'ern-$\nu$ kernels, and \Cref{ass:prior_spectral_decay} \textup{(ii)} holds for SE kernels.
The Mat\'ern smoothness must, in addition, be large enough for
\Cref{ass:main_gp_setup}, that is, $\nu > 1$. The additional \Cref{ass:delta_state_proxy_smoothness} and initial-design conditions in \Cref{cor:ur_activation_main} are used only to verify
\(N_T^{\rm UR}\); they are not needed for
\Cref{eq:tv_bound_by_ur_count}.
If \(N_T^{\rm UR}\equiv T\),
\Cref{eq:tv_bound}
hold for every \(T\geq1\).
\end{remark}

\Cref{prop:ur_calibration_link_main} connects sequential design to
posterior accuracy: \(\Delta_t\) controls reduction in
\(V_t\), and \(V_t\) controls the TV error. The spectral-decay
condition then converts repeated \(\Delta_t\)-maximisation into a
budget-dependent bound on the final EP's expected TV error. To our knowledge,
\Cref{thm:sale_main_calibration} gives the first explicit
budget-dependent bound on the expected TV error of the final
surrogate posterior under sequential posterior-weighted UR. 

The ideal $\Delta_t$ rule, however, is not used directly in SALE. As mentioned in \Cref{sec:intro}, before localisation, $\Delta_t$ inherits the uncertainty in $\bar\pi_t$ and can spend evaluations in inference-irrelevant regions \citep{seo2000gaussian,gramacy2020surrogates}. Moreover, evaluating \(\Delta_t\) requires repeated integration against \(\bar\pi_t\), covariance calculations, and continuous optimisation over \(\Omega\). SALE therefore
uses a state-dependent computational proxy motivated by $\Delta_t$.

Let \(\mathcal X_t=\{\para_s\}_{s=1}^S\) be an approximate sample from
\(\bar\pi_t\). The Monte Carlo version of \(\Delta_t\) is
\begin{equation*}
    \widehat\Delta_t^{\rm MC}(\para)
=
\frac{\sum_{s=1}^S
\kappa_t(\para_s,\para)^2}
{S\{\sigma_t^2(\para)+\varsigma^2\}}.
\end{equation*}

When localisation remains unresolved, SALE tempers $\bar\pi_t$ using \(\mathbf S_t\) from \Cref{subsec: state}. For \(p\in[0,1]\), let
\(
\rho_t^{(p)}(\dee\para)
\propto 
\exp\{p\mu_t(\para)\}\bar\pi_t(\dee\para).
\)
The state-specific laws are
\(
\rho_{t,1}(\dee\para)
:=
\rho_t^{(0)}(\dee\para)
=
\bar\pi_t(\dee\para),
\)
and
\(
\rho_{t,0}(\dee\para)
:=
\rho_t^{(1)}(\dee\para)
\propto
\exp\{\mu_t(\para)\}\bar\pi_t(\dee\para).
\)
Thus, \(\mathbf S_t=1\) gives pure \(\bar\pi_t\)-weighted calibration, while \(\mathbf S_t=0\) tilts the weighting measure toward high GP posterior means. Given \(
p_t^\star
=
\PP(\mathbf S_t=0\mid\mathcal H_t),
\) the Bayes action under posterior expected reverse-KL loss is \(\rho_t^{(p_t^\star)}\); see \Cref{app:reverse_kl_calibration} for derivation.

Replacing \(p_t^\star\) by the allocation proxy \(\widehat p_{t,k}\) gives the tempered Monte Carlo criterion
\begin{equation}
\label{eq:tempered_mc_criterion}
\widehat\Delta_{t,\widehat p_{t,k}}^{\mathrm{MC}}(\para)
=
\frac{
\sum_{s=1}^S
w_{t,s}^{(\widehat p_{t,k})}
\kappa_t(\para_s,\para)^2
}{
\sigma_t^2(\para)+\varsigma^2
}, \quad w_{t,s}^{(\widehat p_{t,k})}
\propto
\exp\{\widehat p_{t,k}\mu_t(\para_s)\}, \quad \sum_{s= 1}^S w_{t,s}^{(\widehat p_{t,k})} = 1.
\end{equation}

\(\widehat\Delta_{t,\widehat p_{t,k}}^{\mathrm{MC}}\) is a first-principles bridge from \(\Delta_t\) to practical implementation. It identifies two ingredients that should be retained in practice: state-tempered posterior weights and GP uncertainty reduction. We convert these ingredients into a cheaper proxy in three steps.

First, UR candidates are restricted to
\(
\para \in \mathcal U_t
:=
\operatorname{unique}(\mathcal X_t).
\)
This removes the need for continuous optimisation, while approximately retaining the $\bar\pi_t$-weighting in $\Delta_t$.

Second, SALE reuses the most recently refreshed sample
\(\mathcal X_{r(t)}\), where
\(
t-r(t)\leq T_{\max},
\) for a fixed \(T_{\max}\), rather than regenerating
\(\mathcal X_t\) at every active time \(t\). 

Third, evaluating all cross-covariance terms in
\Cref{eq:tempered_mc_criterion} remains expensive. We therefore replace the
full covariance calculation by a local ranking proxy. On the log scale, \(w_{t,s}^{(\widehat p_{t,k})}\) contributes \(\widehat p_{t,k}\mu_t(\para)\), while \(\log\sigma_t(\para)\) supplies a local GP-uncertainty ranking. This gives
\begin{equation}
\label{eq: UR_t_main}
\mathrm{UR}_t(\para)
:=
\widehat p_{t,k}\mu_t(\para)+\log\sigma_t(\para).
\end{equation}
The implemented UR rule is therefore
\(
\para_{n_t+1}
=
\arg\max_{\para\in\mathcal U_{r(t)}}
\mathrm{UR}_t(\para)\),
\(\mathcal U_{r(t)}
=
\operatorname{unique}\{\mathcal X_{r(t)}\}.
\)
\(\mathcal U_{r(t)}\) restricts UR to the current $\bar\pi_t$ sample,
\(\log\sigma_t(\para)\) ranks local uncertainty, and
\(\widehat p_{t,k}\mu_t(\para)\) provides a value correction while localisation
remains unresolved.

\Cref{thm:sale_practical_score_calibration} gives a budget-dependent
expected-TV rate for the implemented UR score in \Cref{eq: UR_t_main};
the proof is in \Cref{appdx:fixed_sample_calibration}. Let
\(\mathfrak r_0=0<\mathfrak r_1<\mathfrak r_2<\cdots\) be the refresh
times, so that
\(
r(t)=\max\{\mathfrak r_i:\mathfrak r_i \leq t\}.
\)
At \(\mathfrak r_i\),
\(\mathcal X_{\mathfrak r_i}
=\{\para_{\mathfrak r_i,s}\}_{s=1}^S\) is a sample such that
\(
\para_{\mathfrak r_i,s}\mid\mathcal H_{\mathfrak r_i}
\sim\bar\pi_{\mathfrak r_i}\), \(1\leq s\leq S
\). \(\mathcal X_{\mathfrak r_i}\) is held fixed until the next refresh. Define
\begin{equation*}
    J_t
\in
\arg\max_{1\leq s\leq S}
\left\{
\widehat p_{t,k}\mu_t(\para_{r(t),s})
+
\log\sigma_t(\para_{r(t),s})
\right\},
\end{equation*}
using deterministic tie-breaking and setting \(\log0=-\infty\).
The UR query is
\(\para_{t+1}=\para_{r(t),J_t}\).

\begin{theorem}
\label{thm:sale_practical_score_calibration}
Suppose \Cref{ass:main_gp_setup} holds, except that its fresh-randomness
condition \textup{(A3)} is replaced by the persistent-sample condition in
\Cref{ass:bounded_lag_refresh}, and suppose
\Cref{ass:delta_state_proxy_smoothness} holds. Assume that
\(\mathcal D_0\) contains at least \(k+1\) distinct input locations. Fix \(S\geq1\) and \(\upsilon^2>0\) with
\(\upsilon^2\geq\varsigma^2 \ge 0\). Suppose that
\(
\Pr(B_t=1\mid\mathcal H_t,\mathcal X_{r(t)})
=
\widehat p_{t,k}
\)
and that the UR branch uses the fixed-sample rule above. Then, as
\(T\to\infty\),
\begin{equation*}
\mathbb E\!\left\{
{\mathrm{TV}}(\bar\pi_T,\pi_\dataseq)
\right\}
=
\begin{cases}
\mathcal O\!\left(
T^{-\frac{\nu}{2\nu+d}}
\{\log(T+1)\}^{\frac{4\nu+d}{2(2\nu+d)}}
\right),
& \text{for a Mat\'ern kernel with \(\nu > 1/2\)},\\[1mm]
\mathcal O\!\left(
T^{-1/2}\{\log(T+1)\}^{\frac{d+1}{2}}
\right),
& \text{for an SE kernel}.
\end{cases}
\end{equation*}
\end{theorem}

\begin{remark} 
The Mat\'ern and SE rates follow from the information-gain bounds from \citet[Theorem~7 and Corollary~8; see \Cref{appdx:fixed_sample_calibration}]{iwazaki2025improvedregret}. Relative to the ideal-UR benchmark \(\mathscr O^{+}(T)\) in \Cref{thm:sale_main_calibration}, these rates have the same polynomial powers of \(T\), but incur logarithmic losses. The Mat\'ern rate acquires the displayed logarithmic factor, while the SE rate has one additional factor \(\sqrt{\log(T+1)}\). These losses arise from the maximal-information-gain bound used in the proof. The allocation argument only requires \(\widehat p_{t,k}\) to remain uniformly below one.
\end{remark}
\subsection{Integrated Workflow of SALE}
\label{subsec:integrated_workflow}

\begin{algorithm}[t]
\caption{Sampling-based adaptive active learning (SALE)}
\label{alg:sale}
\begingroup
\footnotesize
\setlength{\abovedisplayskip}{3pt}
\setlength{\belowdisplayskip}{3pt}
\setlength{\abovedisplayshortskip}{2pt}
\setlength{\belowdisplayshortskip}{2pt}
\setlength{\textfloatsep}{5pt}
\setlength{\intextsep}{5pt}

\begin{algorithmic}[1]
\Require Objective \(f\); search domain \(\Omega\); nearest-neighbour count
\(k\geq1\); initial design size \(n_0 \geq k+1\); sequential evaluation budget \(T\geq 1\); ssMCMC sample size \(S\); BO stopping threshold \(p_{\rm BO}\); PSRF tolerance \(\varepsilon_R\); refresh limits \(T_{\min}<T_{\max}\).

\State Generate a distinct initial space-filling design
\(\{\para_i\}_{i=1}^{n_0}\subset\Omega\).
\State Evaluate
\(y_i=f(\para_i)+\varepsilon_i\), \(i\in [n_0]\), and set
\(
\mathcal D_0
=
\{(\para_i,y_i)\}_{i=1}^{n_0}.
\)
\State Estimate the GP hyperparameters and fit the GP posterior \(\PP_0\).
\State Compute the allocation proxy \(\widehat p_{0,k}\) using $k$-NN curvature adjusted distance.
\State Generate \(\mathcal X_0\) using \Cref{alg:ssmcmc}, approximately
targeting \(\bar\pi_0\), and construct \(\Omega_0\).
\State Initialise the refresh interval, set \(r(0)\gets0\), and set
\(\widehat R\gets\infty\).

\For{\(t=0,\dots,T-1\)}
    \State Set the working sample and BO search region to
    \(
    \mathcal X_t^w\gets\mathcal X_{r(t)}\),
    \(\Omega_t^w\gets\Omega_{r(t)}.
    \)
    \State Draw
    \(B_t\sim\operatorname{Bernoulli}(\widehat p_{t,k})\).

    \If{\(B_t=1\)}
        \State Generate \(\para_{n_t+1}\) from the AO target restricted to
        \(\Omega_t^w\) using ESA; see \Cref{alg:ESA}.
    \Else
        \State Set
        \(\mathcal U_t^w\gets\operatorname{unique}(\mathcal X_t^w)\).  Select
        \(
        \para_{n_t+1}
        \in
        \arg\max_{\para\in\mathcal U_t^w}
        \widehat p_{t,k}\mu_t(\para)+\log\sigma_t(\para).
        \)
    \EndIf

    \State Evaluate
    \(y_{n_t+1}=f(\para_{n_t+1})+\varepsilon_{n_t+1}\) and set
    \(
    \mathcal D_{t+1}
    =
    \mathcal D_t
    \cup
    \{(\para_{n_t+1},y_{n_t+1})\}.
    \)
    \State Update the GP posterior \(\PP_{t+1}\) conditional on
    \(\mathcal D_{t+1}\).
    \State Update the refresh interval using the moving calibration
    diagnostic.

    \If{a refresh is triggered}
        \State Set \(r(t+1)\gets t+1\).
        \State Re-estimate the GP hyperparameters and refit
        \(\PP_{t+1}\).
        \State Generate \(\mathcal X_{t+1}\) using
        \Cref{alg:ssmcmc}, approximately targeting
        \(\bar\pi_{t+1}\).
        \State Construct \(\Omega_{t+1}\) from
        \(\mathcal X_{t+1}\).
        \State Recompute the allocation proxy \(\widehat p_{t+1,k}\) using $k$-NN curvature adjusted distance.
        \State Compute the maximum coordinatewise PSRF \(\widehat R\)
        between \(\mathcal X_{t+1}\) and \(\mathcal X_t^w\).

        \If{\(\widehat p_{t+1,k}<p_{\rm BO}\) and
        \(\widehat R<1+\varepsilon_R\)}
            \State \textbf{break}
        \EndIf
    \Else
        \State Set \(r(t+1)\gets r(t)\).
        \State Recompute the allocation proxy \(\widehat p_{t+1,k}\).
    \EndIf
\EndFor

\State \Return The final design, GP posterior, ssMCMC sample, and BO search region.

\end{algorithmic}
\endgroup
\end{algorithm}

\Cref{alg:sale} summarises the implemented SALE workflow. SALE refreshes
\(\mathcal X_t\) and \(\Omega_t\) intermittently, jointly with GP
hyperparameter updates. The refresh schedule is controlled by a moving
diagnostic of one-step-ahead GP predictive calibration. A refresh is allowed after \(T_{\min}\) iterations since the previous refresh and is forced
after \(T_{\max}\) iterations, ensuring a bounded refresh lag.

A practical convergence assessment is conducted at refresh times, when the current and previous ssMCMC samples can be compared. SALE stops when \(\widehat p_{t,k}<p_{\rm BO}\) and the maximum coordinatewise potential scale reduction factor \citep[PSRF;][]{vehtari2021rank}, \(\widehat R\), satisfies \(\widehat R<1+\varepsilon_R\). \(p_{\rm BO}\) and \(\varepsilon_R\) are user-defined. The first condition indicates that further BO is no longer favoured; the second indicates that the $\bar\pi_t$ sample has stabilised. Otherwise, SALE continues until the evaluation budget is exhausted. 

GP hyperparameter fitting, the predictive-calibration refresh diagnostic, the stopping-rule computation, and all numerical constants are specified in \Cref{oa:sale_workflow_settings}.

\section{Numerical Experiments}\label{sec: empirical}

We use analytic benchmarks and simulated-likelihood examples as controlled
settings for evaluating finite-budget surrogate posterior approximation. The
examples cover several common posterior geometries and admit reference
posteriors.

As mentioned in \Cref{sec:intro}, we assess robustness by TV reduction across problems, run-to-run variability, and the frequency of severe failures under common settings and fixed evaluation budgets. For all assessments in this section, we exclude the initial design $\mathcal{D}_0$.

\subsection{Analytic Benchmarks}\label{subsec: simulation}

We consider four analytic log-unnormalised posterior targets on
\(
\para:=(\theta_1,\dots,\theta_d)\in\Omega=[-5,5]^d\),
\(d\in\{4,6,8\},
\)
with constants omitted. For the Quartic-tail Gaussian and Bimodal Mixture targets, define
\(
\sigma_j
=
0.30\times6^{(j-1)/(d-1)}\),
\(\mathbf D_\sigma
=
\operatorname{diag}(\sigma_1,\dots,\sigma_d)\), 
\([\mathbf R]_{ij}
=
0.5^{|i-j|}.
\)

\begin{itemize}[itemsep=2pt,topsep=2pt]

\item \emph{Quartic-tail Gaussian.}
Let \(\bm z=\mathbf D_\sigma^{-1}\para\) and
\(
f(\para)
=
-\frac12\bm z^\top\mathbf R^{-1}\bm z
-\lambda_{\rm tail}\sum_{j=1}^d s_j(\bm z)^2,
\)
where
\(
s_j(\bm z)
=
\tau_{\rm tail}
\operatorname{softplus}
\{(z_j^2-r^2)/\tau_{\rm tail}\}\),
\(r=2\),
\(
\tau_{\rm tail}=0.35\),
\(\lambda_{\rm tail}=0.10,
\)
and
\(
\operatorname{softplus}(x)=\log\{1+\exp(x)\}.
\)
The target is approximately Gaussian if the scaled coordinates remain
within \(|z_j|\lesssim r\). Beyond this region, the \(s_j^2\) terms
produce rapidly decaying quartic tails. This mimics problems that have scale-separated likelihoods.

\item \emph{Neal's Funnel.}
The target is
\(
f(\para)
=
-\theta_1^2/2
-(d-1)\theta_1/2
-\exp(-\theta_1)
\sum_{j=2}^d\theta_j^2/2.
\)
This is Neal's funnel truncated to \(\Omega\), with \(\theta_1\) controlling
the conditional scale of \(\theta_2,\dots,\theta_d\).

\item \emph{Rosenbrock.}
The target is
\(
f(\para)
=
-\sum_{j=1}^{d-1}
\{(\theta_j-1)^2/100
+
(\theta_j^2-\theta_{j+1})^2
\}.
\)

\item \emph{Bimodal Mixture.}
The target is
\(
f(\para)
=
\log\left[
\exp\{\ell_1(\para)\}
+
\exp\{\ell_2(\para)\}
\right],
\)
where \(\ell_1\) is the Quartic-tail Gaussian log target above and
\(
\ell_2(\para)
=
\eta_d
-\bm z_1^\top \mathbf R^{-1}\bm z_1/2\),
\(\bm z_1
=
(0.85\mathbf D_\sigma)^{-1}(\para-\bm\mu_d),
\)
with
\(
\bm\mu_d
=
\sqrt{8/d}\,\bm1_d\),
\(\eta_d
=
-0.5+(3-d)\log(0.85).
\)
This choice gives approximate component masses of \(72\%\) for
\(\exp\{\ell_1\}\) and \(28\%\) for \(\exp\{\ell_2\}\) across the dimensions.
\end{itemize}

\begin{figure}[t]
    \centering
    \includegraphics[width=0.95\linewidth]
    {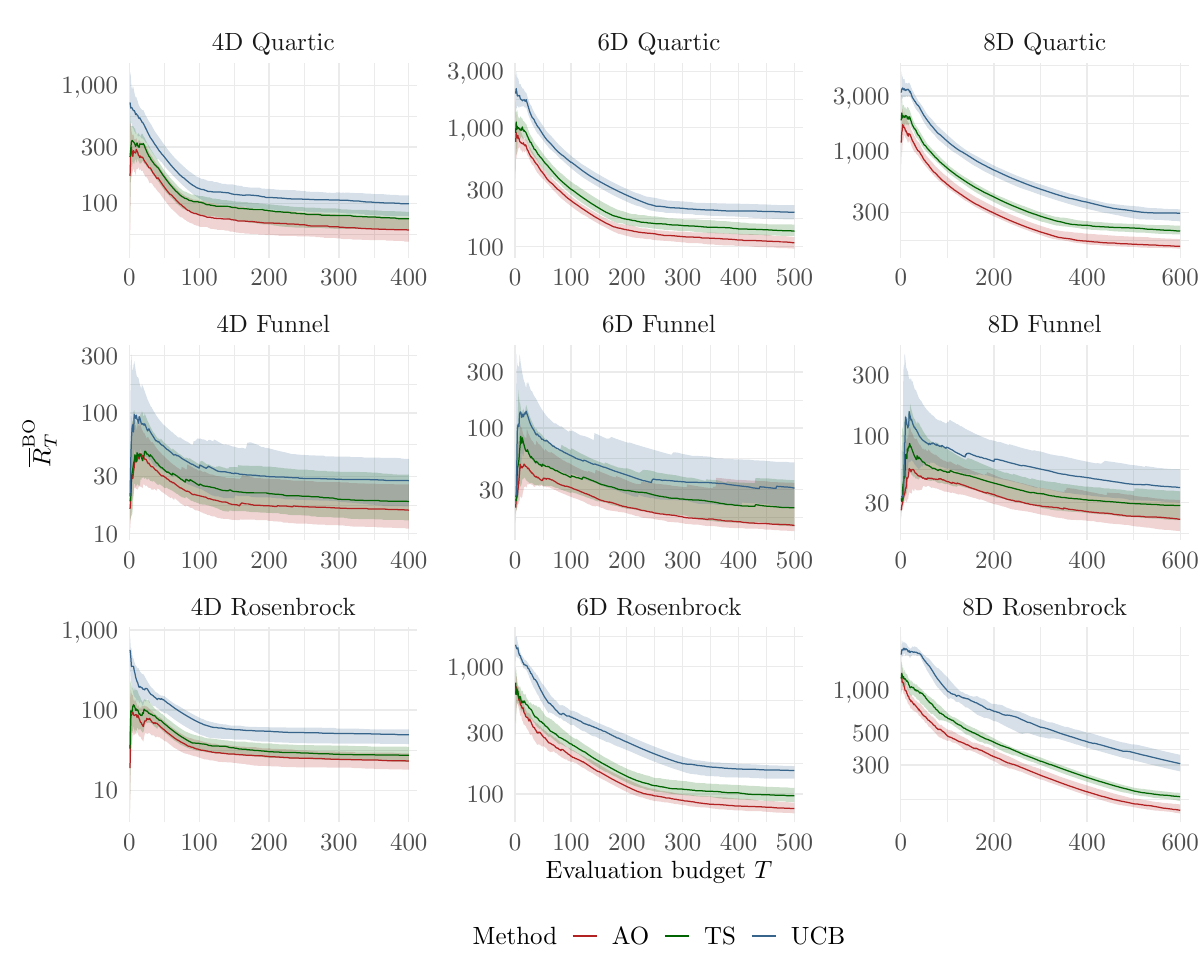}
    \caption{\small Average BO-branch regret
    \(\bar R_T^{\mathrm{BO}}\) for AO, TS, and UCB.
    Curves show medians over \(50\) runs at each evaluation budget;
    shading shows interquartile ranges. Lower is better.}
    \label{fig:average regret}
\end{figure}

\subsubsection{Optimisation Performance}\label{sec:opt performance}

We first compare AO, TS, and UCB as BO strategies within SALE, keeping all other
components and tuning choices fixed. Average BO-branch regret is evaluated on
the three unimodal targets. The Bimodal Mixture is instead evaluated by TV
distance.

Recall \(B_t\) from
\Cref{subsec: state}, and let
\(
N_T^{\rm BO}=\sum_{t=0}^{T-1}B_t
\)
be the number of BO evaluations.
Excluding the initial design, the average BO-branch regret is
\begin{equation*}
    \bar R_T^{\mathrm{BO}}
:=
\frac{1}{N_T^{\mathrm{BO}}}
\sum_{t=0}^{T-1}
B_t
\left\{
f(\para^\star)-f(\para_{t+1})
\right\},
\qquad
N_T^{\mathrm{BO}}>0.
\end{equation*} 
This measures the mean optimisation loss over the BO evaluations by
budget \(T\). Normalising by \(N_T^{\mathrm{BO}}\) removes the
effect of differences in the number of BO evaluations. Computational details
for the three strategies are given in \Cref{oa:implementation_details}. For each target-dimension combination, we use \(50\) independent replicates. Within each dimension and replicate, the same 20-point Latin hypercube initial design is used across methods and targets.

\Cref{fig:average regret} gives the median and the interquartile range (IQR) of \(\bar R_T^{\mathrm{BO}}\). It shows a consistent ordering across the nine target-dimension settings. After the initial part of the trajectories, AO maintains the lowest median average regret, TS is intermediate, and UCB has the highest regret. At the final budget, AO reduces median average regret by \(15\%\)--\(30\%\) relative to TS and by \(40\%\)--\(53\%\) relative to UCB across all settings. The largest reduction relative to TS occurs for the 6D Funnel, while the biggest advantage over UCB is for Rosenbrock. AO also has a lower final upper quartile than both alternatives in every setting, so its advantage extends beyond median performance.

These results indicate that AO produces queries with smaller objective gaps
on average than UCB and TS on the considered targets. Corresponding simple-regret trajectories are reported in \Cref{oa:additional_results}. AO and TS attain near identical simple regret, indicating that both identify competitive high-value points at similar budgets, while AO makes more consistent BO queries along the full trajectory.

\subsubsection{Inference Performance}\label{sec: inf performance}

\begin{figure}[t!]
    \centering
    \includegraphics[width=0.95\linewidth]{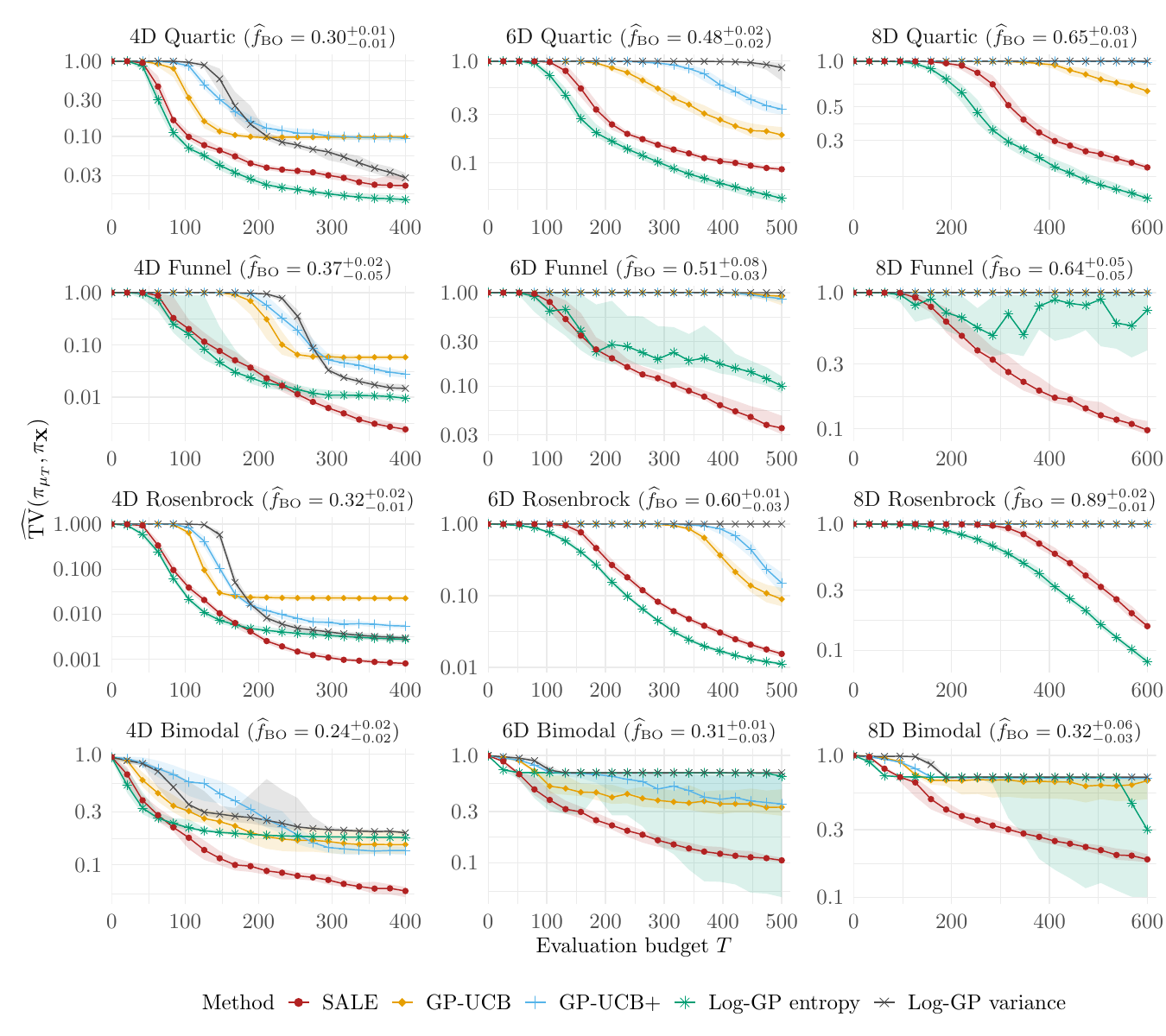}
    \caption{\small TV estimates \(\widehat{\mathrm{TV}}(\pi_{\mu_T},\pi_\dataseq)\) versus budget \(T\) for SALE and external baselines across benchmarks. Curves show medians over \(50\) runs; shading shows IQR. Lower is better. \(\widehat f_{\rm BO}\) gives the median and IQR of BO fractions for SALE over the full budget. For GP-UCB+ on the 6D bimodal target, summaries are computed over 49 runs since one run had numerical failure.}
    \label{fig:TV external}
\end{figure}

We compare SALE with four full-domain GP-AL baselines under a common experimental framework: GP-UCB, following \citet{gutmann2016bolfi, li2026boss};
GP-UCB+, following \citet{kim2025enhancing}; log-GP entropy, following
\citet{wang2018adaptive}; and log-GP variance, following \citet{kandasamy2017query}. We implement these sequential design rules rather than exact reproductions of every paper-specific implementation. These baselines are considered because they permit the most direct and controlled comparisons with SALE to ensure a fair assessment.

All methods use the same GP surrogate class, initial design, GP-fitting schedule, optimisation protocol, and evaluation budget. Computational details are given in \Cref{oa:implementation_details}; internal comparisons of SALE using AO, TS, and UCB are in \Cref{oa:additional_results}. Although the primary computational budget is the number of objective evaluations, we also report representative wall-clock timing audits in \Cref{oa:additional_results}.

At each evaluation budget \(T\), we assess inference quality using
\(
{\mathrm{TV}}(\pi_{\mu_T}, \pi_\dataseq)\). Although the theoretical estimator in
\Cref{thm:sale_main_calibration} is
\(\bar\pi_T\), accurately estimating \({\mathrm{TV}}(\bar\pi_T, \pi_\dataseq)\) would require repeated
normalisation and averaging over many GP sample paths at every
budget and replicate. We therefore use \(\pi_{\mu_T}\) as a cheap
deterministic surrogate-posterior summary for all methods. The estimator
\(\widehat{\mathrm{TV}}(\pi_{\mu_T},\pi_\dataseq)\) for \(
{\mathrm{TV}}(\pi_{\mu_T}, \pi_\dataseq)\) is described in
\Cref{appdx:tv_computation}.

\Cref{fig:TV external} provides the median and IQR of TV trajectories for
all benchmark targets. No method dominates every target in terms of median performance. SALE's main empirical advantage is robustness: it reduces TV consistently across settings and avoids several severe failures seen under the considered external baselines.

On the Quartic target, log-GP entropy has the lowest final median TV in all
three dimensions. On Rosenbrock, SALE is strongest in 4D, while log-GP
entropy is strongest in 6D and 8D. The highly exploitative behaviour of log-GP entropy is effective for these regular unimodal targets. The remaining external
baselines become markedly less effective as the dimension increases.

SALE has the lowest median TV on the Funnel target in every dimension, with
the clearest separation in 6D and 8D. The anisotropic funnel geometry is
difficult to recover under log-GP entropy: the advantage of exploitation in log-GP entropy on the regular targets does not carry over. 

On the Bimodal Mixture, SALE has the lowest median TV in every dimension. Log-GP entropy occasionally reaches low TV in 6D and 8D but has a much wider IQR, indicating substantial run-to-run variation. The other baselines generally plateau at higher TV. These results are consistent with SALE recovering and calibrating modal mass more reliably.

\subsection{Simulated Examples}\label{subsec: simulated}

We consider two simulated examples in which the likelihood evaluation requires nontrivial numerical computation. The first is a spatial Mat\'ern Gaussian random-field (GRF) model:
    \*[
y(s_i)\mid Z,\delta
&\stackrel{\mathrm{ind}}{\sim}
\mathcal N\{Z(s_i),\delta^2\},
\qquad s_i\in[0,1]^2,\quad i\in[n],\\
Z(\cdot)
&\sim
\mathcal{GP}\left\{\mu,\mathcal C(\cdot,\cdot)\right\},\\
\mathcal C(s,s')
&=
\kappa_M(s,s';\sigma_f,\ell,\nu),
\qquad s,s'\in[0,1]^2.
\]
$\kappa_M(s,s'; \sigma_f, \ell, \nu)$ is the stationary Mat\'ern kernel with signal variance $\sigma_f^2$, length-scale $\ell$, and smoothness parameter $\nu$. \(n=200\) noisy observations $\{y_i\}_{i=1}^{200}$ are generated, and we infer $\para = (\mu, \delta, \sigma_f, \ell, \nu) \in \Reals^5$.

The second example is a Lotka--Volterra (LV) inverse problem. Let \(x(t)>0\) and \(y(t)>0\) denote the prey and predator populations, respectively. The observations are generated by
\*[
X_{i} \sim 
\mathcal{N}\left\{\log x(t_i), \sigma^2\right\}, & \quad Y_{i} \sim 
\mathcal{N}\left\{\log y(t_i), \sigma^2\right\}, \\
\frac{\dee x(t)}{\dee t}
=
\alpha x(t)-\beta x(t)y(t),
&\quad
\frac{\dee y(t)}{\dee t}
=
\delta x(t)y(t)-\gamma y(t),
\]
with \(x(0)=x_0\) and \(y(0)=y_0\). Observations are simulated for both species at \(n=20\) time points in \([0.5,20]\). The parameters of interest are $\para = (\alpha,\beta,\gamma,\delta,x_0,y_0,\sigma)\in\Reals^7$.

\begin{figure}[t!]
    \centering
    \includegraphics[width=0.95\linewidth]{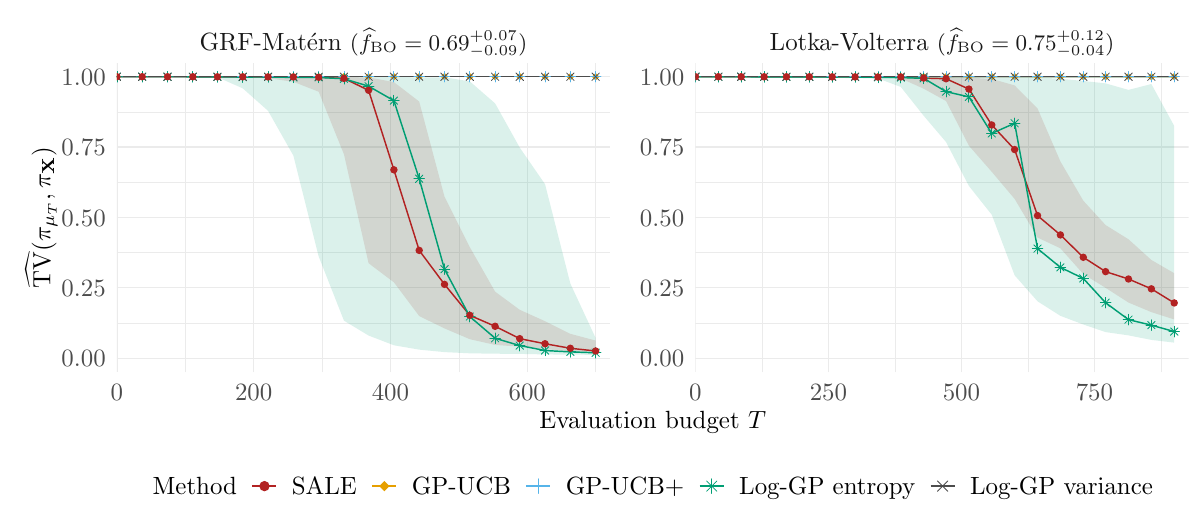}
    \caption{\small TV estimates \(\widehat{\mathrm{TV}}(\pi_{\mu_T},\pi_\dataseq)\) versus budget \(T\) for SALE and external baselines across simulated problems. Curves show medians over \(30\) runs; shading shows IQR. Lower is better. \(\widehat f_{\rm BO}\) gives the median and IQR of BO fractions for SALE over the full budget.}
    \label{fig:simulation}
\end{figure}

Data-generating values, priors, reparameterisations, and search domains are
given in \Cref{oa:implementation_details}. Both examples are kept
sufficiently small to obtain reference posterior samples by MCMC. SALE and
the external baselines use the same experimental protocol as in the analytic
benchmarks, with \(30\) independent replicates per problem and paired initial
designs across methods within each replicate; see \Cref{oa:implementation_details}.

\Cref{fig:simulation} shows that SALE and log-GP entropy are the only
methods with substantial median TV reduction under the available budgets.
Their final median performance is similar on the GRF--Mat\'ern example,
while log-GP entropy has a lower final median on the LV example. SALE,
however, has a considerably narrower IQR in both examples. At the final LV
budget \(T=900\), 8 of the 30 log-GP entropy runs have
\(
\widehat{\mathrm{TV}}(\pi_{\mu_T},\pi_\dataseq)>0.9,
\)
compared with 2 of the 30 SALE runs. Internal comparisons among the SALE BO
strategies are reported in \Cref{oa:additional_results}.

Both examples exhibit strong scale separation: the log-posterior is relatively regular within its main bulk regions but decreases sharply outside them. Log-GP entropy can therefore perform very well once its evaluations concentrate in the appropriate region, as reflected by its competitive lower quartile, but its wide IQR shows that this behaviour is not reliable across initial designs. SALE's narrower IQR is consistent with its adaptive localisation and $\bar\pi_t$-guided restriction of subsequent learning.

\subsection{Ablation Study}

\begin{figure}[t!]
    \centering
    \includegraphics[width=0.85\linewidth]
    {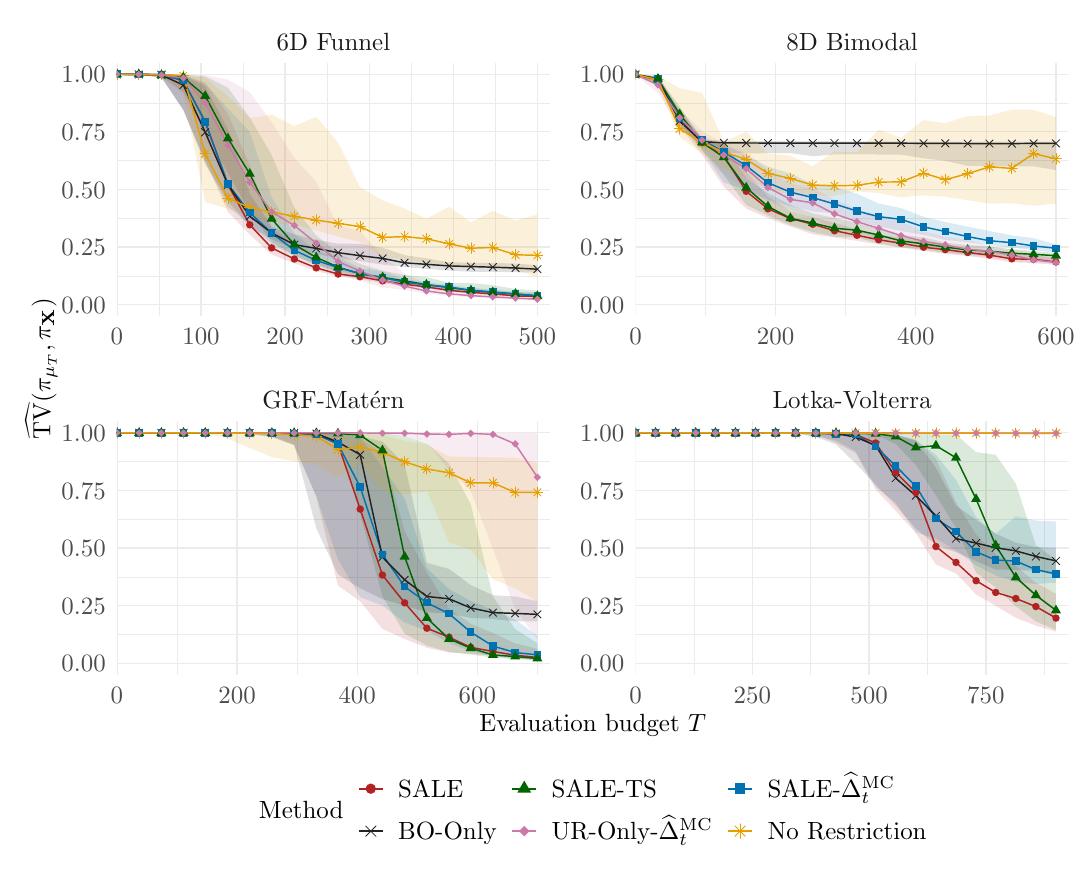}
    \caption{\small TV estimates
    \(\widehat{\mathrm{TV}}(\pi_{\mu_T},\pi_\dataseq)\) versus budget \(T\)
    for SALE and its component ablations. Curves show medians; shading shows
    IQRs. The 6D Funnel and 8D Bimodal targets use \(50\) runs, while the
    GRF--Mat\'ern and LV examples use \(30\) runs. Lower is
    better. For No Restriction on GRF--Mat\'ern, summaries use \(29\) runs
    as one run had a numerical failure.}
    \label{fig:ablation TV}
\end{figure}

We conduct a targeted ablation study on four representative targets: the 6D
Funnel, 8D Bimodal, GRF--Mat\'ern, and LV examples. The study examines the
\(\bar\pi_t\)-guided restriction of BO and UR, BO--UR allocation, AO in the BO
branch, and the practical UR proxy. We also compare the proxy with the Monte
Carlo criterion \(\widehat{\Delta}_t^{\rm MC}\).

The experimental protocol is otherwise unchanged. \Cref{fig:ablation TV}
reports the full TV trajectories. SALE denotes the complete workflow. SALE-TS
replaces AO with TS in the BO branch.
SALE-\(\widehat{\Delta}_t^{\rm MC}\) replaces the practical UR proxy by
\(\widehat{\Delta}_t^{\rm MC}\), with candidates restricted to
\(\mathcal X_{r(t)}\). BO-Only disables the UR branch, whereas
UR-Only-\(\widehat{\Delta}_t^{\rm MC}\) disables the BO branch and applies
\(\widehat{\Delta}_t^{\rm MC}\) over \(\mathcal X_{r(t)}\). No Restriction
keeps the BO search region fixed at \(\Omega\) and uses a uniform UR candidate
set over \(\Omega\).

The largest overall deterioration occurs under No Restriction, which has much
higher final TV on every target and one numerical failure on the GRF--Mat\'ern
example. This ablation jointly removes the \(\bar\pi_t\)-guided BO restriction
and UR candidate construction. It therefore supports the importance of their
combined localisation role.

BO-Only performs considerably worse than SALE on every target, supporting the
need for UR calibration after localisation.
UR-Only-\(\widehat{\Delta}_t^{\rm MC}\) eventually achieves competitive final
TV on the Funnel and Bimodal targets, but improves later than SALE in
\Cref{fig:ablation TV}. On GRF--Mat\'ern, it declines only near the final
budget and remains highly variable, while on LV it does not
improve. On those harder examples, the contrast between
SALE-\(\widehat{\Delta}_t^{\rm MC}\) and
UR-Only-\(\widehat{\Delta}_t^{\rm MC}\) supports the need for a BO-based
localisation phase.

Replacing AO by TS does not uniformly worsen the final median, but SALE-TS has
a larger upper quartile on all four targets, with the largest differences on
the GRF--Mat\'ern and LV examples. Its trajectories in
\Cref{fig:ablation TV} also show later and more variable TV reduction under TS.
This complements the optimisation comparison in \Cref{sec:opt performance}.

Finally, SALE-\(\widehat{\Delta}_t^{\rm MC}\) has a higher median and upper
quartile than SALE on all four targets. Thus the practical proxy performs
better under the tested budgets even though
\(\widehat{\Delta}_t^{\rm MC}\) more directly approximates \(\Delta_t\).
One plausible explanation is that the state-dependent value term
\(\widehat p_{t,k}\mu_t(\para)\) remains useful while \(\bar\pi_t\) is not yet a
reliable weighting measure.

\section{Applications}\label{sec: applications}

We apply SALE to two expensive-likelihood problems from
econometrics and astrophysics. Their computational bottlenecks are distinct:
solving a Bellman fixed point and marginalising a conditional latent Gaussian
model. For each application, we report one SALE run from a 20-point Latin
hypercube initial design using the common default settings.

\subsection{Dynamic Discrete-Choice Model for Bus-Engine Replacement}
\label{sec:app_rust}

Dynamic discrete-choice (DDC) models represent sequential decisions as Markov decision processes and characterise optimal decisions through Bellman equations \citep{Rust1994,AguirregabiriaMira2010}. Under full-solution likelihoods, the Bellman fixed point must be computed via dynamic programming (DP) to evaluate the likelihood, which is expensive.

Existing Bayesian DDC methods reduce the cost per MCMC iteration by reusing or approximating DP solutions \citep{ImaiJainChing2009,Norets2009,Norets2012}. More recent work develops specialised Hamiltonian and reversible-jump MCMC \citep{NoretsShimizu2024}, or accelerates the DP solver itself \citep{Chen2025}. 

We use SALE as a complementary strategy by reducing the number of full likelihood evaluations. Every likelihood value is obtained by solving the original Bellman equations, while approximation is confined to the GP surrogate. The same SALE configuration used in the preceding experiments is applied without modifying the DP solver.

We consider the bus-engine replacement problem of \citet{Rust1987}, in
which a maintenance manager decides each month whether to replace a bus
engine. The data contain \(8{,}156\) monthly transitions from \(104\)
buses, including \(60\) replacement decisions. For bus \(i\), let
\(h\in [m_i]\) index its observed transitions, with
\(\sum_i m_i=8{,}156\). Let
\(x_{ih}\in[K-1]_0\) denote mileage since the most recent
replacement, discretised into 500-mile intervals, and define
\(z_{x_{ih}}=x_{ih}/(K-1)\). We use \(K=800\) mileage states.

We consider a seven-parameter extension of the Rust model,
\(
\para=(R_C,c_1,c_2,c_3,\mu,r,\omega).
\)
Here \(R_C\) is the replacement cost and
\(
C(x)=c_1z_x+c_2z_x^2+c_3z_x^3
\)
is the cost of keeping the engine. Let \(J_{ih}\) denote the number of
mileage-state intervals advanced between consecutive observations. We model
\(J_{ih}\) using a zero-inflated negative-binomial (NB) distribution with
mean \(\mu\), size \(r\), and zero-inflation probability \(\omega\), allowing
both excess unchanged mileage states and overdispersion. The detailed model construction, likelihood function, priors, parameter transformations, and search domains are given in \Cref{oa:implementation_details}.

\begin{figure}[t]
    \centering
    \includegraphics[width=0.85\linewidth]{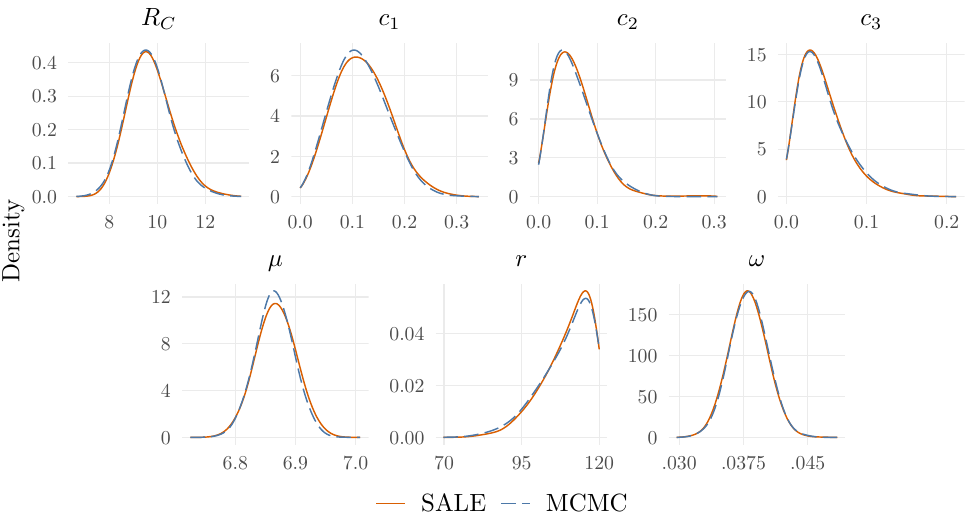}
    \caption{Posterior marginals for the bus-engine replacement model from SALE and reference MCMC.}
    \label{fig:rust density}
\end{figure}

\begin{figure}[t]
    \centering
    \includegraphics[width=0.7\linewidth]{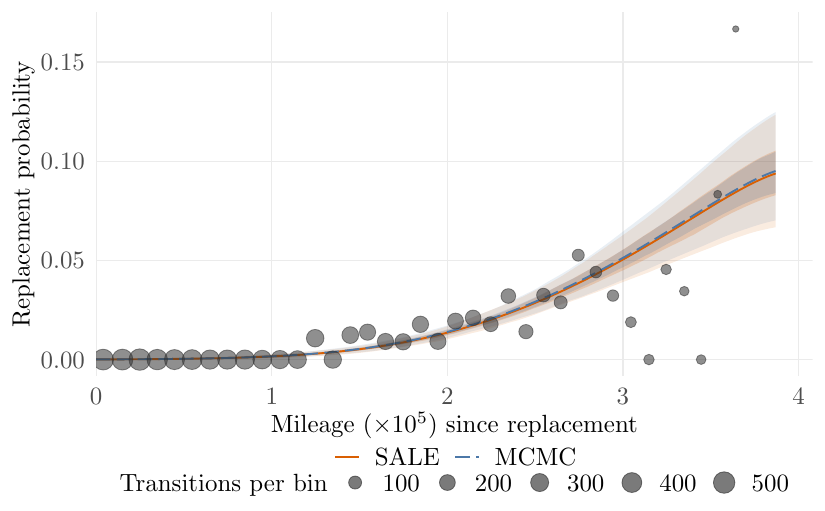}
    \caption{Fitted bus-engine replacement probability versus bus mileage from SALE and reference MCMC. Curves are the posterior means, while the shaded regions are $50\%$ and $90\%$ credible regions.}
    \label{fig:rust_nfxp_policy}
\end{figure}

Evaluating the full Bellman likelihood requires solving an 800-state Bellman
system. One evaluation takes approximately \(0.26\) seconds on a PC with an AMD Ryzen 7 9700X 3.8\,GHz CPU.

We ran SALE for \(500\) sequential evaluations, for \(520\) Bellman solves in total. The complete run took \(15\) minutes. At the final budget, the SALE diagnostics were \(\widehat p_{T,k}=0.13\) and \(\widehat R=1.002\). After discarding the first \(2{,}000\) of \(10{,}000\) final ssMCMC draws, the parameter-wise effective sample sizes (ESS) ranged from \(966\) to \(2{,}826\). The retained ssMCMC draws approximate \(\bar\pi_T\) and are used for the SALE marginals and policy summaries below. To match the numerical experiments, the reported TV diagnostic is instead computed from the plug-in posterior \(\pi_{\mu_T}\).

For reference, we ran a \(40{,}000\)-iteration random-walk Metropolis chain using the full likelihood. An adaptive Metropolis chain initialised from prior information mixed poorly. We therefore initialised the reference chain at the SALE posterior mode and scaled its proposal using the final SALE posterior covariance. The reference run required \(40{,}000\) Bellman solves, took over three hours, and had a parameter-wise ESS range of \(631\) to \(1{,}122\). Thus, beyond producing a surrogate posterior, SALE supplied the localisation and covariance information needed to make direct full-likelihood MCMC workable.

Using the reference MCMC sample, we obtain \(\widehat{\mathrm{TV}}(\pi_{\mu_T},\pi_\dataseq)=0.102,\) with a \(95\%\) moving-block bootstrap interval of
\([0.098,\,0.107]\). The posterior marginals in \Cref{fig:rust density} are almost indistinguishable across all seven parameters. The agreement also extends to the decision-relevant output. \Cref{fig:rust_nfxp_policy} shows nearly identical posterior replacement probabilities over the observed mileage range.

The SALE summary therefore agrees closely with the SALE-initialised
full-likelihood reference chain while using \(520\) rather than \(40{,}000\)
Bellman solves. This is approximately \(77\) times fewer full-likelihood
evaluations for the SALE run. The example illustrates how SALE can wrap an
existing full-solution likelihood without altering the Bellman solver or
introducing DDC-specific approximation machinery into the AL method.

\subsection{Detection of Ultra-Diffuse Galaxies}

Ultra-diffuse galaxies \citep[UDGs;][]{VanDokkum2015} have garnered significant attention in astrophysics as they can help us understand dark matter. However, UDGs are extremely faint and difficult to detect directly in astronomical images \citep{Abraham2014}. A complementary strategy is to detect them through spatial clustering of their star
clusters (SCs), which trace the underlying gravitational potential and remain visible even when the galaxy itself is faint \citep{Li_2022,van_Dokkum_2024,li2025_poisson, li2025CDG2}.

We consider the log-Gaussian Cox process \citep[LGCP;][]{moller1998log} model of \citet{Li_2022}. Let \(\mathbf X\subseteq \mathcal S\subset \Reals^2\) be the observed SC point pattern, where \(\mathcal S\) is the image mask. The model is
\begin{align*}
    \mathbf X &\sim \mathrm{IPP}\left\{\Lambda(s)\right\}, \qquad s \in \mathcal S,\\
\log \Lambda(s)
&=
\beta_0
+
\sum_{j=1}^{N_G}
\beta_j
\exp\left[-\left\{\frac{r(s;c_j)}{R_j}\right\}^{1/n_j}\right]
+
\mathcal U(s),\\
\mathcal U(s)
&\sim
\mathcal{GP}\{0,\kappa_{\mathrm{Mat\acute ern}}(\cdot;\sigma_f,\ell,\nu = 1)\}.
\end{align*}
Here \(\beta_0\) is the log-background SC intensity and the 
summation terms are S\'ersic intensity profiles
\citep[e.g.,][]{li2025mathpop,li2025_poisson} for the \(N_G\) visible
galaxies. For the $j$-th visible galaxy, $c_j$ is the known galactic centre, \(r(s;c_j)\) is the projected distance from \(s\) to \(c_j\), and $(R_j,n_j)$ are the scale and shape parameters for its SC intensity profile. The residual field \(\mathcal U(s)\) captures clustering not explained by the background or visible galaxies. Positive excursions of
\(\mathcal U(s)\) therefore provide the UDG-detection signal.

\begin{figure}[t!]
    \centering
    \includegraphics[width=0.7\linewidth]{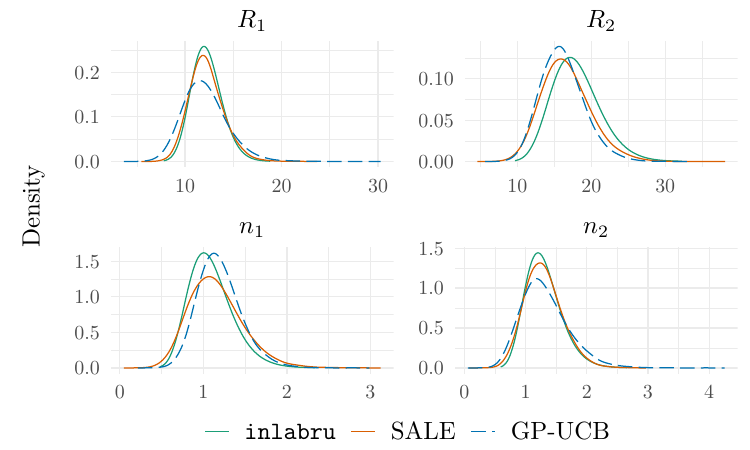}
    \caption{\small Posterior marginals of the nonlinear parameters in the UDG example from \texttt{inlabru}, SALE, and the GP-UCB surrogate method of \citet{li2026boss}.}
    \label{fig:UDG density}
\end{figure}

Let
\(
\para=(R_1,n_1,\ldots,R_{N_G},n_{N_G})
\)
denote the nonlinear parameters. Conditional on \(\para\), the model is a latent Gaussian model (LGM) and can be fitted efficiently by integrated nested Laplace approximation \citep[INLA;][]{inla}. Joint inference for \(\para\) is harder because its nonlinear dependence moves the model outside the standard LGM class required by INLA. The package \texttt{inlabru} handles this dependence through iterative local linearisation \citep{Bachl_2019}, whose accuracy depends on the quality of the linear approximation.

A more direct route conditions on $\para$, performs a full conditional-INLA fit, and uses the resulting marginal log-likelihood $\ell_{\rm INLA}(\para)$ and the corresponding log-unnormalised posterior $f(\para)$ to enable posterior sampling \citep{gomez2018markov, berild2022importance}. This route, however, is expensive: one $\ell_{\rm INLA}(\para)$ evaluation takes approximately \(10\) seconds on the same computer used for the DDC application in \Cref{sec:app_rust}. We use SALE to infer \(\para\) and compare it with the GP-UCB surrogate approach originally proposed for this problem by \citet{li2026boss}.

The data are from a \emph{Hubble Space Telescope} survey of the Perseus
Galaxy Cluster \citep{harris2020piper}. The image contains \(262\)
detected SCs, two visible galaxies, and two confirmed UDGs. Consequently,
\(
\para=(R_1,n_1,R_2,n_2)\in\Reals^4.
\)
The priors follow \citet{Li_2022,li2026boss}, and we restrict the parameters to
\(
R_1\in[2,30]\),
\(R_2\in[2,45]\),
\(n_1,n_2\in[0.05,4]\).

SALE and GP-UCB use the same GP-fitting and experimental settings as in
\Cref{sec: empirical}. Each method receives \(100\) evaluations of
\(f\), including the same \(20\)-point initial Latin hypercube design.
SALE takes \(19\) minutes and GP-UCB takes \(18\) minutes, so the
computational cost of both methods is dominated by $\ell_{\rm INLA}(\para)$
evaluations. The final SALE diagnostics are
\(\widehat p_{T,k}=0.106\) and \(\widehat R=1.0003\).

\Cref{fig:UDG density} compares the posterior marginal densities from SALE, GP-UCB, and \texttt{inlabru}. SALE and \texttt{inlabru} agree closely for several parameters but differ most visibly for $R_2$. GP-UCB shows additional shifts and wider marginals, particularly for $R_1$ and $n_1$.

Since a reference posterior is unavailable, we diagnose the SALE and GP-UCB surrogates at 50 held-out parameter values, with $25$ drawn from each of the final ssMCMC samples. After weighting the points by normalised \(\exp\{f(\para)\}\), the weighted root mean squared error and mean absolute error of the surrogate mean are \(0.15\) and \(0.11\) for SALE, compared with \(0.77\) and \(0.52\) for GP-UCB. Normalising \(\exp\{f(\para)\}\) and \(\exp\{\mu_T(\para)\}\) over the same points gives discrete TV
diagnostics of \(0.05\) for SALE and \(0.22\) for GP-UCB. These diagnostics favour SALE over GP-UCB under the given evaluation budget and validation set.

\begin{figure}[t!]
    \centering
    \includegraphics[width=0.7\linewidth]{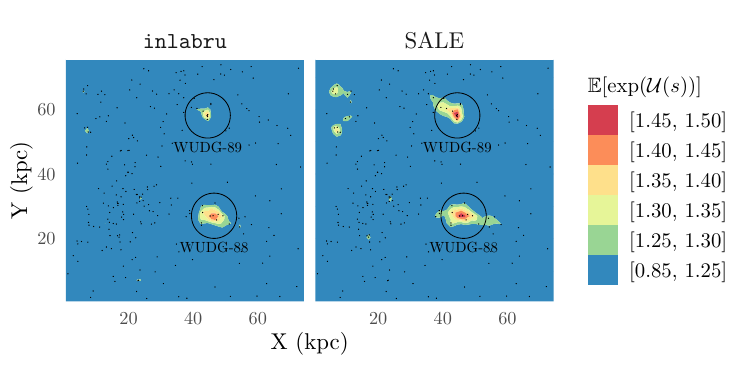}
    \caption{\small Posterior mean of \(\exp\{\mathcal U(s)\}\) from \texttt{inlabru} and SALE for the UDG example. Black points are star clusters; circles mark the two confirmed UDGs, with names shown below.}
    \label{fig:UDG mean US}
\end{figure}

Because \texttt{inlabru} does not estimate the same $f(\para)$ as SALE and GP-UCB, we compare ten repeated $f(\para)$ evaluations at the posterior modes returned by SALE and \texttt{inlabru}. $f(\para)$ at the SALE mode is $0.091 \pm 1.3\times 10^{-6}$ log-units higher than the \texttt{inlabru} mode, suggesting that part of the discrepancy in \Cref{fig:UDG density} is caused by the linearisation in \texttt{inlabru}.

Using the SALE posterior, we propagate uncertainty in \(\para\) to
\(\mathcal U(s)\) using a 17-point adaptive Gauss--Hermite-type quadrature
rule in local eigen-coordinates \citep{aghqtheory}; its construction is
described in \Cref{oa:implementation_details}.

\Cref{fig:UDG mean US} shows a stronger residual-intensity signal under SALE near both confirmed UDGs, particularly WUDG-89. The difference is consistent with the posterior of \(R_2\): SALE places more mass on smaller values than \texttt{inlabru}, leaving more of the nearby SC intensity to the residual field. Thus, the treatment of the nonlinear parameters affects not only their posterior marginals but also the scientifically relevant UDG-detection map. Together with the held-out diagnostics, this example illustrates the practical value of reliable surrogate construction when each likelihood evaluation requires a separate conditional-INLA fit.

\section{Conclusion}\label{sec:discussion}

We introduced SALE, a posterior-aware active learning framework for Bayesian inference with expensive likelihood or posterior evaluations. SALE uses the expected posterior \(\bar\pi_t\) to guide both BO and UR, with a state-dependent allocation rule shifting evaluations from posterior localisation toward calibration. Building on the forward-KL optimality of \(\bar\pi_t\), the theory quantifies AO's stability--optimisation trade-off, gives an explicit budget-dependent bound on the expected TV error of the final EP under the ideal
\(\bar\pi_t\)-weighted UR rule, and establishes an explicit expected-TV rate
for the implemented UR score. Across the benchmarks, simulated examples, and two applications, SALE produced robust posterior approximations under limited evaluation budgets and reduced severe failures relative to the competing methods. These results demonstrate that sequential design with principled exploitation of the structure of posterior distributions can support reliable Bayesian inference when repeated likelihood evaluation is prohibitively expensive.

\acks{The author thanks Alex Stringer for helpful comments and suggestions. This work was partly supported by a University of Toronto Data Sciences Institute Doctoral Fellowship (2023--2025) and CANSSI Multidisciplinary Doctoral Trainee Program. Computational resources were provided by the Digital Research Alliance of Canada. OpenAI's ChatGPT 5.6 was used for proofreading and assistance with checking mathematical proofs, and OpenAI Codex 5.6 was used for code tidying and optimisation and to automate parts of the numerical implementation. All mathematical arguments, references, numerical results, and code were independently verified by the author, who takes full responsibility for the manuscript.}


\makeatletter
\let\jmlrnormalsectionformat\@sict
\makeatother

\crefalias{section}{appendix}
\crefalias{subsection}{subappendix}
\crefalias{subsubsection}{subsubappendix}

\appendix
\numberwithin{equation}{section}

\section{Core Implementation and Metric Details}\label{app: impl}

This appendix records the core implementation details used by SALE and
the TV estimator for the numerical experiments. Extended numerical
settings, problem-specific configurations, and
secondary empirical results are provided in \hyperref[oa:start]{Online Appendix}.

\subsection{Implementation of ssMCMC}\label{app:ssMCMC}

For all examples, ssMCMC uses \(S=10{,}000\) surrogate-path draws and \(L=10\) Metropolis--Hastings (MH) updates per draw. The chain is initialised at the incumbent maximiser. The first \(1{,}000\) outer iterations use a Gaussian random-walk proposal; the remainder use an adaptive Gaussian proposal with covariance as in \citet{roberts2009examples}. Surrogate paths are generated with the RFF--Matheron construction in \Cref{app:esa}.

\subsection{Allocation-Proxy Implementation Conventions}
\label{appdx:pt_proxy}

Using the curvature-adjusted distances \(d_{t,i}\) defined in
\Cref{subsec: state}, let
\(d_{t,(1)}\leq d_{t,(2)}\leq\cdots\) be the ordered positive distances and
let \(\mathcal I_+\) index them. The finite-sample convention is
\begin{equation}
\Delta_{t,k}^{\rm loc}
=
\dfrac{1}{k}\sum_{j=1}^k d_{t,(j)}
\mathbf 1\{|\mathcal I_+|\geq k\}
+
\dfrac{1}{|\mathcal I_+|}
\sum_{i\in\mathcal I_+}d_{t,i}
\mathbf 1\{1\leq|\mathcal I_+|<k\},
\quad
\Delta_{t,k}^{\rm loc}=\infty
\ \text{if }|\mathcal I_+|=0.
\label{eq:Delta_t_proxy}
\end{equation}
The allocation probability \(\widehat p_{t,k}\) is then defined as in
\Cref{eq:p_t_proxy}. We use \(k=5\) and \(a=0.25\) throughout. Larger
\(k\) averages over a broader neighbourhood and generally delays UR, whereas
larger \(a\) promotes an earlier transition. Under the retained-initial-design condition \(n_0\geq k+1\) with distinct input locations, at least \(k\) positive distances are available at every active iteration, so \Cref{eq:Delta_t_proxy} reduces to the mean of the \(k\)
smallest positive distances.

\(\Delta_{t,k}^{\rm loc}\) is a cheap local-spacing diagnostic rather than a
calibrated estimate of \(p_t^\star\). A local-excursion and fill-distance
motivation is given in \Cref{oa:allocation_proxy}.

\subsection{Finite-Dimensional ESA Representation}\label{app:esa}

\Cref{alg:ESA} gives the levelwise target, block updates, and exchange
acceptance ratio. Here we specify the finite-dimensional path representation
and terminal local refinement used in computation. ESA is run on the current
working region \(\Omega_{r(t)}\); as in the main text,
\(\pi_{\tau,t}\) denotes the corresponding restricted AO law.

For a stationary kernel \(\kappa\), draw
\(\omega_j\sim\rho\) from its spectral distribution and
\(b_j\sim\operatorname{Unif}(0,2\pi)\), independently, and define
\citep{rahimi2007random}
\*[
\bm\phi_m(\para)
=
\sqrt{\frac{2\kappa(0)}{m}}
\left[
\cos(\omega_1^\top\para+b_1),\ldots,
\cos(\omega_m^\top\para+b_m)
\right]^\top.
\]
Let
\(
\bm\Phi_{t,m}
=
\left[
\bm\phi_m(\para_1),\ldots,\bm\phi_m(\para_{n_t})
\right]^\top
\),
\(\bm A_t=(\bm K_t+\varsigma^2\bm I_{n_t})^{-1}\).
Conditional on
\(\Xi_m=\{(\omega_j,b_j)\}_{j=1}^m\), which is held fixed within an ESA
run, the decoupled RFF--Matheron path is
\begin{equation*}
f_{t,m}^{\rm dec}(\para;\bm w,\bm\varepsilon)
=
\bm\phi_m(\para)^\top\bm w
+
\bm k_t(\para)^\top\bm A_t
\left\{
\bm y_t-\bm\Phi_{t,m}\bm w-\bm\varepsilon
\right\},
\end{equation*}
where
\(
\bm w\sim\mathcal N(\bm0, \bm I_m)
\),
\(
\bm\varepsilon\sim
\mathcal N(\bm0,\varsigma^2\bm I_{n_t})
\),
independently. RFF supplies the prior-path component, while the Matheron
correction retains the exact kernel vector and Gram matrix
\citep{matheron1973intrinsic,wilson2021pathwise}.

Writing
\(
\bm z=(\bm w^\top,\bm\varepsilon^\top)^\top
\),
the path is a deterministic function of \(\bm z\), so \(Q_f\) is
implemented on this Gaussian latent state and ESA returns
\(
(\widetilde\para_t,\widetilde{\bm z}_t)
\).
Levelwise convergence of the finite-dimensional construction is established
in \Cref{oa:rff_esa_convergence}.

Because \(\widehat p_{t,k}\) in \Cref{eq:p_t_proxy} increases with \(\Delta_{t,k}^{\rm loc}\), the terminal ESA states can remain
too dispersed around \(\para_t^\dagger\) for \(\widehat p_{t,k}\) to become
sufficiently small under a practical evaluation budget. We therefore apply a budget-limited local refinement to the terminal output. Let
\(
\widetilde{\bm z}_t
=
(\widetilde{\bm w}_t^\top,\widetilde{\bm\varepsilon}_t^\top)^\top
\),
\(
\widetilde f_{t,m}
=
f_{t,m}^{\rm dec}
(\cdot;\widetilde{\bm w}_t,\widetilde{\bm\varepsilon}_t)
\),
and set
\(
\para_{n_t+1}
=
\mathfrak L_{\widetilde f_{t,m}}(\widetilde\para_t)
\),
where \(\mathfrak L_h(\para)\) denotes local optimisation of \(h\)
initialised at \(\para\). This moves the proposal towards a local maximiser
of the sampled path within the basin selected by ESA, improving the local
design resolution and driving \(\widehat p_{t,k}\) downward. As a secondary
effect, it partially counteracts the optimisation bias associated with
a fixed \(\tau\), quantified for exact AO in
\Cref{thm:ao_main_text}; it does not attempt global optimisation of the
terminal path \(\widetilde f_{t,m}\). The complete proposal kernels, temperature schedule,
refinement budget, and numerical settings are reported in \Cref{oa:esa_settings}.

\subsection{TV Distance Computation}
\label{appdx:tv_computation}

For the benchmark and simulated examples, let
\(
\para^{(1)},\ldots,\para^{(M)}
\sim
\pi_\dataseq
\)
be a reference sample, and let
\(
\pi_{\mu_T}=\pi(\cdot\mid\mu_T)
\)
be the plug-in law from
\Cref{def: expected posterior}.
Since
\begin{equation*}
{\mathrm{TV}}(\pi_{\mu_T}, \pi_\dataseq)
=
\frac12
\EE_{\pi_\dataseq}
\left[
\left|
1-
\frac{\exp\{\mu_T(\para)-f(\para)\}}
{\EE_{\pi_\dataseq}
 [\exp\{\mu_T(\para)-f(\para)\}]}
\right|
\right],
\end{equation*}
define
\(
\ell_m
=
\mu_T(\para^{(m)})-f(\para^{(m)})
\),
\(
\bar w
=M^{-1}\sum_{m=1}^M\exp(\ell_m)
\).
The estimator is
\begin{equation*}
\widehat{\mathrm{TV}}(\pi_{\mu_T},\pi_\dataseq)
=
\frac{1}{2M}
\sum_{m=1}^M
\left|
1-\frac{\exp(\ell_m)}{\bar w}
\right|.
\end{equation*}
Construction and diagnostics of the reference samples are given in \Cref{oa:implementation_details}.

\section{Proofs}
\label{app: proofs}

This appendix records additional assumptions, notations, standard GP results, and proofs of major technical results in the paper.

\subsection{Notation, Useful Results, and Additional Assumptions}
\label{appdx:proof_setup}

We work under Assumption~\ref{ass:main_gp_setup} and use active-time
indexing. Unless a result explicitly takes \(\mathcal D_0=\emptyset\),
\(\mathcal D_0\) may be any finite initial design whose input locations are
chosen independently of \(f\), the observation noises, and the subsequent
algorithmic randomness. Results with an initial design are read
conditionally on that design after restarting the active index $t$. Write
\(
\EE_t(\cdot)
=
\EE(\cdot\mid\mathcal H_t)\),
\(\operatorname{Var}_t(\cdot)
=
\operatorname{Var}(\cdot\mid\mathcal H_t).
\)
Unless stated otherwise, \(\|\cdot\|\) is Euclidean and
\(\|\cdot\|_{\rm op}\) is the matrix operator norm. For measurable
\(A\subseteq\mathbb R^d\), set \(|A|:=\lambda(A)\). We retain the
Euclidean-ball, covering-number, and diameter notation
\(\mathbb B(\para,r)\), \(N(\epsilon,\Omega,\mathscr d)\), and
\(\operatorname{diam}(\Omega,\mathscr d)\) from
\Cref{ass:main_gp_setup}. Let
\(\mathcal P(\Omega)\) be the Borel probability measures on \(\Omega\). For later use, write
\(
G_t:=f_t-\mu_t
\)
and
\(
M_t^\mu
:=
\sup_{\para,\para'\in\Omega}
|\mu_t(\para)-\mu_t(\para')|.
\)

For AO, retain \(\para^\star\), \(R_T\), and \(\mathfrak R_T\) from
\Cref{def:ao_regret}. Fix a measurable maximiser selector
\(s:C(\Omega)\to\Omega\), take
\begin{equation*}
    \para^\star=s(f),
    \qquad
    \para_{f_t}^\star=s(f_t),
    \qquad
    \para_{t+1}^{\rm AO}\mid f_t
    \sim
    \pi_\tau(\cdot\mid f_t),
\end{equation*}
and write
\(
    \Delta_{t+1}^{\rm AO}
    =
    \EE_t\{
    f(\para^\star)-f(\para_{t+1}^{\rm AO})
    \}.
\)
Then
\(
\mathfrak R_T
=
\sum_{t=0}^{T-1}\EE\{\Delta_{t+1}^{\rm AO}\}
\)
by the tower property.

For UR, retain \(N_T^{\rm UR}\) from
\Cref{subsec:ur_calibration} and write
\(
I_t^{\rm UR}:=1-B_t
\), so that
\(
N_T^{\rm UR}=\sum_{t=0}^{T-1}I_t^{\rm UR}.
\)

Posterior-variance contraction gives
\(
0
\le
\sigma_t^2(\para)
\le
\kappa(\para,\para)
\le
B_\kappa .
\) Write
\begin{equation*}
    \mathscr d_t^2(\para,\para')
:=
\operatorname{Var}_t\{f(\para)-f(\para')\},
\end{equation*}
for the posterior canonical metric. This is also the
conditional canonical metric of \(G_t\).

\subsubsection{Standard Gaussian-Process Supremum Bound}

We repeatedly use the following conditional form of Dudley's entropy bound
and the Borell--TIS inequality
\citep[Theorems~1.3.3 and~2.1.1]{adler2007random};
see also \citet{dudley1967sizes}.

\begin{lemma}
\label{lem:conditional_gp_supremum}
Conditionally on a \(\sigma\)-field \(\mathcal G\), let
\(X=\{X(\para):\para\in\mathcal T\}\) be a centred separable continuous
Gaussian process on compact \(\mathcal T\), with conditional canonical
metric
\begin{equation*}
    \mathscr d_X(\para,\para')^2
:=
\operatorname{Var}
\{X(\para)-X(\para')\mid\mathcal G\}.
\end{equation*}
Suppose that, almost surely,\(
\mathscr d_X(\para,\para')
\leq
\mathscr d_0(\para,\para')\),
\(\sup_{\para\in\mathcal T}
\operatorname{Var}\{X(\para)\mid\mathcal G\}
\leq v^2,
\)
and
\begin{equation*}
    J(\mathscr d_0)
:=
\int_0^{\operatorname{diam}(\mathcal T,\mathscr d_0)}
\sqrt{\log N(\epsilon,\mathcal T,\mathscr d_0)}
\,\dee\epsilon
<\infty.
\end{equation*}
Then, writing
\(
S_X
=
\sup_{\para\in\mathcal T}|X(\para)|,
\)
\begin{equation*}
    m_{\mathcal G}
:=
\EE[S_X\mid\mathcal G]
\leq
v+C_DJ(\mathscr d_0)
\quad\text{a.s.},
\end{equation*}
for a universal constant \(C_D<\infty\), and for $u \geq 0$,
\begin{equation}
\label{eq:Borel-TIS tail bound}
    \Pr(S_X>m_{\mathcal G}+u\mid\mathcal G)
\leq
2\exp\left(-\frac{u^2}{2v^2}\right),
\quad\text{a.s.}
\end{equation}
Hence, for all \(q\geq1, a>0\), \(\EE(S_X^q\mid\mathcal G)
\leq C_q\{v^q+J(\mathscr d_0)^q\}\), \(\EE\{\exp(aS_X)\mid\mathcal G\}<\infty\) a.s., where \(C_q<\infty\) depends only on \(q\).
\end{lemma}

The result follows by applying the cited inequalities under the regular
conditional Gaussian law; the moment conclusions follow by integrating \Cref{eq:Borel-TIS tail bound}.
\subsubsection{UR Calibration Assumptions}
\begin{assumption}
\label{ass:prior_spectral_decay}
Let \(\mathcal C_0^\lambda\) be the prior covariance operator on
\(L^2(\Omega,\lambda)\), that is,
\begin{equation*}
    (\mathcal C_0^\lambda h)(\para)
=
\int_\Omega \kappa(\para,\bm u)h(\bm u)\,\lambda(\dee\bm u),
\end{equation*}
where \(\kappa\) is the prior covariance kernel. Let
\(
\lambda_{0,1}^\lambda\ge\lambda_{0,2}^\lambda\ge\cdots\ge0
\)
be its eigenvalues. There exists \(A_0<\infty\) such that at least one of
the following conditions holds:
\begin{enumerate}[label=\textup{(\roman*)},leftmargin=*,itemsep=1pt]
\item \emph{Polynomial spectral decay:} for some \(\nu>0\),
\(
    \lambda_{0,j}^\lambda
\le
A_0j^{-(1+2\nu/d)}\), \(j\ge1\).
\item \emph{Exponential spectral decay:} for some \(a_0>0\),
\(
    \lambda_{0,j}^\lambda
\le
A_0\exp\{-a_0j^{1/d}\}\), \(j\ge1\).
\end{enumerate}
\end{assumption}

The polynomial and exponential conditions correspond, respectively, to the
standard Mat\'ern-$\nu$ and SE kernels documented in
\Cref{rem:ideal_ur_calibration_scope}. The exponential condition implies the
polynomial condition for every finite \(\nu>0\), after changing \(A_0\).

\subsubsection{Allocation-Proxy Smoothness}
\begin{assumption}
\label{ass:delta_state_proxy_smoothness}
The prior GP admits a version with a.s. \(C^2(U)\) sample paths on an open
neighbourhood \(U\supset\Omega\). For each \(r,s\in[d]\), the
second-derivative process
\(
D_{rs}f(\para)=\partial_r\partial_s f(\para)
\)
is a centred continuous GP on \(\Omega\), jointly Gaussian with finite
collections of function values, with canonical metric
\(
\mathscr d_{0,rs}(\para,\para')^2
:=
\mathrm{Var}\{D_{rs}f(\para)-D_{rs}f(\para')\}
\)
satisfying
\begin{equation*}
    J_0^{(rs)}
:=
\int_0^{\mathrm{diam}\left(\Omega, \mathscr d_{0,rs}\right)}
\sqrt{\log N\left(\eps,\Omega,\mathscr d_{0,rs}\right)}\,\dee\eps
<\infty,
\end{equation*}
and
\(
\sup_{\para\in\Omega}
\mathrm{Var}\{D_{rs}f(\para)\}
\le
B_{2,rs}
<
\infty .
\)

Assume also that
\(
\partial_r\partial_s\mu_t(\para)
=
\EE_t\{D_{rs}f(\para)\}
\)
a.s., \(\para\in\Omega\). The curvature matrix used in the allocation
proxy has the form
\(
\mathbf C_t
=
\mathcal M\{
-
(\mathbf H_t
+
\mathbf H_t^\top)/2\},
\)
where $\mathbf H_t = \nabla^2\mu_t(\para_t^\dagger)$ and \(\mathbf C_t\succ0\), and the deterministic map
\(\mathcal M\) satisfies
\(
\|\mathcal M(\mathbf H)\|_{\rm op}
\leq
C_{\mathcal M}\{1+\|\mathbf H\|_{\rm op}\}
\)
for every symmetric matrix \(\mathbf H\). Here \(\mathcal M\) abstracts the
implementation's numerical positive-definite stabilisation after negative
symmetrisation of the local Hessian, such as a ridge correction or
positive-eigenvalue projection. The order statistics defining
\(\Delta_{t,k}^{\rm loc}\) are computed after excluding zero
curvature-adapted distances, equivalently after excluding evaluations at
the incumbent location under \(\mathbf C_t\succ0\).
\end{assumption}

\subsection[Proofs of the AO Limit and Perturbation Results]{Proofs of \Cref{prop:ao_ts_limit_main,lem:ao_perturb}}
\label{appdx:proof_opt}

\begin{proof}[\Cref{prop:ao_ts_limit_main}]
Let \(\bar\lambda=\lambda/|\Omega|\). By
\Cref{ass:main_gp_setup}, \(\bar\lambda\) has full support on
compact \(\Omega\). For \(\PP_t\)-almost every \(f_t\), the path is continuous
and has the unique maximiser \(\para_{f_t}^\star\). Hence the
zero-temperature Laplace principle \citep{hwang1980laplace} gives, for every
bounded continuous \(g:\Omega\to\mathbb R\),
\begin{equation*}
\int_\Omega g(\para)\,\pi_\tau(\dee\para\mid f_t)
\rightarrow
g(\para_{f_t}^\star), \quad \text{ as } \tau \downarrow 0.
\end{equation*}
Since the left-hand side is bounded in absolute value by
\(\|g\|_\infty\), conditional dominated convergence yields
\begin{equation*}
    \int_\Omega g(\para)\,\pi_{\tau,t}(\dee\para)
\rightarrow
\int g(\para_{f_t}^\star)\,\PP_t(\dee f_t) , \quad \text{ as } \tau \downarrow 0.
\end{equation*}
This proves the asserted weak convergence.
\end{proof}

\begin{proof}[\Cref{lem:ao_perturb}]
For simplicity, write
\(\varrho=\varrho(g_1,g_2)\). By
\Cref{def:annealed_objective}, the corresponding normalising constants satisfy
\begin{equation*}
    e^{-\varrho/\tau}\mathcal Z_\tau(g_2)
\leq
\mathcal Z_\tau(g_1)
\leq
e^{\varrho/\tau}\mathcal Z_\tau(g_2).
\end{equation*}
Thus
\(
\left\|
\log(p_{\tau,1}/p_{\tau,2})
\right\|_\infty
\leq 2\varrho/\tau.
\)
For two probability measures $\pi$ and $\rho$ with $\pi \ll \rho$, if the Radon-Nikodym derivative satisfies \(e^{-a}\leq\dee \pi/\dee \rho \leq e^a\) a.e., then
\begin{equation*}
    \TV(\pi, \rho) \leq\tanh(a/2).
\end{equation*}
Indeed, taking \(B=\{\dee \pi/\dee \rho\geq1\}\), \(x=\pi(B)\), and \(y=\rho(B)\),
the inequalities \(x\leq e^ay\) and
\(1-y\leq e^a(1-x)\) give
\(x-y\leq(e^a-1)/(e^a+1)\). Consequently,
\begin{equation*}
    \TV\left\{
\pi_\tau(\cdot\mid g_1),\pi_\tau(\cdot\mid g_2)
\right\}
\leq
\tanh(\varrho/\tau).
\end{equation*}

Let
\(\gamma\in\Gamma(\mathbb Q_1,\mathbb Q_2)\).
Since \(\gamma\) has marginals \(\mathbb Q_1\) and \(\mathbb Q_2\), for
every \(A\in\mathscr{B}(\Omega)\),
\begin{align*}
\left|
\mathcal A_\tau(\mathbb Q_1)(A)
-
\mathcal A_\tau(\mathbb Q_2)(A)
\right|
&\leq
\int
\TV\left\{
\pi_\tau(\cdot\mid h_1),
\pi_\tau(\cdot\mid h_2)
\right\}
\gamma(\dee h_1,\dee h_2)
\\
&\leq
\frac{1}{\tau}
\int
\|h_1-h_2\|_\infty\,
\gamma(\dee h_1,\dee h_2),
\end{align*}
where \(\tanh(u)\leq u\) for \(u\geq0\) was used. Taking the supremum
over \(A\) and then the infimum over \(\Gamma(\mathbb Q_1,\mathbb Q_2)\) proves the AO-law bound.

Finally, under \Cref{ass:main_gp_setup}\textup{(A5)}, posterior covariance
contraction and \Cref{lem:conditional_gp_supremum} give
\(
\EE_t\{\|f_t-\mu_t\|_\infty\}<\infty
\) a.s.
Since \(\mu_t\) is continuous on compact \(\Omega\), \(\PP_t\) therefore
has a finite first sup-norm moment almost surely.
\end{proof}

\subsection[Completion of the AO Regret Proof]{Proof of \Cref{thm:ao_main_text}}
\label{app:ao_regret}

Throughout this section, the AO query at active time \(t\) is generated
before \(y_{t+1}\) is observed. Conditional on \(\mathcal H_t\),
\(\para_{t+1}^{\mathrm{AO}}\) is generated from an independent posterior draw
\(f_t\sim\PP_t\) and auxiliary randomness, both conditionally independent of
the latent objective \(f\) and the fresh observation noise. Thus the AO design
is non-anticipative. Conditional on the realised design locations, the usual
deterministic-design GP covariance recursion and information-gain identities
therefore apply to the adaptive AO trajectory.

\subsubsection{Posterior Regularity}

\begin{lemma}
\label{lem:ao_post_C1}
Under \Cref{ass:main_gp_setup}, for every \(t\ge0\),
\(\PP_t(C^1(U))=1\) a.s. Hence \(f_t\) and \(G_t\) may be chosen in
\(C^1(U)\) a.s. Moreover, for each \(j\in[d]\),
\(D_jG_t\mid\mathcal H_t\) is a centred continuous GP on \(\Omega\). If
\(
\mathscr d_{t,j}(\para,\para')^2
:=
\mathrm{Var}_t\{D_jG_t(\para)-D_jG_t(\para')\},
\)
then, almost surely,
\begin{equation}
\label{eq:ao_post_grad_metric_dom}
\mathscr d_{t,j}(\para,\para')
\le
 \mathscr d_{0,j}(\para,\para'),
\qquad
\forall \para,\para'\in\Omega,
\end{equation}
and
\begin{equation}
\label{eq:ao_post_grad_var_dom}
\mathrm{Var}_t\{D_jG_t(\para)\}
\le
B_{\partial,j},
\qquad
\forall \para\in\Omega.
\end{equation}
\end{lemma}

\begin{proof}
Since \(\Pr\{f\in C^1(U)\}=1\),
\(
\PP_t\{C^1(U)\}
=
\EE\!\left[
\mathbf 1_{\{f\in C^1(U)\}}
\mid\mathcal H_t
\right]
=
1
\) a.s,
Hence \(f_t\in C^1(U)\) a.s., and
\(G_t=f_t-\mu_t\in C^1(U)\) a.s. by
\Cref{ass:main_gp_setup}\textup{(A5)}.

The remaining assertions follow from the standard closure of GPs under linear differentiation and Gaussian conditioning
\citep[Equation~(2.24) and Section~9.4]{rasmussen2006gaussian}.
Indeed, the non-anticipation condition in
\Cref{ass:main_gp_setup}\textup{(A3)} allows these identities to be
applied sequentially along the realised adaptive design. Thus
\(D_jG_t\mid\mathcal H_t\) is a centred continuous GP, and its posterior
covariance is the corresponding prior covariance minus a
positive-semidefinite Gaussian-conditioning term. Applying this covariance
order to \(D_jf(\para)\) and
\(D_jf(\para)-D_jf(\para')\) gives
\Cref{eq:ao_post_grad_var_dom,eq:ao_post_grad_metric_dom}, respectively.
\end{proof}

\begin{lemma}
\label{lem:ao_Lt}
Define \(L_t=1+\mathrm{Lip}(f_t)\) and \(L_\star=1+\mathrm{Lip}(f)\), where
\begin{equation*}
    \mathrm{Lip}(g)
=
\sup_{\para,\para'\in\Omega,\ \para\ne\para'}
\frac{|g(\para)-g(\para')|}{\|\para-\para'\|}.
\end{equation*}
Under \Cref{ass:main_gp_setup}, \(L_t,L_\star<\infty\) a.s., and there
is a constant \(C_L<\infty\) such that
\begin{equation}
\label{eq:ao_CL}
\sup_{t\ge0}\EE\{\log_+(r_\Omega L_t)\}
=
\EE\{\log_+(r_\Omega L_\star)\}
\le C_L.
\end{equation}
\end{lemma}

\begin{proof}
The \(C^1\) sample-path property and compact convexity of \(\Omega\) imply
\begin{equation*}
    \mathrm{Lip}(g|_\Omega)
\le
\sup_{\para\in\Omega}\|\nabla g(\para)\|,
\quad g\in C^1(U).
\end{equation*}
Thus \(L_t,L_\star<\infty\) a.s. Let
\(\Phi(g)=\log_+[r_\Omega\left\{1+\mathrm{Lip}(g|_\Omega)\right\}]\). Since
\(f_t\mid\mathcal H_t\) and \(f\mid\mathcal H_t\) have the same conditional law,
\begin{equation*}
    \EE\{\Phi(f_t)\}
=
\EE[\EE_t\{\Phi(f_t)\}]
=
\EE[\EE_t\{\Phi(f)\}]
=
\EE\{\Phi(f)\},
\end{equation*}
which proves the equality in \Cref{eq:ao_CL}. For finiteness,
\Cref{lem:conditional_gp_supremum}, applied to \(D_jf\) with
\(\mathscr d_0=\mathscr d_{0,j}\) and \(v^2=B_{\partial,j}\), gives
\(
\EE\left\{\sup_{\para\in\Omega}|D_jf(\para)|\right\}<\infty,
\ j\in[d].
\)
Consequently,
\begin{equation*}
    \EE\{\mathrm{Lip}(f)\}
\leq
\sum_{j=1}^d
\EE\left\{\sup_{\para\in\Omega}|D_jf(\para)|\right\}
<\infty.
\end{equation*}
Since \(\log_+(x)\leq x\) for \(x\geq0\), the common expectation in
\Cref{eq:ao_CL} is finite.
\end{proof}

\begin{lemma}
\label{lem:ao_Mbar}
Let \(\mathcal S_t^{\mathrm{draw}}=\sup_{\para\in\Omega}|f_t(\para)-\mu_t(\para)|\).
Under \Cref{ass:main_gp_setup}, there exists \(\bar M<\infty\),
depending only on \(J_0\) and \(B_\kappa\), such that
\begin{equation*}
    \left[\EE\left\{(\mathcal S_t^{\mathrm{draw}})^2\right\}\right]^{1/2}
\le
\bar M,
\qquad
\forall t\ge0.
\end{equation*}
\end{lemma}

\begin{proof}
Condition on \(\mathcal H_t\). Posterior
covariance contraction gives
\(
\mathscr d_t\leq\mathscr d_0,\)
\(\sup_{\para\in\Omega}
\operatorname{Var}_t\{G_t(\para)\}
\leq B_\kappa.
\)
The \(q=2\) conclusion of \Cref{lem:conditional_gp_supremum}, with
\(J(\mathscr d_0)=J_0\), therefore gives
\begin{equation*}
    \EE_t\left[
\left\{\sup_{\para\in\Omega}|G_t(\para)|\right\}^2
\right]
\leq
C_2(J_0^2+B_\kappa)
\quad\text{a.s.}
\end{equation*}
Taking expectations proves the claim with
\(
\bar M=\{C_2(J_0^2+B_\kappa)\}^{1/2}.
\)
\end{proof}

\subsubsection{Posterior Identities and Gibbs Softening}

\begin{lemma}
\label{lem:ao_posterior_identities}
For every \(t\ge0\),
\(
\EE_t\left\{f(\para^\star)\right\}
=
\EE_t\{f_t(\para_{f_t}^\star)\}\),
\(\EE_t\left\{f(\para_{t+1}^{\mathrm{AO}})\right\}
=
\EE_t\left\{\mu_t(\para_{t+1}^{\mathrm{AO}})\right\}.
\)
\end{lemma}

\begin{proof}
Conditionally on \(\mathcal H_t\), \(f\) and \(f_t\) are i.i.d. from \(\PP_t\).
Therefore \(f\left\{s(f)\right\}=
f(\para^\star)\) and \(f_t\left\{s(f_t)\right\}=
f_t(\para_{f_t}^\star)\) have the same conditional law,
proving the first identity. For the second, \(\para_{t+1}^{\mathrm{AO}}\) is
generated from \(f_t\) and auxiliary randomness, and is conditionally
independent of \(f\) given \(\mathcal H_t\). Thus
\begin{equation*}
    \EE_t\left\{f(\para_{t+1}^{\mathrm{AO}})\right\}
=
\EE_t\left[\EE\left\{f(\para_{t+1}^{\mathrm{AO}})\mid\mathcal H_t,\para_{t+1}^{\mathrm{AO}}\right\}\right]
=
\EE_t\left\{\mu_t(\para_{t+1}^{\mathrm{AO}})\right\}.
\end{equation*}
\end{proof}

\begin{lemma}
\label{lem:ao_softening_pathwise}
Under \Cref{ass:main_gp_setup}, define
\(
C_\Omega
=
1+\log\!\left\{|\Omega|/(c_\Omega r_\Omega^d)\right\}.
\)
Then, for every \(t\ge0\),
\begin{equation*}
f_t(\para_{f_t}^\star)
-
\EE\{f_t(\para_{t+1}^{\mathrm{AO}})\mid\mathcal H_t,f_t\}
\le
\tau
\left\{
C_\Omega
+
d\log_+\!\left(r_\Omega L_t/\tau\right)
\right\}
\quad\text{a.s.}
\end{equation*}
\end{lemma}

\begin{proof}
Condition on \(\mathcal H_t\) and \(f_t=g\). Let
\(\para_g^\star=s(g)\), \(g^\star=g(\para_g^\star)\),
\(L=1+\mathrm{Lip}(g)\), and \(r=r_\Omega\wedge \tau/L\). For every
\(\tilde\para\in \mathbb{B}(\para_g^\star,r)\cap\Omega\),
\begin{equation*}
    g(\tilde\para)
\ge
g^\star-\mathrm{Lip}(g)\|\tilde\para-\para_g^\star\|
\ge
g^\star-Lr
\ge
g^\star-\tau .
\end{equation*}
Hence, with \(\mathcal Z_\tau(g)\) from
\Cref{def:annealed_objective}, we have
\begin{equation*}
    \mathcal Z_\tau(g)
\ge
\int_{\mathbb B(\para_g^\star,r)\cap\Omega}
\exp\{g(\bm u)/\tau\}\,\lambda(\dee\bm u)
\ge
\exp(g^\star/\tau-1)c_\Omega r^d.
\end{equation*}
Therefore,
\begin{equation}\label{eq: g_star_bound}
    g^\star-\tau\log \mathcal Z_\tau(g)
\le
\tau-\tau\log c_\Omega-d\tau\log r.
\end{equation}

Let
\(
p_{\tau,g}
:=
\dee\pi_\tau(\cdot\mid g)/\dee\lambda
\)
be the Lebesgue density from
\Cref{def:annealed_objective}. Therefore,
\begin{equation*}
    \tau\log\mathcal Z_\tau(g)
=
\int_\Omega g(\para)p_{\tau,g}(\para)\,\lambda(\dee\para)
+
\tau \mathbb H(p_{\tau,g}),
\quad
\mathbb H(p_{\tau,g})
:=
-\int_\Omega p_{\tau,g}(\para)\log p_{\tau,g}(\para)\,\lambda(\dee\para).
\end{equation*}
Moreover, \(\mathrm{KL}(p_{\tau,g}\|u_\Omega)\geq0\), where
\(u_\Omega=|\Omega|^{-1}\mathbf 1_\Omega\), gives
\(\mathbb H(p_{\tau,g})\leq\log|\Omega|\). Hence
\begin{equation*}
    \int_\Omega g(\para)p_{\tau,g}(\para)\,\lambda(\dee\para)
\geq
\tau\log\mathcal Z_\tau(g)-\tau\log|\Omega|.
\end{equation*}
Combining this inequality with \Cref{eq: g_star_bound} gives
\begin{equation*}
    g^\star-
    \int_\Omega g(\para)p_{\tau,g}(\para)\,\lambda(\dee\para)
\le
\tau\left\{
1+\log\!\left(|\Omega|/c_\Omega\right)-d\log r
\right\}.
\end{equation*}
Finally, since \(r=r_\Omega\wedge \tau/L\),
\(
-\log r
\le
-\log r_\Omega+\log_+\!\left(r_\Omega L/\tau\right).
\)
Substitution gives the claimed bound.
\end{proof}
\begin{corollary}
\label{cor:ao_softening}
Under \Cref{ass:main_gp_setup},
\(
\EE_t[
 f_t(\para_{f_t}^\star)
 -
 f_t(\para_{t+1}^{\mathrm{AO}})
]
\le
\tau
\left[
C_\Omega
+
d\,\EE_t\!\left\{
\log_+\!\left(r_\Omega L_t/\tau\right)
\right\}
\right].
\)
\end{corollary}

\begin{proof}
Take conditional expectation in \Cref{lem:ao_softening_pathwise}.
\end{proof}

\subsubsection{Variance Accumulation}

Fix an auxiliary variance \(\upsilon^2>0\) satisfying
\(\upsilon^2\ge\varsigma^2\). \(\upsilon^2\) is used only in
the information-gain argument. For a deterministic design
\(\para_{1:t}\), write
\(
\bm K_t(\para_{1:t})
=
\{\kappa(\para_i,\para_j)\}_{i,j=1}^t,
\)
\(\bm k_t(\para;\para_{1:t})
=
[\kappa(\para,\para_1),\ldots,\kappa(\para,\para_t)]^\top,
\)
and define the auxiliary posterior variance
\begin{equation*}
\sigma_{t,\upsilon}^2(\para;\para_{1:t})
=
\kappa(\para,\para)
-
\bm k_t(\para;\para_{1:t})^\top
\{\bm K_t(\para_{1:t})+\upsilon^2\bm I_t\}^{-1}
\bm k_t(\para;\para_{1:t}).
\end{equation*}
Define the corresponding regularised maximal information gain by
\begin{equation*}
\gamma_T(\upsilon)
=
\sup_{\para_{1:T}\in\Omega^T}
\frac12
\log\det
\{\bm I_T+\upsilon^{-2}\bm K_T(\para_{1:T})\}.
\end{equation*}

\begin{lemma}
\label{lem:ao_varsum}
Under Assumption~\ref{ass:main_gp_setup} \textup{(A1), (A2), and (A5)},
fix any
\(\upsilon^2>0\) satisfying \(\upsilon^2\ge\varsigma^2\), and let
\(
C_{\gamma,\upsilon}
=
2B_\kappa/\log(1+\upsilon^{-2}B_\kappa).
\)
Then, for every \(T\ge1\) and every possibly adaptive query sequence,
\begin{equation*}
\sum_{t=0}^{T-1}
\sigma_t^2(\para_{t+1})
\le
C_{\gamma,\upsilon}\gamma_T(\upsilon)
\quad\text{a.s.}
\end{equation*}
Consequently,
\(
\sum_{t=0}^{T-1}
\EE[
\sigma_t^2(\para_{t+1})
]
\le
C_{\gamma,\upsilon}\gamma_T(\upsilon).
\)
\end{lemma}

\begin{proof}
Fix a realised sequence of query locations
\(\para_{1:T}\), where the active index $t$ starts after any initial
design $\mathcal{D}_0$. Conditioning on the initial observations can only reduce posterior
variance. For the active locations alone, replacing \(\varsigma^2\) by
\(\upsilon^2\geq\varsigma^2\) can only increase posterior variance, by matrix
inverse monotonicity. Hence
\begin{equation*}
    \sigma_t^2(\para)
\leq
\sigma_{t,\upsilon}^2(\para;\para_{1:t}),
\qquad t\geq0.
\end{equation*}
When \(\varsigma^2=0\), this comparison is applied after redundant exact
observations have been removed, as in
Assumption~\ref{ass:main_gp_setup}.

For any deterministic sequence \(\para_{1:T}\), the sequential determinant
identity \citep[Lemma~5.3]{srinivas2012information} gives
\begin{equation*}
\frac12
\sum_{t=0}^{T-1}
\log\left\{
1+
\frac{
\sigma_{t,\upsilon}^2(\para_{t+1};\para_{1:t})
}{
\upsilon^2
}
\right\}
=
\frac12
\log\det
\{\bm I_T+\upsilon^{-2}\bm K_T(\para_{1:T})\}
\le
\gamma_T(\upsilon).
\end{equation*}
This identity is algebraic in the realised locations and therefore applies
pathwise to the adaptive query sequence.

Moreover,
\(
0
\le
\sigma_{t,\upsilon}^2(\para;\para_{1:t})
\le
B_\kappa.
\)
Since \(a/\log(1+a/\upsilon^2)\) is increasing,
\begin{equation*}
    a
\le
\frac{B_\kappa}
{\log(1+B_\kappa/\upsilon^2)}
\log(1+a/\upsilon^2),
\qquad 0\le a\le B_\kappa.
\end{equation*}
Applying this inequality to the auxiliary variances and using their
pathwise domination of the actual variances yields
\begin{equation*}
\sum_{t=0}^{T-1}
\sigma_t^2(\para_{t+1})
\le
\frac{2B_\kappa}
{\log(1+B_\kappa/\upsilon^2)}
\gamma_T(\upsilon)
=
C_{\gamma,\upsilon}\gamma_T(\upsilon).
\end{equation*}
Taking expectations proves the second assertion.
\end{proof}
\subsubsection{Regret Decomposition and Posterior-Selection Bias}

Define the AO posterior-selection bias
\(
b_{t+1}^{\mathrm{sel}}
=
\EE_t\{
 f_t(\para_{t+1}^{\mathrm{AO}})
 -
 \mu_t(\para_{t+1}^{\mathrm{AO}})
\}.
\)

\begin{lemma}
\label{lem:ao_direct_decomp}
Under \Cref{ass:main_gp_setup}, for every \(t\ge0\),
\(
\Delta_{t+1}^{\mathrm{AO}}
=
\EE_t\{
 f_t(\para_{f_t}^\star)
 -
 f_t(\para_{t+1}^{\mathrm{AO}})\}
+
b_{t+1}^{\mathrm{sel}}.
\)
\end{lemma}

\begin{proof}
By \Cref{lem:ao_posterior_identities},
\(
\Delta_{t+1}^{\mathrm{AO}}
=
\EE_t\{f_t(\para_{f_t}^\star)\}
-
\EE_t\{\mu_t(\para_{t+1}^{\mathrm{AO}})\}.
\)
Adding and subtracting \(\EE_t\{f_t(\para_{t+1}^{\mathrm{AO}})\}\) gives the
result.
\end{proof}

We now control \(b_{t+1}^{\mathrm{sel}}\) by a deterministic covering
argument applied to \(G_t=f_t-\mu_t\). Recall from
\Cref{appdx:proof_setup} that \(\mathscr d_t\) is the conditional canonical
metric of \(G_t\).

\begin{lemma}
\label{lem:ao_Gt_Lip_tail}
Let \(\widetilde L_t=1+\mathrm{Lip}(G_t)\), and define
\(
\widetilde M_\partial
=
d\sqrt{B_\partial}
+
C_D\sum_{j=1}^d J_0^{(j)}\),
\(
B_\partial
=
\max_{1\le j\le d}B_{\partial,j},
\)
where \(C_D\) is the universal constant in Dudley's entropy bound. Under
\Cref{ass:main_gp_setup}, for every \(u\ge0\),
\begin{equation}
\label{eq:ao_Gt_Lip_tail}
\Pr\left(\widetilde L_t>1+\widetilde M_\partial+u\right)
\le
2d\exp\!\left(-\frac{u^2}{2d^2B_\partial}\right).
\end{equation}
\end{lemma}

\begin{proof}
Condition on \(\mathcal H_t\) and write
\(
S_{t,j}
:=
\sup_{\para\in\Omega}|D_jG_t(\para)|.
\)
By \Cref{lem:ao_post_C1}, the hypotheses of
\Cref{lem:conditional_gp_supremum} hold for \(D_jG_t\), with
\(\mathscr d_0=\mathscr d_{0,j}\) and \(v^2=B_{\partial,j}\). Hence, for
every \(x\geq0\),
\begin{equation*}
    \Pr\left\{
S_{t,j}>
\sqrt{B_{\partial,j}}+C_DJ_0^{(j)}+x
\middle|\mathcal H_t
\right\}
\leq
2\exp\left(-\frac{x^2}{2B_\partial}\right)
\quad\text{a.s.}
\end{equation*}
Since \(G_t\in C^1(U)\) and \(\Omega\) is compact and convex,
\(
\mathrm{Lip}(G_t)
\leq
\sum_{j=1}^d S_{t,j},
\)
while
\begin{equation*}
    \sum_{j=1}^d
\left\{\sqrt{B_{\partial,j}}+C_DJ_0^{(j)}\right\}
\leq
\widetilde M_\partial.
\end{equation*}
A union bound with \(x=u/d\), followed by averaging over
\(\mathcal H_t\), proves \Cref{eq:ao_Gt_Lip_tail}.
\end{proof}

\begin{corollary}
\label{cor:ao_Gt_horizon}
For \(T\ge1\) and \(\eta_T\in(0,1)\), define
\(
\widetilde\Lambda_T(\eta_T)
=
1+\widetilde M_\partial
+
d\sqrt{2B_\partial\log(2dT/\eta_T)}
\)
and
\(
\widetilde{\mathcal G}_T(\eta_T)
=
\bigcap_{t=0}^{T-1}
\{\widetilde L_t\le \widetilde\Lambda_T(\eta_T)\}.
\)
Then \(\Pr\{\widetilde{\mathcal G}_T(\eta_T)\}\ge1-\eta_T\).
\end{corollary}

\begin{proof}
Apply \Cref{lem:ao_Gt_Lip_tail} with
\(u=d\sqrt{2B_\partial\log(2dT/\eta_T)}\), then union bound over
\(t\in[T-1]_0\).
\end{proof}

\begin{lemma}
\label{lem:ao_d0_euclid}
Under \Cref{ass:main_gp_setup}, there is a constant
\(C_0<\infty\) such that, for all \(t\ge0\) and \(\para,\para'\in\Omega\),
\(
\mathscr d_t(\para,\para')
\le
\mathscr  d_0(\para,\para')
\le
 C_0\|\para-\para'\|,
\)
and
\(
|\sigma_t(\para)-\sigma_t(\para')|
\le
C_0\|\para-\para'\|.\)
\end{lemma}

\begin{proof}
The same posterior-variance contraction argument as in
\Cref{lem:ao_post_C1}, applied to the scalar linear functional
\(f(\para)-f(\para')\), gives \(\mathscr d_t\le \mathscr d_0\). Let
\(
M_j=\sup_{\para\in\Omega}|D_jf(\para)|\), \(L_\nabla=\sup_{\para\in\Omega}\|\nabla f(\para)\|.
\)
The \(q=2\) conclusion of \Cref{lem:conditional_gp_supremum}, applied to
each \(D_jf\), gives \(\EE(M_j^2)<\infty\). Therefore
\begin{equation*}
    \EE[L_\nabla^2]
\leq
d\sum_{j=1}^d\EE[M_j^2]
<\infty.
\end{equation*}
Since \(f\in C^1(U)\) a.s. and \(\Omega\) is compact convex,
\(
|f(\para)-f(\para')|
\le
L_\nabla\|\para-\para'\|\) a.s. Squaring and taking expectations yields
\(
    \mathscr d_0(\para,\para')\le C_0\|\para-\para'\|, C_0=\{\EE[L_\nabla^2]\}^{1/2}.\) Finally, for \(\para,\para'\in\Omega\), let
\(
a=\sigma_t^2(\para)\),
\(b=\sigma_t^2(\para')\),
\(c=\kappa_t(\para,\para').
\)
Since \(\kappa_t\) is a covariance kernel, \(c\le\sqrt{ab}\) by
Cauchy-Schwarz. Hence
\(
|\sigma_t(\para)-\sigma_t(\para')|^2
=
a+b-2\sqrt{ab}
\le
\mathscr d_t(\para,\para')^2
\), which proves the claim.
\end{proof}

\begin{lemma}
\label{lem:ao_det_grids}
There exists \(C_{\mathrm{cov}}<\infty\) such that, for every \(h\in(0,1]\),
there is a finite deterministic \(\mathcal C(h)\subset\Omega\) satisfying
\(
\sup_{\para\in\Omega}\min_{\bm u \in\mathcal C(h)}\|\para-\bm u\|\le h\), with
\(|\mathcal C(h)|\le C_{\mathrm{cov}}h^{-d}.
\)
\end{lemma}

\begin{proof}
Cover a bounded cube containing \(\Omega\) by lattice cubes of side length
\(h/\sqrt d\), and choose one representative point of \(\Omega\) from each
cell intersecting \(\Omega\).
\end{proof}

Fix \(T\ge1\), choose \(\eta_T,\bar\delta_T\in(0,1)\), and set
\(\delta_T=\eta_T+\bar\delta_T\le1\). For \( t\in[T-1]_0\), define
\(
h_{t,T}
=(t+1)^{-2}\{1+\widetilde\Lambda_T(\eta_T)+C_0\}^{-1}.
\)
Let \(\mathcal C_{t,T}=\mathcal C(h_{t,T})\), and choose a deterministic
nearest-grid map \([\cdot]_{t,T}:\Omega\to\mathcal C_{t,T}\). Then
\(
|\mathcal C_{t,T}|
\le
C_{\mathrm{cov}}
\{1+\widetilde\Lambda_T(\eta_T)+C_0\}^d
(t+1)^{2d}.
\)

\begin{lemma}
\label{lem:ao_grid_conf}
Let \(\varpi_t\ge1\) with
\(\sum_{t=0}^\infty\varpi_t^{-1}\le1\), and define
\(
\beta_t
=
2\log\!\left(|\mathcal C_{t,T}|\varpi_t/\bar\delta_T\right).
\)
Let
\(
\mathcal E_T(\bar\delta_T)
=
\bigcap_{t=0}^{T-1}
\bigcap_{\para \in\mathcal C_{t,T}}
\{G_t(\para)\le \sqrt{\beta_t}\,\sigma_t(\para)\}.
\)
Then \(\Pr\{\mathcal E_T(\bar\delta_T)\}\ge1-\bar\delta_T.\)
\end{lemma}

\begin{proof}
For fixed \(t\) and \(\para \in\mathcal C_{t,T}\), conditionally on \(\mathcal H_t\),
\(G_t(\para)\) is centred Gaussian with variance \(\sigma_t^2(\para)\). Therefore
\begin{equation*}
    \Pr\{G_t(\para)>\sqrt{\beta_t}\sigma_t(\para)\mid\mathcal H_t\}
\le
\exp(-\beta_t/2).
\end{equation*}
Taking expectations and union bounding over \(\para\in\mathcal C_{t,T}\) and
\(t\in[T-1]_0\),
\begin{equation*}
    \Pr\{\mathcal E_T(\bar\delta_T)^c\}
\le
\sum_{t=0}^{T-1}|\mathcal C_{t,T}|\exp(-\beta_t/2)
=
\bar\delta_T\sum_{t=0}^{T-1}\varpi_t^{-1}
\le
\bar\delta_T.
\end{equation*}
\end{proof}
The grids \(\mathcal C_{t,T}\) are proof devices only. They induce the
confidence radii \(\beta_t\), but they do not discretise the AO policy or
restrict \(\para_{t+1}^{\mathrm{AO}}\).

\begin{lemma}
\label{lem:ao_continuum_upgrade}
On \(\mathcal F_T=\widetilde{\mathcal G}_T(\eta_T)\cap
\mathcal E_T(\bar\delta_T)\), for every \(t\in[T-1]_0\) and
\(\para\in\Omega\),
\begin{equation*}
G_t(\para)
\le
\sqrt{\beta_t}\,\sigma_t(\para)
+
\frac{1+\sqrt{\beta_t}}{(t+1)^2}.
\end{equation*}
\end{lemma}

\begin{proof}
On \(\mathcal F_T\),
\begin{equation*}
    G_t(\para)
\le
G_t([\para]_{t,T})+\widetilde\Lambda_T(\eta_T)h_{t,T}
\le
\sqrt{\beta_t}\sigma_t([\para]_{t,T})
+
\widetilde\Lambda_T(\eta_T)h_{t,T}.
\end{equation*}
By \Cref{lem:ao_d0_euclid},
\(\sigma_t([\para]_{t,T})\le\sigma_t(\para)+C_0h_{t,T}\). Hence
\(
G_t(\para)
\le
\sqrt{\beta_t}\sigma_t(\para)
+
\{\widetilde\Lambda_T(\eta_T)+C_0\sqrt{\beta_t}\}h_{t,T}.
\)
The definition of \(h_{t,T}\) gives the stated bound.
\end{proof}

\begin{lemma}
\label{lem:ao_bias_disc}
For every \(t\in[T-1]_0\),
\begin{equation*}
\EE[b_{t+1}^{\mathrm{sel}}]
\le
\sqrt{\beta_t}\,\EE\{\sigma_t(\para_{t+1}^{\mathrm{AO}})\}
+
\frac{1+\sqrt{\beta_t}}{(t+1)^2}
+
\bar M\sqrt{\delta_T}.
\end{equation*}
\end{lemma}

\begin{proof}
By \Cref{cor:ao_Gt_horizon,lem:ao_grid_conf},
\(\Pr(\mathcal F_T)\ge1-\delta_T\). Since
\(b_{t+1}^{\mathrm{sel}}=\EE_t\{G_t(\para_{t+1}^{\mathrm{AO}})\}\), the tower property gives
\(\EE[b_{t+1}^{\mathrm{sel}}]=\EE\{G_t(\para_{t+1}^{\mathrm{AO}})\}\), and
\begin{equation*}
    \EE[b_{t+1}^{\mathrm{sel}}]
=
\EE\{G_t(\para_{t+1}^{\mathrm{AO}})\mathbf 1_{\mathcal F_T}\}
+
\EE\{G_t(\para_{t+1}^{\mathrm{AO}})\mathbf 1_{\mathcal F_T^c}\}.
\end{equation*}
On \(\mathcal F_T\), \Cref{lem:ao_continuum_upgrade} holds simultaneously for
all \(\para\in\Omega\); in particular it holds at the random point
\(\para_{t+1}^{\mathrm{AO}}\). Therefore
\*[
    \EE\{G_t(\para_{t+1}^{\mathrm{AO}})\mathbf 1_{\mathcal F_T}\}
&\le
\sqrt{\beta_t}\,
\EE\{\sigma_t(\para_{t+1}^{\mathrm{AO}})\mathbf 1_{\mathcal F_T}\}
+
\frac{1+\sqrt{\beta_t}}{(t+1)^2} \\
&\le
\sqrt{\beta_t}\,\EE\{\sigma_t(\para_{t+1}^{\mathrm{AO}})\}
+
\frac{1+\sqrt{\beta_t}}{(t+1)^2}.
\]
For the second term, by Cauchy-Schwarz and \Cref{lem:ao_Mbar}
\*[
\EE\{G_t(\para_{t+1}^{\mathrm{AO}})\mathbf 1_{\mathcal F_T^c}\}
&\le
\EE\{\mathcal S_t^{\mathrm{draw}}\mathbf 1_{\mathcal F_T^c}\} \\
&\le
[\EE\{(\mathcal S_t^{\mathrm{draw}})^2\}]^{1/2}\Pr(\mathcal F_T^c)^{1/2}
\le
\bar M\sqrt{\delta_T}.    
\]
Combining the two bounds
proves the claim.
\end{proof}

\subsubsection[Proof of the AO Regret Theorem]
{Proof of \Cref{thm:ao_main_text}}

\begin{proof}
In the preceding covering construction, choose
\(
\varpi_t=\pi^2(t+1)^2/6\),
\(\eta_T=\bar\delta_T=T^{-2}/2\),
\(
\delta_T=T^{-2}.
\)
By \Cref{lem:ao_direct_decomp,cor:ao_softening,lem:ao_bias_disc}, the
tower property, and Cauchy-Schwarz,
\begin{equation*}
    \mathfrak R_T
\le
T\tau
\left[
C_\Omega
+
d\,\EE\left\{
\log_+(r_\Omega L_\star/\tau)
\right\}
\right]+
\left\{
C_{\gamma,\upsilon}\gamma_T(\upsilon)
\sum_{t=0}^{T-1}\beta_t
\right\}^{1/2}
+
\sum_{t=0}^{T-1}
\frac{1+\sqrt{\beta_t}}{(t+1)^2}
+
T\bar M\sqrt{\delta_T},
\end{equation*}
where \Cref{lem:ao_varsum} was used for the variance sum. Define
\begin{equation*}
    \beta_T^\star
=
2\left[
\log\left(
\frac{C_{\rm cov}\pi^2}{6\bar\delta_T}
\right)
+
d\log\{1+\widetilde\Lambda_T(\eta_T)+C_0\}
+
2(d+1)\log T
\right].
\end{equation*}
The grid-cardinality bound gives
\begin{equation*}
    \beta_t\le\beta_T^\star,
\qquad
\sum_{t=0}^{T-1}\beta_t\le T\beta_T^\star,
\qquad
\sum_{t=0}^{T-1}
\frac{1+\sqrt{\beta_t}}{(t+1)^2}
\le
\frac{\pi^2}{6}(1+\sqrt{\beta_T^\star}).
\end{equation*}
Also,
\(
\EE\{\log_+(r_\Omega L_\star/\tau)\}
\le
C_L+\log_+(1/\tau)
\)
by \Cref{lem:ao_Lt}. Therefore
\begin{equation*}
    \mathfrak R_T
\le
T\tau
\left[
C_\Omega+d\{C_L+\log_+(1/\tau)\}
\right] +
\left\{
C_{\gamma,\upsilon}\beta_T^\star
T\gamma_T(\upsilon)
\right\}^{1/2}
+
\frac{\pi^2}{6}(1+\sqrt{\beta_T^\star})
+
\bar M .
\end{equation*}
Under the selected \(\eta_T\) and \(\bar\delta_T\),
\(
\widetilde\Lambda_T(\eta_T)
=
1+\widetilde M_\partial
+
d\sqrt{2B_\partial\log(4dT^3)},
\)
so \(\beta_T^\star=\mathcal O(\log T)\) for fixed \(d\). Substitution
gives
\begin{equation*}
    \mathfrak R_T
=
\mathcal O\left\{
\sqrt{T\gamma_T(\upsilon)\log T}
\right\}
+
\mathcal O\left[
T\tau\{1+\log_+(1/\tau)\}
\right],
\end{equation*}
after absorbing the lower-order
\(\mathcal O(\sqrt{\log T})+\mathcal O(1)\) terms.
\end{proof}

\subsection[Proofs of the Ideal-UR Calibration Results]
{Proofs of \Cref{prop:ur_calibration_link_main,thm:sale_main_calibration}}
\label{appdx:ideal_sale_delta_calibration}

This section analyses an idealised UR branch that chooses evaluations by the posterior-weighted one-step IMSPE-reduction objective. The calibration rate is controlled through deterministic lower bounds on the realised number of UR evaluations, \(N_T^{\rm UR}\). The argument first controls the spectral envelope of the \(\bar\pi_t\)-weighted covariance operator, then converts realised UR steps into decay of \(V_t\), and finally converts \(V_t\) into TV error.

Let
\(
\mathcal J_t^-
=
\mathcal H_t\vee\sigma(B_0,\ldots,B_{t-1})\),
\(\mathcal J_t^+
=
\mathcal J_t^-\vee\sigma(B_t,\para_{t+1}).
\)
By \Cref{ass:main_gp_setup}\textup{(A3)} and the fresh branch
randomisation,
\begin{equation*}
    \PP(f\mid\mathcal J_t^-)
=
\PP(f\mid\mathcal J_t^+)
=
\PP(f\mid\mathcal H_t)
=
\PP_t.
\end{equation*}
Thus the past branch labels and current pre-response input contain no
information about \(f\) beyond \(\mathcal H_t\). Moreover,
\(
\mathcal J_t^-
\subseteq
\mathcal J_t^+
\subseteq
\mathcal J_{t+1}^-,
\)
so
\(
\mathcal G_{2t}:=\mathcal J_t^-\),
\(\mathcal G_{2t+1}:=\mathcal J_t^+
\)
is an interlaced filtration.

\subsubsection{Uniform Density and Spectral Envelopes}

\begin{lemma}
\label{lem:posterior_density_envelope_integrable}
Under the inherited Bayesian continuous-GP setup, let
\(
B_\pi
:=
\sup_{\para\in\Omega}p_\dataseq(\para).
\)
Then \(B_\pi<\infty\) a.s. and \(\EE[B_\pi]<\infty\).
\end{lemma}

\begin{proof}
Since \(f\) is continuous on compact \(\Omega\), the normalising constant
\(
\mathcal Z(f)
\) from \Cref{def: expected posterior}
satisfies
\(
\mathcal Z(f) \ge |\Omega|\exp\{\inf_{\bm u\in\Omega}f(\bm u)\}
\)
a.s. Hence
\*[
    B_\pi
&=
\sup_{\para\in\Omega}\exp\{f(\para)\}/\mathcal Z(f) \\
&\le
|\Omega|^{-1}
\exp\!\left\{
\sup_{\para\in\Omega}f(\para)
-
\inf_{\para\in\Omega}f(\para)
\right\}
\le
|\Omega|^{-1}\exp(2M_f),
\]
where
\(
M_f=\sup_{\para\in\Omega}|f(\para)|.
\)
By \Cref{ass:main_gp_setup}\textup{(A5)} and the linear-exponential
moment conclusion of \Cref{lem:conditional_gp_supremum}, applied to the
prior process \(f\),
\(
\EE\{\exp(aM_f)\}<\infty,
\) \(a > 0\).
In particular, \(\EE\{\exp(2M_f)\}<\infty\). Therefore
\(
\EE[B_\pi]
\le
|\Omega|^{-1}\EE\{\exp(2M_f)\}
<\infty.
\)
The almost-sure finiteness of \(B_\pi\) follows from \(\EE\{\exp(2M_f)\}<\infty\).
\end{proof}

\begin{lemma}
\label{lem:posterior_density_domination}
Let
\(
B_{\pi,t}:=\EE_t[B_\pi].
\)
Then \(\{B_{\pi,t}\}_{t\ge0}\) is a nonnegative martingale and,
for every \(\delta\in(0,1)\),
\begin{equation*}
    \Pr\!\left(
\sup_{t\ge0}B_{\pi,t}
\le \EE[B_\pi]/\delta
\right)
\ge
1-\delta.
\end{equation*}
With
\(
\mathcal E_t^{\dataseq}(\delta)
:=
\left\{
\max_{0\le s\le t}B_{\pi,s}
\le \EE[B_\pi]/\delta
\right\},
\)
\(\mathcal E_t^{\dataseq}(\delta)\) is
\(\mathcal{H}_t\)-measurable,
\(
\mathcal E_{t+1}^{\dataseq}(\delta)
\subseteq
\mathcal E_t^{\dataseq}(\delta),
\)
\(
\Pr\{\mathcal E_t^{\dataseq}(\delta)\}\ge 1- \delta,
\)
and on \(\mathcal E_t^{\dataseq}(\delta)\),
\begin{equation*}
    \bar p_s(\para)
\le \EE[B_\pi]/\delta, \qquad \forall \para\in\Omega,\ 0\le s\le t.
\end{equation*}
\end{lemma}

\begin{proof}
By \Cref{lem:posterior_density_envelope_integrable}, \(B_\pi\) is integrable.
Therefore
\(
B_{\pi,t}=\EE_t[B_\pi]
\)
is a nonnegative martingale with respect to \(\{\mathcal H_t\}_{t\ge0}\).
Doob's maximal inequality gives, for every finite \(n\),
\begin{equation*}
    \Pr\!\left(
\max_{0\le t\le n}B_{\pi,t}
> \EE[B_\pi]/\delta
\right)
\le
\frac{\delta}{\EE[B_\pi]}\EE[B_{\pi,n}]
=
\delta.
\end{equation*}
Letting \(n\to\infty\) gives the stated horizon-uniform bound.

It remains to verify the density domination. For every \(\para\in\Omega\),
the definition of the expected posterior density and the tower property give
\begin{equation*}
    \bar p_t(\para)
=
\EE_t\!\left\{p_\dataseq(\para)\right\}
\le
\EE_t\!\left\{
\sup_{\bm u\in\Omega}p_\dataseq(\bm u)
\right\}
=
\EE_t[B_\pi]
=
B_{\pi,t}.
\end{equation*}
Thus, on \(\mathcal E_t^{\dataseq}(\delta)\),
\(
\bar p_s(\para)
\le
B_{\pi,s}
\le
\EE[B_\pi]/\delta\),
\(\forall \para\in\Omega,\ 0\le s\le t.
\)
The measurability and monotonicity of
\(\mathcal E_t^{\dataseq}(\delta)\) are immediate from its definition.
\end{proof}

\begin{lemma}
\label{lem:uniform_weighted_spectral_envelope}
Assume \Cref{ass:prior_spectral_decay}. For each \(t\), let
\(\mathcal C_t\) be the posterior covariance operator on
\(L^2(\Omega,\bar\pi_t)\), that is,
\begin{equation*}
    (\mathcal C_t h)(\para)
=
\int_\Omega \kappa_t(\para,\bm u)h(\bm u)\,\bar\pi_t(\dee\bm u),
\end{equation*}
and let
\(
\lambda_{t,1}\ge\lambda_{t,2}\ge\cdots\ge0
\)
be its eigenvalues. For every \(\delta\in(0,1)\), on
\(\mathcal E_t^{\dataseq}(\delta)\),
\begin{align*}
\lambda_{s,j}
&\le
A_\delta j^{-(1+2\nu/d)},
&&j\ge1,\quad 0\le s\le t,
&&\text{under polynomial decay},
\\
\lambda_{s,j}
&\le
A_\delta\exp\{-a_0j^{1/d}\},
&&j\ge1,\quad 0\le s\le t,
&&\text{under exponential decay},
\end{align*}
respectively, where
\(
A_\delta=A_0\EE[B_\pi]/\delta.
\)
\end{lemma}
\begin{proof}
Fix \(s\le t\), and set
\(
M_\delta=\EE[B_\pi]/\delta.
\)
On \(\mathcal E_t^{\dataseq}(\delta)\), \(\bar\pi_s\) has density bounded by
\(M_\delta\). Let \(\mathcal C_{0,s}^\pi\) be the prior covariance operator on
\(L^2(\Omega,\bar\pi_s)\):
\begin{equation*}
    (\mathcal C_{0,s}^\pi h)(\para)
=
\int_\Omega \kappa(\para,\bm u)h(\bm u)\,\bar\pi_s(\dee\bm u).
\end{equation*}
We first compare the eigenvalues of
\(\mathcal C_{0,s}^\pi\) with those of
\(\mathcal C_0^\lambda\) on \(L^2(\Omega,\lambda)\).

Let \(\mathcal H_\kappa\) be the reproducing kernel Hilbert space of the prior kernel \(\kappa\), and let
\(
I_\lambda:\mathcal H_\kappa\to L^2(\Omega,\lambda)\),
\(I_{\bar\pi_s}:\mathcal H_\kappa\to L^2(\Omega,\bar\pi_s)
\)
be the canonical embedding maps. Since
\begin{equation*}
    \|h\|_{L^2(\bar\pi_s)}^2
=
\int_\Omega h(\para)^2\bar p_s(\para)\lambda(\dee\para)
\le
M_\delta
\|h\|_{L^2(\Omega,\lambda)}^2,
\end{equation*}
we have
\(
\|I_{\bar\pi_s}h\|_{L^2(\bar\pi_s)}
\le
M_\delta^{1/2}
\|I_\lambda h\|_{L^2(\Omega,\lambda)}
\)
for every \(h\in\mathcal H_\kappa\). Hence
\begin{equation*}
    a_j(I_{\bar\pi_s})
\le
M_\delta^{1/2}a_j(I_\lambda),
\ j\ge1,
\end{equation*} where \(a_j(\cdot)\) denotes the \(j\)-th approximation number of a compact operator.
The nonzero eigenvalues of the covariance operators satisfy
\(
\lambda_{0,s,j}^\pi
=
a_j(I_{\bar\pi_s})^2\),
\(\lambda_{0,j}^\lambda
=
a_j(I_\lambda)^2,
\)
because
\(
\mathcal C_{0,s}^\pi=I_{\bar\pi_s}I_{\bar\pi_s}^\ast
\)
and
\(
\mathcal C_0^\lambda=I_\lambda I_\lambda^\ast
\).
Therefore
\(
    \lambda_{0,s,j}^\pi
\le
M_\delta \lambda_{0,j}^\lambda.
\)

It remains to pass from the prior covariance operator to the posterior covariance operator.
Exact GP regression gives posterior covariance domination in the Loewner sense:
\(
0\preceq \kappa_s \preceq \kappa,
\)
meaning that for every finite signed measure \(\rho\) on \(\Omega\),
\begin{equation*}
    \iint_{\Omega^2}
\kappa_s(\para,\bm u)\,\rho(\dee\para)\rho(\dee\bm u)
\le
\iint_{\Omega^2}
\kappa(\para,\bm u)\,\rho(\dee\para)\rho(\dee\bm u).
\end{equation*}
Equivalently, for every \(h\in L^2(\Omega,\bar\pi_s)\),
\(
\langle h,\mathcal C_s h\rangle_{L^2(\bar\pi_s)}
\le
\langle h,\mathcal C_{0,s}^\pi h\rangle_{L^2(\bar\pi_s)}.
\)
The Courant--Fischer min--max principle for compact positive self-adjoint
operators on \(L^2(\Omega,\bar\pi_s)\) then yields
\(
\lambda_{s,j}\le \lambda_{0,s,j}^\pi
\)
for every \(j\ge1\). Combining the two bounds gives
\begin{align*}
\lambda_{s,j}
&\le
A_\delta j^{-(1+2\nu/d)}
&&\text{under polynomial decay},
\\
\lambda_{s,j}
&\le
A_\delta\exp\{-a_0j^{1/d}\}
&&\text{under exponential decay}.
\end{align*}
Since \(s\le t\) was arbitrary, the claim holds uniformly for all
\(0\le s\le t\) on \(\mathcal E_t^{\dataseq}(\delta)\).
\end{proof}

\subsubsection{Weighted Variance Decay under Realised UR Steps}

\begin{proof}[Proof of the first assertion of
\Cref{prop:ur_calibration_link_main}]
Fix \(\para\in\Omega\), and use a superscript \((\para)\) for the
post-update quantities in the hypothetical experiment that fixes the next
query at \(\para\) without observing its response. Write
\(\bar\pi_{t+1}^{(\para)}\) for the corresponding expected posterior.

Given \(\mathcal H_t\), the GP variance update is deterministic:
\begin{equation}\label{eq:sigma_t+1_reduction}
\sigma_{t+1}^{2,(\para)}(\bm u)
=
\sigma_t^2(\bm u)
-
\frac{\kappa_t(\bm u,\para)^2}
{\sigma_t^2(\para)+\varsigma^2}.
\end{equation}
When \(\varsigma^2=0\) and \(\sigma_t^2(\para)=0\), the reduction term in
\Cref{eq:sigma_t+1_reduction} is interpreted according to the zero-variance convention
in \Cref{def:weighted_variance_reduction}, and is therefore zero.

For any bounded \(\mathcal H_t\)-measurable \(h\), the tower
property and the definition of \(\bar\pi_t\) give
\begin{equation*}
    \EE\!\left\{
\int_\Omega
h(\bm u)\,\bar\pi_{t+1}^{(\para)}(\dee\bm u)
\,\middle|\,
\mathcal H_t
\right\}
=
\EE\!\left\{
\int_\Omega
h(\bm u)\,\pi_\dataseq(\dee\bm u)
\,\middle|\,
\mathcal H_t
\right\}
=
\int_\Omega h(\bm u)\,\bar\pi_t(\dee\bm u).
\end{equation*}
Taking
\(
h(\bm u)
=
\sigma_{t+1}^{2,(\para)}(\bm u),
\)
which is \(\mathcal H_t\)-measurable and does not depend on the realised
response, yields
\begin{equation}
\label{eq:delta_variance_reduction_identity}
\EE\!\left\{
V_{t+1}^{(\para)}
\,\middle|\,
\mathcal H_t
\right\}
=
\int_\Omega
\sigma_{t+1}^{2,(\para)}(\bm u)\,
\bar\pi_t(\dee\bm u)
=
V_t-\Delta_t(\para).
\end{equation}

Equivalently, if \(\para_{t+1}\) is generated before its response and is
conditionally independent of \(f\) given \(\mathcal H_t\), evaluating \Cref{eq:delta_variance_reduction_identity} at \(\para_{t+1}\) gives
\(
    \EE\!\left[
V_{t+1}
\,\middle|\,
\mathcal H_t,\para_{t+1}
\right]
=
V_t-\Delta_t(\para_{t+1})\) a.s.
\end{proof}

\begin{lemma}
\label{lem:delta_lower_bound_by_weighted_variance}
Assume \Cref{ass:prior_spectral_decay}. Fix
\(\delta_\pi\in(0,1)\) and a fixed noise variance
\(\varsigma^2\geq0\). There exists a deterministic constant
\(c_{\Delta,\delta_\pi}>0\), and, under exponential spectral decay, a
deterministic constant \(C_{L,\delta_\pi}\geq B_\kappa\), such that, for
every \(t\geq0\), on \(\mathcal E_t^{\dataseq}(\delta_\pi)\), the following
applicable bounds hold. These constants may depend on the fixed
\(\varsigma^2\), but not on \(t\). Write
\(
L_{\delta_\pi}(v)
:=
1+\log\!\left(C_{L,\delta_\pi}/v\right)\), \(0<v\leq B_\kappa\),
and define expressions with denominator \(L_{\delta_\pi}(V_t)^d\) to be
zero when \(V_t=0\).
For the \emph{general case} of $\varsigma^2 \ge 0$, take \(a=2\). If \(\varsigma^2=0\), the
\emph{noiseless refinement} is obtained by taking \(a=1\). For either
applicable value of \(a\),
\begin{equation*}
\sup_{\para\in\Omega}\Delta_t(\para)
\ge
\begin{dcases}
c_{\Delta,\delta_\pi}V_t^{a + d/(2\nu)},
& \text{under polynomial spectral decay},\\
c_{\Delta,\delta_\pi}V_t^a L_{\delta_\pi}(V_t)^{-d},
& \text{under exponential spectral decay}.
\end{dcases}
\end{equation*}
\end{lemma}

\begin{proof}
If \(V_t=0\), all four claims are trivial. Suppose \(V_t>0\). Since
\(\bar\pi_t\) is a probability measure and
\(\sigma_t^2(\bm v)\leq B_\kappa\),
\[\label{eq:general_averaging_bound}
 \sup_{\para\in\Omega}\Delta_t(\para)
&\ge
\int_\Omega\Delta_t(\bm v)\,\bar\pi_t(\dee\bm v) \\
& \ge
\frac{1}{B_\kappa+\varsigma^2}
\iint_{\Omega^2}
\kappa_t(\bm u,\bm v)^2
\,\bar\pi_t(\dee\bm u)\bar\pi_t(\dee\bm v)
=
\frac{1}{B_\kappa+\varsigma^2}
\sum_{j\ge1}\lambda_{t,j}^2.
\]
If \(\varsigma^2=0\), define the variance-tilted probability measure
\(
    Q_t(\dee\bm v)
:=\sigma_t^2(\bm v)\bar\pi_t(\dee\bm v)/V_t.
\)
When \(\sigma_t^2(\bm v)=0\), positive semidefiniteness gives
\(\kappa_t(\bm u,\bm v)=0\) for every \(\bm u\), so the zero-variance
convention in \Cref{def:weighted_variance_reduction} is compatible with
the cancellation below. Hence
\begin{equation}\label{eq:noiseless_averaging_bound}
    \sup_{\para\in\Omega}\Delta_t(\para)
\ge
\int_\Omega\Delta_t(\bm v)\,Q_t(\dee\bm v)
=
\frac{1}{V_t}
\iint_{\Omega^2}
\kappa_t(\bm u,\bm v)^2
\,\bar\pi_t(\dee\bm u)\bar\pi_t(\dee\bm v)
=
\frac{1}{V_t}\sum_{j\ge1}\lambda_{t,j}^2.
\end{equation}
In either case,
\begin{equation*}
V_t
=
\int_\Omega \kappa_t(\bm u,\bm u)\,\bar\pi_t(\dee\bm u)
=
\sum_{j\ge1}\lambda_{t,j}.
\end{equation*}

On \(\mathcal E_t^{\dataseq}(\delta_\pi)\),
\Cref{lem:uniform_weighted_spectral_envelope} and an integral comparison
give constants \(C_{A,\delta_\pi}<\infty\) and \(b_0>0\) such that, for
\(m\geq1\),
\begin{equation*}
\sum_{j>m}\lambda_{t,j}
\leq
\begin{cases}
C_{A,\delta_\pi}m^{-2\nu/d},
&\text{under polynomial spectral decay},\\
C_{A,\delta_\pi}\exp\{-b_0m^{1/d}\},
&\text{under exponential spectral decay}.
\end{cases}
\end{equation*}
Choose \(C_{L,\delta_\pi}\geq B_\kappa\) and \(C_m<\infty\) sufficiently
large, and set
\begin{equation*}
m_t
:=
\begin{cases}
\left\lceil C_mV_t^{-d/(2\nu)}\right\rceil,
&\text{under polynomial spectral decay},\\
\left\lceil C_mL_{\delta_\pi}(V_t)^d\right\rceil,
&\text{under exponential spectral decay}.
\end{cases}
\end{equation*}
The constants can be chosen so that
\(
\sum_{j>m_t}\lambda_{t,j}\leq V_t/2
\),
while
\begin{equation*}
m_t
\leq
C_{\delta_\pi}
\begin{cases}
V_t^{-d/(2\nu)},
&\text{under polynomial spectral decay},\\
L_{\delta_\pi}(V_t)^d,
&\text{under exponential spectral decay}.
\end{cases}
\end{equation*}
Therefore, Cauchy's inequality gives
\begin{equation*}
\sum_{j\geq1}\lambda_{t,j}^2
\geq
\frac{1}{m_t}
\left(\sum_{j=1}^{m_t}\lambda_{t,j}\right)^2
\geq
c_{\delta_\pi}
\begin{cases}
V_t^{2+d/(2\nu)},
&\text{under polynomial spectral decay},\\
V_t^2/L_{\delta_\pi}(V_t)^d,
&\text{under exponential spectral decay}.
\end{cases}
\end{equation*}
Absorbing all constants as $c_{\Delta,\delta_\pi}$, and combining the above with \Cref{eq:general_averaging_bound} proves the two general
bounds, while combining it with \Cref{eq:noiseless_averaging_bound} proves the two
noiseless refinements.
\end{proof}
\begin{lemma}
\label{lem:weighted_variance_rate_by_ur_count}
Assume the Mat\'ern/SE kernel alternative in
\Cref{thm:sale_main_calibration}. Fix
\(\delta_\pi\in(0,1)\) and a fixed noise variance
\(\varsigma^2\geq0\). Suppose that whenever \(I_t^{\rm UR}=1\), the
idealised UR rule chooses
\(
\para_{t+1}\in\arg\max_{\para\in\Omega}\Delta_t(\para).
\)
At non-UR iterations, the selected input may be chosen by any rule taking
values in \(\Omega\), provided it is generated before observing \(y_{t+1}\)
and is conditionally independent of the latent objective \(f\) given
\(\mathcal H_t\). Then there exist \(C_{\delta_\pi}<\infty\) and
\(c_{\delta_\pi}>0\), possibly depending on the fixed
\(\varsigma^2\) but independent of \(T\), such that, for every
\(T\geq1\) and every deterministic
\(n_T\in[T]_0\), the following bounds hold, with the second
being the sharper refinement available when \(\varsigma^2=0\):
\begin{equation*}
\EE\!\left\{
V_T\mathbf 1_{\mathcal E_T^{\dataseq}(\delta_\pi)}
\right\}
\le
\begin{cases}
C_{\delta_\pi}\{\mathscr O^{+}(n_T)\}^2 +
B_\kappa\Pr\left(N_T^{\rm UR}<n_T\right),\\[1mm]
C_{\delta_\pi}\{\mathscr O^{0}_{c_{\delta_\pi}/2}(n_T)\}^2+
B_\kappa\Pr\left(N_T^{\rm UR}<n_T\right),
& \text{if \(\varsigma^2=0\)},
\end{cases}
\end{equation*}
$\mathscr O^{+}$ and $\mathscr O^{0}_{c}$ are as defined in \Cref{thm:sale_main_calibration}.

\end{lemma}

\begin{proof}
Under Mat\'ern-$\nu$ kernel, set
\(g(x)=x^p\), where
\begin{equation*}
p=
\begin{cases}
2+d/(2\nu),
&\text{for any fixed \(\varsigma^2\geq0\)},\\
1+d/(2\nu),
&\text{if \(\varsigma^2=0\) and the noiseless refinement is used}.
\end{cases}
\end{equation*}
Under SE kernel, let
\(L_{\delta_\pi}\) be as in
\Cref{lem:delta_lower_bound_by_weighted_variance}, set \(g(0)=0\), and
define
\begin{equation*}
g(x)
:=
\frac{x^a}{L_{\delta_\pi}(x)^d},
\qquad
a=
\begin{cases}
2,&\text{in the general case},\\
1,&\text{for the noiseless refinement},
\end{cases}
\qquad 0<x\leq B_\kappa.
\end{equation*}
In all four cases, \(g\) is nonnegative, nondecreasing, and convex on
\([0,B_\kappa]\). For the SE case, this follows, for \(x>0\),
from
\begin{align*}
g'(x)
&=
x^{a-1}
L_{\delta_\pi}(x)^{-d}
\left\{a+d/L_{\delta_\pi}(x)\right\},
\\
g''(x)
&=
x^{a-2}L_{\delta_\pi}(x)^{-d}
\left\{
a(a-1)
+d(2a-1)/L_{\delta_\pi}(x)
+d(d+1)/L_{\delta_\pi}(x)^2
\right\}
>0,
\end{align*}
and the continuous extension at zero.
The events \(\mathcal E_t^{\dataseq}(\delta_\pi)\) are decreasing and
\(\mathcal H_t\)-measurable. Let
\(
W_t
=
V_t\mathbf 1_{\mathcal E_t^{\dataseq}(\delta_\pi)}
\).
Then \(W_t \in [0,B_\kappa]\). Since
\(
\mathcal E_{t+1}^{\dataseq}(\delta_\pi)
\subseteq
\mathcal E_t^{\dataseq}(\delta_\pi)
\),
the posterior-ignorability identities and
\Cref{eq:delta_variance_reduction_identity} imply, for any allowed non-UR input, that
\(
\EE[W_{t+1}\mid\mathcal J_t^+]\leq W_t
\)
on \(\{I_t^{\rm UR}=0\}\).
On \(\mathcal E_t^{\dataseq}(\delta_\pi)\), the ideal UR choice,
\Cref{eq:delta_variance_reduction_identity,lem:delta_lower_bound_by_weighted_variance}
give
\begin{equation*}
\EE[V_{t+1}\mid\mathcal J_t^+]
\le
V_t-c_{\Delta,\delta_\pi}g(V_t)
\end{equation*} on \(\{I_t^{\rm UR}=1\}.\)
Combining the two cases gives
\begin{equation}
\label{eq:W_t_drift_augmented}
\EE[W_{t+1}\mid\mathcal J_t^+]
\le
W_t-I_t^{\rm UR}c_{\Delta,\delta_\pi}g(W_t).
\end{equation}

Choose \(c_0\in(0,c_{\Delta,\delta_\pi}]\) so that
\(
c_0\sup_{x\in[0,B_\kappa]}g'(x)\le 1,
\)
and define
\(\phi(x):=x-c_0g(x)\), \(x\in [0,B_\kappa]\).
Then \(\phi\) is nondecreasing and concave on \([0,B_\kappa]\).
\Cref{eq:W_t_drift_augmented} implies that a non-UR step cannot increase
\(W_t\) in conditional expectation, while a UR step gives the conditional
drift \(W_t-c_0g(W_t)=\phi(W_t)\).

Define
\(
X_{2t}=X_{2t+1}:=W_t, \ t\ge0.
\)
By \Cref{eq:W_t_drift_augmented},
\(\{X_k,\mathcal G_k\}_{k\ge0}\) is a bounded nonnegative
supermartingale. At an odd index which is a UR decision,
\(
    \EE[X_{2t+2}\mid\mathcal G_{2t+1}]
\le
\phi(X_{2t+1}).
\)

Define the \(n\)th UR decision time and its completion time by
\begin{align*}
\zeta_n^{\rm UR}
&:=
\inf\left\{
t\ge0:
\sum_{s=0}^{t}I_s^{\rm UR}\ge n
\right\},
\qquad n\ge1,
\\
\rho_0^{\rm UR}
&:=0,
\qquad
\rho_n^{\rm UR}
:=
\inf\left\{
t\ge1:
\sum_{s=0}^{t-1}I_s^{\rm UR}\ge n
\right\},
\qquad n\ge1.
\end{align*}
Thus, on \(\{\zeta_n^{\rm UR}<\infty\}\),
\(\rho_n^{\rm UR}=\zeta_n^{\rm UR}+1\). Under the interlaced filtration, for $n \ge 0$,
\begin{equation*}
\mathsf S_n^{\rm UR}
:=
2\rho_n^{\rm UR},
\qquad
\mathsf U_{n+1}^{\rm UR}
:=
2\zeta_{n+1}^{\rm UR}+1,
\end{equation*}
are stopping times. On
\(\{\mathsf U_{n+1}^{\rm UR}<\infty\}\),
\(
\mathsf S_{n+1}^{\rm UR}
=
\mathsf U_{n+1}^{\rm UR}+1.
\)
Set
\begin{equation*}
       Z_n^{\rm UR}
:=
X_{\mathsf S_n^{\rm UR}}
\mathbf 1_{\{\mathsf S_n^{\rm UR}<\infty\}}
=
W_{\rho_n^{\rm UR}}
\mathbf 1_{\{\rho_n^{\rm UR}<\infty\}},
\end{equation*}
and
\(
\widetilde X_{\mathsf U_{n+1}^{\rm UR}}
:=
X_{\mathsf U_{n+1}^{\rm UR}}
\mathbf 1_{\{\mathsf U_{n+1}^{\rm UR}<\infty\}}.
\)
All these variables are set to zero on the corresponding infinite-time
events.

Because \(\{X_k\}\) is bounded, it is uniformly integrable and has an
\(L^1\) terminal value \(X_\infty\). Optional sampling for the possibly
infinite stopping times
\(\mathsf S_n^{\rm UR}\leq\mathsf U_{n+1}^{\rm UR}\) gives
\begin{equation*}
    \EE\!\left[
X_{\mathsf U_{n+1}^{\rm UR}}
\mathbf 1_{\{\mathsf U_{n+1}^{\rm UR}<\infty\}}
+
X_\infty
\mathbf 1_{\{\mathsf U_{n+1}^{\rm UR}=\infty\}}
\,\middle|\,
\mathcal G_{\mathsf S_n^{\rm UR}}
\right]
\leq
X_{\mathsf S_n^{\rm UR}}
\mathbf 1_{\{\mathsf S_n^{\rm UR}<\infty\}}
+
X_\infty
\mathbf 1_{\{\mathsf S_n^{\rm UR}=\infty\}}.
\end{equation*}
On \(\{\mathsf S_n^{\rm UR}<\infty\}\), drop the nonnegative terminal-value
term on the left. On its complement,
\(\mathsf U_{n+1}^{\rm UR}=\infty\), so both truncated variables below
vanish. Consequently,
\begin{equation*}
    \EE\!\left[
\widetilde X_{\mathsf U_{n+1}^{\rm UR}}
\,\middle|\,
\mathcal G_{\mathsf S_n^{\rm UR}}
\right]
\leq
Z_n^{\rm UR}.
\end{equation*}
At \(\mathsf U_{n+1}^{\rm UR}\), the UR drift and the identity
\(\mathsf S_{n+1}^{\rm UR}=\mathsf U_{n+1}^{\rm UR}+1\) give
\begin{equation*}
    \EE\!\left[
Z_{n+1}^{\rm UR}
\,\middle|\,
\mathcal G_{\mathsf S_n^{\rm UR}}
\right]
\leq
\EE\!\left[
\phi(\widetilde X_{\mathsf U_{n+1}^{\rm UR}})
\,\middle|\,
\mathcal G_{\mathsf S_n^{\rm UR}}
\right].
\end{equation*}
Since \(\phi\) is nondecreasing and concave,
\begin{equation*}
    \EE\!\left[
Z_{n+1}^{\rm UR}
\,\middle|\,
\mathcal G_{\mathsf S_n^{\rm UR}}
\right]
\leq
\phi\!\left(
\EE\!\left[
\widetilde X_{\mathsf U_{n+1}^{\rm UR}}
\,\middle|\,
\mathcal G_{\mathsf S_n^{\rm UR}}
\right]
\right)
\leq
\phi(Z_n^{\rm UR}).
\end{equation*}
Let
\(
z_n^{\rm UR}:=\EE[Z_n^{\rm UR}].
\)
Taking expectations and using the concavity of \(\phi\) gives
\begin{equation*}
z_{n+1}^{\rm UR}
\leq
\EE\{\phi(Z_n^{\rm UR})\}
\le
\phi(z_n^{\rm UR})
=
z_n^{\rm UR}-c_0g(z_n^{\rm UR}).
\end{equation*}
Under Mat\'ern-$\nu$ kernel, the standard deterministic recursion comparison
therefore gives
\begin{equation*}
z_n^{\rm UR}
\leq
C_{\delta_\pi}(n+1)^{-1/(p-1)}
=
\begin{cases}
C_{\delta_\pi}(n+1)^{-2\nu/(2\nu+d)},
&\text{for any fixed \(\varsigma^2\geq0\)},\\
C_{\delta_\pi}(n+1)^{-2\nu/d},
&\text{for the noiseless refinement}.
\end{cases}
\end{equation*}
Under SE kernel, the claim is immediate if some
\(z_n^{\rm UR}=0\). Otherwise, write
\(
\ell_n=L_{\delta_\pi}(z_n^{\rm UR})
\).
After reducing \(c_0\), if necessary, so that
\(c_0\max\{1,B_\kappa\}<1\), consider first the general case. The recursion
then implies
\begin{equation*}
z_{n+1}^{\rm UR}
\le
z_n^{\rm UR}
\left\{1-c_0z_n^{\rm UR}\ell_n^{-d}\right\}.
\end{equation*}
Define
\(
h(x):=L_{\delta_\pi}(x)^d/x
\)
for \(x\in(0,B_\kappa]\). Since \(h\) is decreasing and
\(L_{\delta_\pi}(z_{n+1}^{\rm UR})\geq\ell_n\),
\begin{equation*}
h(z_{n+1}^{\rm UR})\geq
\frac{\ell_n^d}
{z_n^{\rm UR}\{1-c_0z_n^{\rm UR}\ell_n^{-d}\}}\geq
h(z_n^{\rm UR})+c_0,
\end{equation*}
where the last step uses \((1-u)^{-1}\geq1+u\).
Thus \(h(z_n^{\rm UR})\geq h(z_0^{\rm UR})+c_0n\). Moreover, let
\(
a_n
:=
C_{\delta_\pi}\{1+\log(n+1)\}^d/(n+1).
\)
For all sufficiently large \(n\), \(a_n\leq B_\kappa\), and
\begin{align*}
L_{\delta_\pi}(a_n)
&=
1+\log\!\left(
\frac{C_{L,\delta_\pi}(n+1)}
{C_{\delta_\pi}\{1+\log(n+1)\}^d}
\right)
\leq
C\{1+\log(n+1)\}.
\end{align*}
Consequently,
\begin{equation*}
    h(a_n)
= L_{\delta_\pi}(a_n)^d/a_n
\leq C_{\delta_\pi}^{-1}C^d(n+1)
\leq c_0n,
\end{equation*}
for all sufficiently large \(n\), after enlarging
\(C_{\delta_\pi}\) if necessary.
Monotonicity of \(h\), followed by an enlargement of
\(C_{\delta_\pi}\) to cover the remaining \(n\), therefore gives
\begin{equation*}
z_n^{\rm UR}
\leq
C_{\delta_\pi}
\frac{\{1+\log(n+1)\}^d}{n+1}.
\end{equation*}

For the noiseless exponential case, the recursion instead gives
\(
z_{n+1}^{\rm UR}
\leq
z_n^{\rm UR}(1-c_0\ell_n^{-d}),
\)
and hence
\begin{equation*}
    \ell_{n+1}
\ge
\ell_n-\log(1-c_0\ell_n^{-d})
\ge
\ell_n+c_0\ell_n^{-d}.
\end{equation*}
Therefore
\(
\ell_{n+1}^{d+1}
\ge
\ell_n^{d+1}+(d+1)c_0,
\)
so
\begin{equation*}
z_n^{\rm UR}
\le
C_{\delta_\pi}
\exp\{-c_{\delta_\pi}(n+1)^{1/(d+1)}\}.
\end{equation*}
For the remainder of the proof, write the applicable bound as
\begin{equation*}
r_{\delta_\pi}(n)
:=
\begin{cases}
C_{\delta_\pi}\{\mathscr O^{+}(n)\}^2,
&\text{for the general bound},\\
C_{\delta_\pi}\{\mathscr O^{0}_{c_{\delta_\pi}/2}(n)\}^2,
&\text{for the noiseless refinement}.
\end{cases}
\end{equation*}

Now fix \(0\leq n_T\leq T\), and let
\(
A_T
:=
\{N_T^{\rm UR}\geq n_T\}
=
\{\rho_{n_T}^{\rm UR}\leq T\}
=
\{\mathsf S_{n_T}^{\rm UR}\leq2T\}.
\)
Set
\(
\mathsf S'_T
:=
\mathsf S_{n_T}^{\rm UR}\wedge2T.
\)
Then \(A_T\in\mathcal G_{\mathsf S'_T}\), so bounded optional sampling gives
\begin{align*}
\EE[W_T\mathbf 1_{A_T}]
&=
\EE[X_{2T}\mathbf 1_{A_T}]
\leq
\EE[X_{\mathsf S'_T}\mathbf 1_{A_T}]
\\
&=
\EE[
X_{\mathsf S_{n_T}^{\rm UR}}\mathbf 1_{A_T}
]
\leq
\EE[Z_{n_T}^{\rm UR}]
\leq
r_{\delta_\pi}(n_T).
\end{align*}
On \(A_T^c\), \(W_T\le B_\kappa\). Hence
\(
    \EE[W_T]
\le
r_{\delta_\pi}(n_T)
+
B_\kappa\Pr\left(N_T^{\rm UR}<n_T\right),
\)
which proves the claim.
\end{proof}
\subsubsection{From Weighted Variance to TV Calibration}

\begin{lemma}[Generalisation of the second assertion of \Cref{prop:ur_calibration_link_main}]
\label{lem:tv_from_weighted_variance}
Under the inherited ideal Bayesian GP setup, for any
\(\mathcal H_t\)-measurable event \(\mathcal E\),
\begin{equation*}
    \EE\!\left\{\TV(\bar\pi_t,\pi_\dataseq)
\mathbf 1_{\mathcal E}
\right\}
\le
\left[
\frac12\EE\{V_t\mathbf 1_{\mathcal E}\}
\right]^{1/2}.
\end{equation*}
Consequently,
\(
\EE\!\left\{\TV(\bar\pi_t,\pi_\dataseq)
\right\}
\le
\left[\EE\{V_t\}/2
\right]^{1/2}.
\)
\end{lemma}

\begin{proof}
Condition on \(\mathcal H_t\). Let
\(
f=\mu_t+G,
\)
\(
g=\mu_t+G',
\)
where \(G\) and \(G'\) are independent centred posterior GP residuals with
common covariance \(\kappa_t\). Using
\Cref{def: expected posterior}, write
\(
\pi_f:=\pi(\cdot\mid f)
\)
and
\(
\pi_g:=\pi(\cdot\mid g).
\)
Conditionally on \(\mathcal H_t\), \(\pi_\dataseq\) has the same law as
\(\pi_f\), while
\(
\bar\pi_t
=
\EE_t[\pi_g].
\)
By convexity of TV, Pinsker's inequality, and Jensen's inequality,
\begin{equation*}
    \EE_t\!\left\{\TV(\bar\pi_t,\pi_f)
\right\}
\le
\EE_t\!\left\{
\TV(\pi_g,\pi_f)
\right\}
\le
\left[
\frac12
\EE_t\!\left\{
\mathrm{KL}(\pi_f\|\pi_g)
\right\}
\right]^{1/2}.
\end{equation*}

For the KL term,
\begin{equation*}
\mathrm{KL}(\pi_f\|\pi_g)
=
\int_\Omega \{f(\para)-g(\para)\}\,\pi_f(\dee\para)
+
\log\mathcal Z(g)-\log\mathcal Z(f).
\end{equation*}
Taking \(\EE_t\) over the conditionally independent pair \((G,G')\), the
normalising-constant terms cancel because \(G\) and \(G'\) have the same
conditional law. The term involving \(G'\) vanishes because \(G'\) is
conditionally independent of \(G\) and \(\EE_t[G'(\para)]=0\). Thus,
\begin{equation}\label{eq:E_t_KL_fg}
    \EE_t\!\left[
\mathrm{KL}(\pi_f\|\pi_g)
\right]
=
\EE_t\!\left\{
\int_\Omega
G(\para)\,
\pi(\dee\para\mid\mu_t+G)
\right\},
\end{equation}
where, \(\pi(\dee\para\mid\mu_t+G)\) is defined according to \Cref{def: expected posterior}.

It remains to bound \Cref{eq:E_t_KL_fg}. Define the log-partition
functional
\begin{equation*}
    \Psi(h)
    =
    \log\int_\Omega \exp\{\mu_t(\para)+h(\para)\}\lambda(\dee\para),
    \qquad h\in C(\Omega).
\end{equation*}
For a deterministic direction \(a\in C(\Omega)\), its first
directional derivative is
\begin{equation*}
    D\Psi_h[a]
    =
    \int_\Omega a(\para)\pi(\dee\para\mid \mu_t+h).
\end{equation*}
Hence
\begin{equation*}
    \int_\Omega G(\para)\pi(\dee\para\mid \mu_t+G)
    =
    D\Psi_G[G].
\end{equation*}

Let
\(\mathcal Q_n=\{C_{n,1},\ldots,C_{n,N_n}\}\) be finite Borel partitions of \(\Omega\) with
\(
\max_{1\le i\le N_n}\operatorname{diam}(C_{n,i})\to0.
\) Choose
\(\bm \xi_{n,i}\in C_{n,i}\), and set
\(
    a_{n,i}=\lambda(C_{n,i})\), 
    \(
     Z_{n,i}=G(\bm \xi_{n,i})\), \(
    \mu_{n,i}=\mu_t(\bm\xi_{n,i})\), \(
    \Sigma_{n,ij}=\kappa_t(\bm\xi_{n,i},\bm\xi_{n,j}).
\)
Then \(\bm Z_n=(Z_{n,1},\ldots,Z_{n,N_n})^\top\) is conditionally Gaussian
given \(\mathcal H_t\), with mean zero and covariance matrix \(\Sigma_n\).
Define the finite-dimensional log-partition function
\begin{equation*}
     \Psi_n(\bm x)
    =
    \log\left\{
        \sum_{i=1}^{N_n}
        a_{n,i}\exp(\mu_{n,i}+x_i)
    \right\},
    \qquad \bm x = (x_1, \dots, x_{N_n}) \in\mathbb R^{N_n},
\end{equation*}
and its corresponding softmax weights
\begin{equation*}
    w_{n,i}(\bm x)
    =
        a_{n,i}\exp(\mu_{n,i}+x_i)/
        \sum_{\ell=1}^{N_n}
        a_{n,\ell}\exp(\mu_{n,\ell}+x_\ell).
\end{equation*}
Then
\(
    \partial_i\Psi_n(\bm x)=w_{n,i}(\bm x),
\)
and
\(
    \partial_{ij}\Psi_n(\bm x)
    =
    w_{n,i}(\bm x)\{\mathbf 1(i=j)-w_{n,j}(\bm x)\}.
\)
The finite-dimensional Gaussian integration-by-parts identity gives
\begin{equation*}
    \EE_t\left\{
        \sum_{i=1}^{N_n}Z_{n,i}\partial_i\Psi_n(\bm Z_n)
    \right\}
    =
   \EE_t\left\{
        \sum_{i,j=1}^{N_n}
        \Sigma_{n,ij}\partial_{ij}\Psi_n(\bm Z_n)
    \right\}.
\end{equation*}
Equivalently,
\[\label{eq:Et-Zn-wni}   
\EE_t\left\{
        \sum_{i=1}^{N_n}Z_{n,i}w_{n,i}(\bm Z_n)
    \right\}
    &=
   \EE_t\left\{
        \sum_{i=1}^{N_n}
        \kappa_t(\bm\xi_{n,i},\bm\xi_{n,i})w_{n,i}(\bm Z_n)
    \right\} \\
    &\quad -
   \EE_t\left\{
        \sum_{i,j=1}^{N_n}
        \kappa_t(\bm\xi_{n,i},\bm\xi_{n,j})
        w_{n,i}(\bm Z_n)w_{n,j}(\bm Z_n)
    \right\}.
\]
This is just Stein's identity in finite dimensions:
\(\EE\{X^\top\nabla\Psi_n(X)\}
=
\EE[\operatorname{tr}\{\Sigma_n\nabla^2\Psi_n(X)\}]\).

Let \(G_n\) and \(\mu_{t,n}\) be the piecewise-constant
approximations of \(G\) and \(\mu_t\) that take values
\(G(\bm\xi_{n,i})\) and \(\mu_t(\bm\xi_{n,i})\) on \(C_{n,i}\),
respectively. By continuity of \(G\), \(\mu_t\), and \(\kappa_t\) on
compact \(\Omega\), \(G_n\to G\) and \(\mu_{t,n}\to\mu_t\) uniformly
almost surely, and the corresponding Gibbs laws converge to
\(\pi(\cdot\mid\mu_t+G)\). Moreover, the left-hand side of
\Cref{eq:Et-Zn-wni} is dominated by \(\sup_{\para\in\Omega}|G(\para)|\), which has finite conditional first moment by posterior covariance contraction and
\Cref{lem:conditional_gp_supremum}; the covariance terms are bounded by
\(B_\kappa\). Dominated convergence therefore yields
\*[   
\EE_t\{D\Psi_G[G]\}
    &=
   \EE_t\left\{
        \int_\Omega
        \kappa_t(\para,\para)
        \pi(\dee\para\mid\mu_t+G)
    \right\} \\
    &\quad -
   \EE_t\left\{
        \iint_{\Omega^2}
        \kappa_t(\para,\bm u)
        \pi(\dee\para\mid\mu_t+G)
        \pi(\dee \bm u\mid\mu_t+G)
    \right\}.
\]
The second term is nonnegative by positive definiteness of
\(\kappa_t\). Therefore
\*[
   \EE_t\{\mathrm{KL}(\pi_f\|\pi_g)\} &= \EE_t\{D\Psi_G[G]\} \\
    &\le
   \EE_t\left\{
        \int_\Omega
        \kappa_t(\para,\para)
        \pi(\dee\para\mid\mu_t+G)
    \right\} = \int_\Omega \sigma_t^2(\para)\bar\pi_t(\dee\para)
    =
    V_t.
\]
Thus
\(
\EE_t\!\left\{
\TV(\bar\pi_t,\pi_\dataseq)
\right\}
\le (V_t/2)^{1/2}.
\)
Since \(\mathcal E\in\mathcal H_t\), multiplying by
\(\mathbf 1_{\mathcal E}\), taking expectations, and applying
Cauchy-Schwarz yields
\begin{equation*}
    \EE\!\left\{\TV(\bar\pi_t,\pi_\dataseq)
\mathbf 1_{\mathcal E}
\right\}
\le
\EE\!\left\{
\left(V_t/2\right)^{1/2}
\mathbf 1_{\mathcal E}
\right\}
\le
\left\{
\frac12\EE[V_t\mathbf 1_{\mathcal E}]
\right\}^{1/2}.
\end{equation*}
Taking \(\mathcal E\) to be the whole sample space gives the second claim.
\end{proof}

\begin{proof}[\Cref{thm:sale_main_calibration}]
Fix \(\delta,q\in(0,1)\), set
\(
\delta_\pi=\delta/2,\)
\(\mathcal E_T
=
\mathcal E_T^{\dataseq}(\delta_\pi),
\)
\(n_T=\lfloor qT\rfloor .
\)
By \Cref{lem:posterior_density_domination,lem:tv_from_weighted_variance},
\begin{equation*}
\EE\left\{\TV(\bar\pi_T,\pi_\dataseq)\right\}
\le
\left\{
\frac12
\EE[V_T\mathbf 1_{\mathcal E_T}]
\right\}^{1/2}
+
\delta_\pi.
\end{equation*}
Since
\(
n_T+1\ge qT\ge q(T+1)/2
\)
and
\(
\Pr(N_T^{\rm UR}<n_T)
\le
\Pr(N_T^{\rm UR}<qT),
\)
applying \Cref{lem:weighted_variance_rate_by_ur_count} and
\((x+y)^{1/2}\le x^{1/2}+y^{1/2}\) gives
\begin{equation*}
\EE\left\{\TV(\bar\pi_T,\pi_\dataseq)\right\}
\le
C_{\delta,q}\mathscr O^{+}(T)
+
C_{\delta,q}
\Pr(N_T^{\rm UR}<qT)^{1/2}
+
\delta.
\end{equation*}
For the SE branch, this also uses
\(
1+\log(n_T+1)\leq1+\log(T+1)
\).

If \(\varsigma^2=0\), the noiseless bound in
\Cref{lem:weighted_variance_rate_by_ur_count} similarly gives
\begin{equation*}
\EE\left\{\TV(\bar\pi_T,\pi_\dataseq)\right\}
\le
C_{\delta,q}\mathscr O^{0}_{c_{\delta,q}}(T)
+
C_{\delta,q}
\Pr(N_T^{\rm UR}<qT)^{1/2}
+
\delta.
\end{equation*}
In the SE case, the fixed rescaling
\(n_T+1\geq q(T+1)/2\) and the outer square root are absorbed into
\(c_{\delta,q}\). Together with the previous general case bound, this proves
\Cref{eq:tv_bound_by_ur_count}.

For the final assertion, let
\(
L_{\rm UR}
:=
\liminf_{T\to\infty}N_T^{\rm UR}/T
>0
\) a.s.
Fix \(\delta\in(0,1)\), and set
\(
\eta_\delta
:=
\min\left\{1/4,\delta^2/(16B_\kappa)\right\}.
\)
Since \(L_{\rm UR}>0\) almost surely, there exists a deterministic
\(q_\delta\in(0,1/2)\) such that
\(
    \Pr(L_{\rm UR}\leq2q_\delta)\leq\eta_\delta.
\)
On \(\{L_{\rm UR}>2q_\delta\}\),
\(\mathbf 1\{N_T^{\rm UR}<q_\delta T\}\to0\). Dominated convergence
therefore gives
\begin{equation*}
    \limsup_{T\to\infty}
\Pr(N_T^{\rm UR}<q_\delta T)
\leq
\Pr(L_{\rm UR}\leq2q_\delta)
\leq
\eta_\delta.
\end{equation*}
Hence, for all sufficiently large \(T\),
\(
\Pr(N_T^{\rm UR}<q_\delta T)\leq2\eta_\delta.
\)

Repeat the preceding variance-to-TV argument with posterior-density event
level \(\delta/2\), \(q=q_\delta\), and
\(n_T=\lfloor q_\delta T\rfloor\). Before absorbing the count-probability
term into \(C_{\delta,q}\), it gives
\begin{align*}
\EE\left\{\TV(\bar\pi_T,\pi_\dataseq)\right\}
&\leq
C_{\delta,q_\delta}\mathscr O^{+}(T)
+
\left\{
\frac{B_\kappa}{2}
\Pr(N_T^{\rm UR}<q_\delta T)
\right\}^{1/2}
+
\frac{\delta}{2},
\\
\EE\left\{\TV(\bar\pi_T,\pi_\dataseq)\right\}
&\leq
C_{\delta,q_\delta}\mathscr O^{0}_{c_{\delta,q_\delta}}(T)
+
\left\{
\frac{B_\kappa}{2}
\Pr(N_T^{\rm UR}<q_\delta T)
\right\}^{1/2}
+
\frac{\delta}{2},
\qquad \varsigma^2=0.
\end{align*}
The definition of \(\eta_\delta\) makes each displayed probability term at
most \(\delta/4\) for sufficiently large \(T\). Enlarging the constants
and the deterministic threshold proves
\Cref{eq:tv_bound}.
\end{proof}

\subsection[Proofs of the Expected-Posterior Martingale and UR-Activation Results]{Proofs of \Cref{prop:barpi_martingale_main,cor:ur_activation_main}}
\label{appdx:ur_activation}

This section proves the martingale property of the expected posterior and
controls the BO--UR allocation proxy.

\subsubsection{Martingale Property of the Expected Posterior}

\begin{proof}[\Cref{prop:barpi_martingale_main}] For bounded measurable \(h\), \(M_t(h)\) is bounded by \(\|h\|_\infty\). Since
\(\{\mathcal H_t\}_{t\ge0}\) is increasing, the tower property gives
\(
\EE\{M_{t+1}(h) \mid \mathcal H_t\}=M_t(h),
\)
so \(M_t(h)\) is a bounded martingale. Hence \(M_t(h)\) converges almost surely
and in \(L^1\).

Let \(\{h_m\}_{m\ge1}\) be a countable convergence-determining class in
\(C(\Omega)\). On an event of probability one, \(M_t(h_m)\) converges for every
\(m\). Equip \(\mathcal P(\Omega)\) with the topology of weak convergence.
Since \(\Omega\) is compact, \(\mathcal P(\Omega)\) is compact under weak
convergence. Hence, every subsequence of \(\{\bar\pi_t\}\) has a weakly convergent further
subsequence. All subsequential limits agree on the class \(\{h_m\}\), hence
coincide. Thus \(\bar\pi_t\Rightarrow\bar\pi_\infty\) almost surely.
\end{proof}

\subsubsection{Allocation Proxy Control}

\begin{lemma}
\label{lem:allocation_proxy_control}
Suppose \Cref{ass:delta_state_proxy_smoothness} holds and
\(\mathcal D_0\) contains at least \(k+1\) distinct input locations. Let
\(
D_\Omega
:=
\sup_{\para,\para'\in\Omega}
\|\para-\para'\|,\)
\(S_0:=\sup_{\para\in\Omega}|f(\para)|\),
\(S_2:=\sup_{\para\in\Omega}
\|\nabla^2f(\para)\|_{\rm op},
\)
and define
\(
Z_j:=\sup_{t\geq0}\EE(S_j\mid\mathcal H_t)\), \(j\in \{0,2\}\),
\(
C_\star:=C_{\mathcal M}(1+Z_2),
\)
\(M_\star:=2Z_0.
\)
With \(c_a:=\log(3)/(2a)\), set
\(\rho_\star
:=
\tanh(c_aD_\Omega \sqrt{C_\star}),
\)
\(Q_\star
:=
(1-\rho_\star)
\exp(-2\rho_\star M_\star).
\)
Then
\begin{equation*}
    \sup_{t\geq0}\widehat p_{t,k}
\leq
\rho_\star
<1
\qquad\text{a.s.},
\end{equation*}
and, for every \(m>0\),
\(
Q_\star>0\) a.s.,
\(\EE(Q_\star^{-m})<\infty.
\)
If \(B_t\) is generated from a fresh
\(\operatorname{Bernoulli}(\widehat p_{t,k})\) draw conditionally on the
pre-branch history, then
\begin{equation*}
    \liminf_{T\to\infty}N_T^{\rm UR}/T
\geq
1-\rho_\star
>0
\qquad\text{a.s.}
\end{equation*}
\end{lemma}

\begin{proof}
Let \(\para_{0,1},\ldots,\para_{0,k+1}\) be distinct locations in \(\mathcal D_0\). At most one can equal the incumbent \(\para_t^\dagger\). Hence at least \(k\) have positive curvature-adjusted distance from \(\para_t^\dagger\), and each such distance is at most
\(
D_\Omega\sqrt{\|\mathbf C_t\|_{\rm op}}.
\)
Because \(\Delta_{t,k}^{\rm loc}\) averages the \(k\) smallest positive
distances,
\(
\Delta_{t,k}^{\rm loc}
\leq
D_\Omega\sqrt{\|\mathbf C_t\|_{\rm op}}.
\)

Applying \Cref{lem:conditional_gp_supremum} with the trivial sigma-field to
\(f\) and to the finitely many second-derivative processes gives
\(
\EE\{\exp(bS_j)\}<\infty\), \(b>0\), \(j\in \{0,2\}\),
where the assertion for \(S_2\) also uses equivalence of matrix norms in
fixed dimension. Moreover,
\begin{equation*}
    \sup_{\para\in\Omega}
\|\nabla^2\mu_t(\para)\|_{\rm op}
\leq
\EE(S_2\mid\mathcal H_t),
\end{equation*}
so the growth condition on \(\mathcal M\) gives
\(
\|\mathbf C_t\|_{\rm op}\leq C_\star.
\)
It follows that
\(
\widehat p_{t,k}
=
\tanh(c_a\Delta_{t,k}^{\rm loc})
\leq
\rho_\star.
\)

The random variables \(Z_0,Z_2\) are finite almost surely and have
exponential moments of every positive order. Indeed, for \(b>0\) and \(r>1\),
conditional Jensen's inequality and Doob's \(L^r\) maximal inequality give
\begin{equation}\label{eq:bZ_j_bound}
\EE\{\exp(bZ_j)\}\leq
\EE\left[
\sup_{t\geq0}
\EE\{\exp(bS_j)\mid\mathcal H_t\}
\right]\leq
\frac{r}{r-1}
\left[
\EE\{\exp(rbS_j)\}
\right]^{1/r}
<\infty.
\end{equation}
Since
\begin{equation*}
    M_t^\mu=
\sup_{\para,\para'\in\Omega}
|\mu_t(\para)-\mu_t(\para')|
\leq
2\EE(S_0\mid\mathcal H_t)
\leq
M_\star,
\end{equation*}
the definition of \(Q_\star\) gives a common lower bound for the
variance-capture factors used below. Finally, for \(x\geq0\),
\(
    \{1-\tanh(x)\}^{-1}
\leq
e^{2x}.
\)
Thus, for every \(m>0\),
\begin{equation*}
    Q_\star^{-m}
\leq
\exp\left(
2mc_aD_\Omega \sqrt{C_\star}
+
2mM_\star
\right),
\end{equation*}
whose expectation is finite by H\"older's inequality and \Cref{eq:bZ_j_bound}; here
\(\sqrt{C_\star}\leq\sqrt{C_{\mathcal M}}(1+Z_2)\).

For the activation claim, let \(\mathcal F_t^-\) be the sigma-field available
immediately before \(B_t\) is drawn, and set
\(
D_t^B:=B_t-\widehat p_{t,k}.
\)
Fresh branch randomisation gives
\(\EE(D_t^B\mid\mathcal F_t^-)=0\). With
\(\mathcal F_t^+:=\sigma(\mathcal F_t^-,B_t)\), these variables form bounded
martingale differences along the natural interlaced pre/post-branch
filtration. The martingale strong law therefore yields
\(
    T^{-1}\sum_{t=0}^{T-1}D_t^B\rightarrow0
\) a.s.
Consequently,
\begin{equation*}
\liminf_{T\to\infty}\frac{N_T^{\rm UR}}T=
\liminf_{T\to\infty}
\left\{
1-\frac1T\sum_{t=0}^{T-1}\widehat p_{t,k}
-\frac1T\sum_{t=0}^{T-1}D_t^B
\right\}
\geq
1-\rho_\star
>0
\qquad\text{a.s.}
\end{equation*}
\end{proof}

\begin{proof}[\Cref{cor:ur_activation_main}]
The first two conclusions follow directly from
\Cref{lem:allocation_proxy_control}. The final conclusions follow
from the last assertion of \Cref{thm:sale_main_calibration}.
\end{proof}

\subsection{Bounded-Lag Fixed-Sample Calibration}\label{appdx:fixed_sample_calibration}

\begin{assumption}
\label{ass:bounded_lag_refresh}
There are increasing \((\mathcal H_t)\)-stopping times
\(\mathfrak r_0=0<\mathfrak r_1<\mathfrak r_2<\cdots\) such that
\(
r(t)=\max\{\mathfrak r_i:\mathfrak r_i\leq t\},\)
\(\mathfrak r_{i+1}-\mathfrak r_i\leq T_{\max}+1
\)
for a deterministic integer \(T_{\max}<\infty\). At each
\(\mathfrak r_i\), a sample
\(\mathcal X_{\mathfrak r_i}
=\{\para_{\mathfrak r_i,s}\}_{s=1}^S\)
is generated using auxiliary randomness independent of \(f\), the
observation noises, and previous auxiliary draws, with
\(
\para_{\mathfrak r_i,s}
\mid\mathcal H_{\mathfrak r_i}
\sim\bar\pi_{\mathfrak r_i}
\) for every \(s\); dependence among the \(S\) sample points is allowed. The realised sample is reused and unchanged for \(\mathfrak r_i\leq t<\mathfrak r_{i+1}\). At each iteration, the branch, BO, and posterior-path randomisations are generated from fresh auxiliary randomness. Conditional on \((\mathcal H_t,\mathcal X_{r(t)})\), this randomness is independent of \(f\), future observation noises, and all auxiliary randomness used at earlier iterations.
\end{assumption}

This construction is ignorable: inductively,
\(\mathcal X_{r(t)}\) and the next query are conditionally independent of
\(f\) given \(\mathcal H_t\). After conditioning on the realised
query, the response depends on \(f\) but not on \(\mathcal X_{r(t)}\). Thus,
for this section, \Cref{ass:bounded_lag_refresh} replaces the
fresh-randomness requirement in
Assumption~\ref{ass:main_gp_setup} \textup{(A3)} while retaining its
non-anticipation conclusion.

\subsubsection{Fixed-Sample Score Calibration}

For each refresh index \(i\geq0\), set
\(
q_i
:=
(1-\widehat p_{\mathfrak r_i,k})
\exp\{-2\widehat p_{\mathfrak r_i,k}
M_{\mathfrak r_i}^\mu\}.
\)

\begin{lemma}
\label{lem:score_variance_capture}
Under the setup of
\Cref{thm:sale_practical_score_calibration}, for every \(i\geq0\),
\(
\EE\!\left[
\sigma_{\mathfrak r_i}^2(\para_{\mathfrak r_i+1})
\,\middle|\,
\mathcal H_{\mathfrak r_i}
\right]
\geq
q_iV_{\mathfrak r_i}.
\)
\end{lemma}

\begin{proof}
Write \(t=\mathfrak r_i\), and let \(J_t^\sigma\) be a deterministically
tie-broken maximiser of
\(\sigma_t(\para_{t,s})\), \(1\leq s\leq S\). If the maximum variance is
zero, the comparison below is immediate. Otherwise, the definition of
\(J_t\) gives
\begin{align*}
\log\sigma_t(\para_{t,J_t})
&\geq
\log\sigma_t(\para_{t,J_t^\sigma})
+
\widehat p_{t,k}
\left\{
\mu_t(\para_{t,J_t^\sigma})
-
\mu_t(\para_{t,J_t})
\right\}
\\
&\geq
\log\sigma_t(\para_{t,J_t^\sigma})
-
\widehat p_{t,k}M_t^\mu.
\end{align*}
Consequently,
\begin{equation*}
    \sigma_t^2(\para_{t,J_t})
\geq
\exp(-2\widehat p_{t,k}M_t^\mu)
\max_{1\leq s\leq S}\sigma_t^2(\para_{t,s}).
\end{equation*}
Since \(t\) is a refresh time, every point in the newly generated sample has
conditional marginal \(\bar\pi_t\). Hence
\begin{equation*}
\EE\!\left\{
\max_{1\leq s\leq S}\sigma_t^2(\para_{t,s})
\,\middle|\,
\mathcal H_t
\right\}
\geq
\frac1S\sum_{s=1}^S
\EE\!\left\{
\sigma_t^2(\para_{t,s})
\,\middle|\,
\mathcal H_t
\right\}
=
V_t.
\end{equation*}
Conditional on \((\mathcal H_t,\mathcal X_t)\), discard the nonnegative
BO-branch contribution and use
\(\Pr(B_t=1\mid\mathcal H_t,\mathcal X_t)
=\widehat p_{t,k}.\)
Because \(\widehat p_{t,k}\) and \(M_t^\mu\) are
\(\mathcal H_t\)-measurable, taking conditional expectation first given
\((\mathcal H_t,\mathcal X_t)\) and then given \(\mathcal H_t\) yields
\begin{equation*}
    \EE\!\left\{
\sigma_t^2(\para_{t+1})
\,\middle|\,
\mathcal H_t
\right\}
\geq
(1-\widehat p_{t,k})
\exp(-2\widehat p_{t,k}M_t^\mu)V_t
=
q_iV_t.
\end{equation*}
\end{proof}

\begin{proof}[\Cref{thm:sale_practical_score_calibration}]
The ignorability argument following \Cref{ass:bounded_lag_refresh} makes the
next query conditionally independent of \(f\) given \(\mathcal H_t\).
The proof of \Cref{prop:ur_calibration_link_main} therefore gives
\(
    \EE(V_{t+1}\mid\mathcal H_t)
=
V_t-
\EE\{\Delta_t(\para_{t+1})\mid\mathcal H_t\}
\leq
V_t.
\)
Thus \(\{V_t,\mathcal H_t\}_{t\geq0}\) is a bounded nonnegative
supermartingale.

By \Cref{lem:allocation_proxy_control},
\(
q_i
\geq
(1-\rho_\star)\exp(-2\rho_\star M_\star)
=
Q_\star
>0
\) a.s.,
and \(\EE(Q_\star^{-1})<\infty\). Since \(q_i\) is
\(\mathcal H_{\mathfrak r_i}\)-measurable,
\Cref{lem:score_variance_capture} gives
\begin{equation*}
\EE(V_{\mathfrak r_i})
\leq
\EE\!\left[
q_i^{-1}
\EE\!\left\{
\sigma_{\mathfrak r_i}^2(\para_{\mathfrak r_i+1})
\,\middle|\,
\mathcal H_{\mathfrak r_i}
\right\}
\right]
=
\EE\!\left\{
q_i^{-1}
\sigma_{\mathfrak r_i}^2(\para_{\mathfrak r_i+1})
\right\}.
\end{equation*}

Let
\(
    \mathcal I=\mathcal I_T
:=
\lceil T/(T_{\max}+1)\rceil.
\)
The bounded-gap condition gives
\(\mathfrak r_i\leq i(T_{\max}+1)\leq T-1\) for \(0\leq i<\mathcal I\).
Optional sampling gives
\begin{equation*}
    \mathcal I\EE(V_T)
\leq
\sum_{i=0}^{\mathcal I-1}\EE(V_{\mathfrak r_i}).
\end{equation*}
Because the refresh times are distinct and lie in
\([T-1]_0\), the pathwise variance-sum bound in
\Cref{lem:ao_varsum} yields
\begin{equation*}
\mathcal I\EE(V_T)
\leq
\EE\left\{
\sum_{i=0}^{\mathcal I-1}
\frac{
\sigma_{\mathfrak r_i}^2(\para_{\mathfrak r_i+1})
}{q_i}
\right\}
\leq
\EE\left\{
Q_\star^{-1}
\sum_{t=0}^{T-1}\sigma_t^2(\para_{t+1})
\right\}
\leq
C_{\gamma,\upsilon}
\EE(Q_\star^{-1})
\gamma_T(\upsilon).
\end{equation*}
The second assertion of
\Cref{prop:ur_calibration_link_main} therefore gives
\begin{equation}\label{eq:TV_var_bound}
    \EE\left\{\TV(\bar\pi_T,\pi_\dataseq)\right\}
\leq
\left\{
\frac{
C_{\gamma,\upsilon}
\EE(Q_\star^{-1})
\gamma_T(\upsilon)}
{2\mathcal I_T}
\right\}^{1/2}.
\end{equation}

To identify its order, write
\(\gamma_T^{\kappa,D}(v)\) when the kernel, domain, and regularisation parameter
need to be explicit. For either kernel alternative in the theorem, use the
standard parameterisation
\(
\kappa(\para,\para')
=
s_\kappa^2
k_0\{\bm L^{-1}(\para-\para')\},
\)
where \(s_\kappa>0\), \(\bm L\) is nonsingular, and \(k_0\) is the
corresponding standard isotropic Mat\'ern-\(\nu\) or SE kernel. Since compact
\(\bm L^{-1}\Omega\) is contained in
\(\mathbb B(\bm0,R)\) for some \(R<\infty\),
stationarity, domain monotonicity, and rescaling give
\begin{equation*}
    \gamma_T^{\kappa,\Omega}(\upsilon)
=
\gamma_T^{k_0,\bm L^{-1}\Omega}(\upsilon/s_\kappa)
\leq
\gamma_T^{k_0,\mathbb B(\bm0,R)}(\upsilon/s_\kappa).
\end{equation*}
By \citet[Theorem~7 and Corollary~8]{iwazaki2025improvedregret},
\begin{equation*}
    \gamma_T(\upsilon)
=
\begin{cases}
\mathcal O\!\left[
T^{d/(2\nu+d)}
\{\log(T+1)\}^{(4\nu+d)/(2\nu+d)}
\right],
& k_0\text{ is Mat\'ern-}\nu,\quad \nu>1/2,
\\[1mm]
\mathcal O\!\left[\{\log(T+1)\}^{d+1}\right],
& k_0\text{ is SE}.
\end{cases}
\end{equation*}
Finally,
\(\mathcal I_T
\geq T/(T_{\max}+1).
\)
Substitution into \Cref{eq:TV_var_bound} proves the claim.
\end{proof}

\subsection{Proof of the Reverse-KL Claim}
\label{app:reverse_kl_calibration}

\begin{proof}
The state-specific laws \(\rho_{t,0}\) and \(\rho_{t,1}\) have strictly
positive Lebesgue densities \(q_{t,0}\) and \(q_{t,1}\). Hence every
candidate probability law \(\pi\) with finite risk has a Lebesgue density
\(q\). Define
\begin{equation*}
    \ell_t(\para)
:=
\EE_t\!\left\{
\log q_{t,\mathbf S_t}(\para)
\right\},
\qquad
\widetilde q_t(\para)
:=
\frac{\exp\{\ell_t(\para)\}}
{\int_\Omega\exp\{\ell_t(\bm u)\}\,\dee\bm u}.
\end{equation*}
Then
\begin{equation*}
    \EE_t\!\left[
\operatorname{KL}\{\rho\|\rho_{t,\mathbf S_t}\}
\right]
=
\operatorname{KL}(q\|\widetilde q_t)
-
\log\int_\Omega\exp\{\ell_t(\bm u)\}\,\dee\bm u.
\end{equation*}
Thus the unique finite-risk minimiser is \(\widetilde q_t\). Moreover,
\(\ell_t(\para)\) equals
\(
\log\bar p_t(\para)+p_t^\star\mu_t(\para)
\)
up to an additive term that does not depend on \(\para\). It follows that
\(
    \widetilde q_t(\para)
\propto
\exp\{p_t^\star\mu_t(\para)\}\bar p_t(\para),
\)
which is the density of \(\rho_t^{(p_t^\star)}\). Candidate laws singular
with respect to \(\lambda\) have infinite risk and therefore cannot
improve on this minimiser.
\end{proof}

\clearpage

\makeatletter
\let\@sict\jmlrnormalsectionformat
\makeatother

\setcounter{section}{0}
\renewcommand{\thesection}{\Alph{section}}
\renewcommand{\thesubsection}{\thesection.\arabic{subsection}}
\renewcommand{\thesubsubsection}
  {\thesubsection.\arabic{subsubsection}}

\renewcommand*{\theHsection}{onlineappendix.\Alph{section}}
\renewcommand*{\theHsubsection}
  {\theHsection.\arabic{subsection}}
\renewcommand*{\theHsubsubsection}
  {\theHsubsection.\arabic{subsubsection}}

\crefalias{section}{onlineappendix}
\crefalias{subsection}{onlineappendix}
\crefalias{subsubsection}{onlineappendix}
\crefalias{appendix}{onlineappendix}
\crefalias{subappendix}{onlineappendix}
\crefalias{subsubappendix}{onlineappendix}

\phantomsection
\label{oa:start}
\addcontentsline{toc}{section}{Online Appendix}

\begin{center}
{\LARGE\bfseries Online Appendix}
\end{center}
\vspace{0.2in}

This online appendix provides additional illustrations, implementation specifications,
problem-specific configurations, numerical-comparison protocols, and
additional empirical results for the main article. It also gives the
auxiliary levelwise convergence analysis for the finite-dimensional
RFF--Matheron implementation of ESA. All notation follows the main paper.

\section{Construction of the Posterior-Guided Search Region}
\label{oa:search_region}

This section gives the numerical construction of the adaptive search region
\(\Omega_t\) used in the SALE implementation. Let \(S\geq2\), and let
\(
\mathcal X_t=\{\para_s\}_{s=1}^S
\)
be an approximate expected-posterior sample from \(\bar\pi_t\). Let
\(
\bm c_t=S^{-1}\sum_{s=1}^S\para_s
\)
and
\begin{equation*}
    \widehat\Sigma_t
=
(S-1)^{-1}
\sum_{s=1}^S
(\para_s-\bm c_t)(\para_s-\bm c_t)^\top.
\end{equation*}
Write
\(
    \widehat\Sigma_t
=
\bm U_t\bm\Lambda_t\bm U_t^\top,
\)
where \(\bm U_t\) is orthogonal and \(\bm\Lambda_t\) is diagonal with
nonnegative entries. Define the covariance-adjustment matrix by
\(
\bm W_t
=
\bm U_t(\bm\Lambda_t^\dagger)^{1/2}\bm U_t^\top,
\)
where \(\bm\Lambda_t^\dagger\) denotes the Moore--Penrose inverse.
When \(\widehat\Sigma_t\) is positive definite, this is the usual
whitening matrix
\(\bm U_t\bm\Lambda_t^{-1/2}\bm U_t^\top\).
Each sample point is transformed to covariance-adjusted coordinates by
\(
\bm z_s=(z_{s,1},\ldots,z_{s,d})^\top
=
\bm W_t(\para_s-\bm c_t)
\),
\(s=1,\ldots,S\).

For a trimming level \(\alpha\in(0,1)\), let
\begin{equation*}
z^L_{t,j}
=
q_{\alpha/2}
\left(
\{z_{s,j}:s=1,\ldots,S\}
\right),
\qquad
z^U_{t,j}
=
q_{1-\alpha/2}
\left(
\{z_{s,j}:s=1,\ldots,S\}
\right),
\end{equation*}
for \(j=1,\ldots,d\), where \(q_p(\cdot)\) denotes the empirical
\(p\)-quantile. Collect these coordinatewise quantiles as
\(
\bm z^L_t=(z^L_{t,1},\ldots,z^L_{t,d})^\top
\)
and
\(
\bm z^U_t=(z^U_{t,1},\ldots,z^U_{t,d})^\top
\).
The \(\bar\pi_t\)-guided search region is
\begin{equation*}
\Omega_t
=
\left\{
\para\in\Omega:
\bm z^L_t
\leq
\bm W_t(\para-\bm c_t)
\leq
\bm z^U_t
\right\},
\end{equation*}
where the inequalities are interpreted componentwise.

Smaller \(\alpha\) gives a more conservative region, whereas larger
\(\alpha\) can accelerate BO localisation but risks excluding posterior
structure before the surrogate is calibrated. We recommend
\(\alpha\leq0.1\) and use \(\alpha=0.1\) for every numerical example, without
example-specific tuning.

\section{Local-Excursion Motivation for the Allocation Proxy}
\label{oa:allocation_proxy}

The main article defines \(p_t^\star\) as a global GP excursion probability.
We here link it to local design resolution near \(\para_t^\dagger\). For a
nonempty compact \(A\subseteq\Omega\), define
\begin{equation*}
    p_t^\star(A)
:=
\PP_t\left\{
\sup_{\para\in A}f_t(\para)>f_t^++\varepsilon_{\rm imp}
\right\}.
\end{equation*}
Thus the global quantity in the main article is
\(
p_t^\star=p_t^\star(\Omega).
\)
Let \(
g_t^+(\para)
:=
\{f_t(\para)-f_t^+\}
-
\{\mu_t(\para)-\mu_t(\para_t^\dagger)\},
\)
which is a centred conditional GP because \(\para_t^\dagger\) is
\(\mathcal H_t\)-measurable. Define
\begin{equation*}
    m_t(A)
:=
\EE\!\left\{
\sup_{\para\in A}g_t^+(\para)
\mid\mathcal H_t
\right\},
\end{equation*}
\begin{equation*}
    s_t^2(A)
:=
\sup_{\para\in A}
\operatorname{Var}_t\{g_t^+(\para)\}
=
\sup_{\para\in A}
\left\{
\sigma_t^2(\para)
+
\sigma_t^2(\para_t^\dagger)
-
2\kappa_t(\para,\para_t^\dagger)
\right\},
\end{equation*}
and
\(b_t(A)
:=
\mu_t(\para_t^\dagger)
-
\sup_{\para\in A}\mu_t(\para)
+
\varepsilon_{\rm imp}.
\)
If the GP posterior has continuous sample paths and \(s_t(A)>0\), the
conditional Borell--TIS inequality gives
\begin{equation}
p_t^\star(A)
\leq
\exp\!\left[
-\frac{\{b_t(A)-m_t(A)\}_+^2}{2s_t^2(A)}
\right],
\qquad
\{x\}_+:=\max(x,0).
\label{eq:oa_local_improvement_bound}
\end{equation}
Indeed, the improvement event implies
\(
\sup_{\para\in A}g_t^+(\para)>b_t(A)
\), after which the bound follows from Gaussian concentration. The case
\(s_t(A)=0\) is immediate. \Cref{eq:oa_local_improvement_bound} separates
the two requirements for resolving a region: \(\mu_t\) must rule out an
improvement larger than \(\varepsilon_{\rm imp}\), and posterior uncertainty
in function differences relative to the incumbent must be small.

For \(r>0\), let
\(
\mathbb B_t(r)
:=
\{\para\in\Omega:
\|\para-\para_t^\dagger\|_{\mathbf C_t}\leq r\},\)
\(\|\bm v\|_{\mathbf C_t}
:=
(\bm v^\top\mathbf C_t\bm v)^{1/2},
\)
and define its local fill distance by
\(
h_t(r)
:=
\sup_{\para\in\mathbb B_t(r)}
\min_{\para_i:\,(\para_i,y_i)\in\mathcal D_t}
\|\para-\para_i\|_{\mathbf C_t}.
\)
Taking \(A=\mathbb B_t(r)\) in
\Cref{eq:oa_local_improvement_bound} links the local excursion probability
to design resolution around \(\para_t^\dagger\). In the noiseless case,
\(\para_t^\dagger\) is an evaluated input, so exact GP interpolation gives
\(
s_t\bigl(\mathbb B_t(r)\bigr)
=
\sup_{\para\in\mathbb B_t(r)}\sigma_t(\para).
\)
Here \(\sigma_t\) is the kernel power function. Under standard kernel and
domain regularity conditions, power-function bounds control
\(s_t\bigl(\mathbb B_t(r)\bigr)\), while corresponding Dudley entropy bounds
control \(m_t\bigl(\mathbb B_t(r)\bigr)\), as \(h_t(r)\) decreases
\citep{dudley1967sizes,wendland2004scattered}. Local coverage and a common
modulus of continuity for \(\mu_t\) likewise give
\(
b_t\bigl(\mathbb B_t(r)\bigr)\rightarrow\varepsilon_{\rm imp}
\).
Thus, \Cref{eq:oa_local_improvement_bound} links increasing local resolution
to a smaller excursion probability once \(\mu_t\) rules out an improvement of
size \(\varepsilon_{\rm imp}\). This conditional argument neither asserts
that a particular BO policy forces \(h_t(r)\to0\) nor treats
\(\widehat p_{t,k}\) as a calibrated estimate of \(p_t^\star\).

Computing \(h_t(r)\) requires a local covering calculation. SALE instead uses
the curvature-adjusted nearest-neighbour statistic defined in the main
article. Although \(\Delta_{t,k}^{\rm loc}\) is not a fill distance, it provides
a cheap local-spacing diagnostic. Moreover,
\begin{equation*}
    p_t^\star
\leq
p_t^\star\bigl(\mathbb B_t(r)\bigr)
+
p_t^\star\bigl(\Omega\setminus\mathbb B_t(r)\bigr).
\end{equation*}
The local proxy addresses only the first term; posterior-guided
restriction and continued BO address unresolved value discovery outside the
neighbourhood.

\section{AO Stabilisation Implication}
\label{oa:wide_spike}

Positive $\tau$ makes AO respond to both the height and the local
volume of a high-valued region. To illustrate this, let
\(\Omega=[-2,2]\), \(\varepsilon\in(0,1)\), and \(w\in(0,1/2)\), and define
\begin{equation*}
    b(x)=-x^2,
\qquad
s_{\varepsilon,w}^{\pm}(x)
=
\pm\varepsilon-(x-1)^2/w^2,
\qquad
g_{\varepsilon,w}^{\pm}(x)
=
\max\left\{
b(x),s_{\varepsilon,w}^{\pm}(x)
\right\}.
\end{equation*}
The path \(g_{\varepsilon,w}^{-}\) has a broad global mode at \(x= 0\) and a
narrow secondary mode at \(x=1\) that is lower by \(\varepsilon\). The
perturbed path \(g_{\varepsilon,w}^{+}\) raises the narrow mode above the
broad mode, giving it a unique global maximiser at \(x=1\). Nevertheless,
\begin{equation*}
W_{1,\infty}
\left(
\delta_{g_{\varepsilon,w}^{-}},
\delta_{g_{\varepsilon,w}^{+}}
\right)
=
\left\|
g_{\varepsilon,w}^{-}
-
g_{\varepsilon,w}^{+}
\right\|_\infty
=
2\varepsilon
\rightarrow0,
\end{equation*}
whereas the TS selection rule switches from \(x=0\) to \(x=1\) for every
\(\varepsilon>0\).

The broad and narrow modes of \(g_{\varepsilon,w}^{+}\) have heights
\(0\) and \(\varepsilon\), respectively, and absolute curvatures
\(H_B=2\) and \(H_N=2/w^2\). Their local Laplace contributions to the AO
normalising constant consequently satisfy
\begin{equation*}
M_N(\tau)/M_B(\tau)
\approx
\left(H_B/H_N\right)^{1/2}
\exp\left(\varepsilon/\tau\right)
=
w\exp\left(\varepsilon/\tau\right).
\end{equation*}
Thus, although TS always selects the maximiser $x = 1$ at the narrow spike under
\(g_{\varepsilon,w}^{+}\), the broad basin at $x = 0$ contributes more AO mass when
\(
w\exp\left(\varepsilon/\tau\right)\ll1.
\)
The balance occurs near
\(
\varepsilon
=
\tau\log(1/w).
\)
For fixed \(\tau > 0\), a sufficiently narrow spike therefore receives
little AO mass even if it is slightly higher. For fixed
\((\varepsilon,w)\), however, the spike dominates as \(\tau\downarrow0\),
consistently with the TS limit in Proposition 7 in the main text.

More generally, if two \(d\)-dimensional local modes have heights
\(h_B,h_N\) and positive curvature matrices \(H_B,H_N\), their local AO
mass ratio is approximately
\begin{equation*}
    \left\{
\frac{\det(H_B)}{\det(H_N)}
\right\}^{1/2}
\exp\left\{(h_N-h_B)/\tau
\right\}.
\end{equation*}
The determinant term also shows why AO can become increasingly resistant to
artificially high, narrow excursions in uncertain surrogate paths. Suppose
that a narrow mode has \(r>1\) times the curvature of a broad mode in
\(d_{\rm n}\) directions and comparable curvature in the remaining
directions. Then
\(
\left\{\det(H_B)/\det(H_N)
\right\}^{1/2}
\approx
r^{-d_{\rm n}/2},
\)
so the narrow mode must have a height advantage of
\(
h_N-h_B
\gtrsim d_{\rm n}\tau\log(r)/2
\)
before its local AO mass matches that of the broad mode. TS, by contrast,
switches to the narrow mode as soon as \(h_N>h_B\), regardless of its local
volume.

This distinction is relevant when the surrogate remains highly uncertain.
Under a sparse design, an individual GP posterior path may contain a narrow
excursion whose height is artificially inflated by posterior uncertainty,
rather than reflecting a genuine feature of the objective. As the number of
poorly resolved directions or effectively distinct candidate regions
increases, there are more opportunities for such a large positive excursion.
At the same time, if the excursion is narrow across \(d_{\rm n}\) directions,
its volume penalty under AO compounds as \(r^{-d_{\rm n}/2}\). AO can therefore
withstand an artificial height advantage of order
\(d_{\rm n}\tau\log(r)/2\) before moving most
of its query mass to the excursion, whereas TS responds immediately once the
excursion becomes the sampled-path maximiser.

The relevant height advantage is measured relative to the vertical scale of
the objective. Indeed, for any \(a\in\mathbb R\) and \(c>0\),
\(
\pi_{c\tau}(\dee x\mid a+cg)
=
\pi_\tau(\dee x\mid g).
\)
Thus, \(\tau\) rescales with the objective, and the height--volume comparison
depends on the dimensionless gap \((h_N-h_B)/\tau\), rather than on the
absolute numerical height of the excursion.

Highly anisotropic or scale-separated targets with increasing dimension are susceptible to the above mechanism. Their geometry can be difficult to resolve under
limited evaluations, making artificially high, narrow excursions more
plausible in early GP posterior paths. This provides one possible explanation
for the increasing AO-TS regret separation with \(d\) in
Figure 1 of the main text, and for the latter, more variable TV reduction of
SALE-TS on the GRF--Mat\'ern and Lotka--Volterra examples.

\section{Proposal Kernels and Numerical Settings for ESA}
\label{oa:esa_settings}

ESA is initialised at a random point of the current working expected-posterior
sample \(\mathcal X_{r(t)}\). We use
\(
m=1024\),
\(L=40\),
\(\tau_\ell=1.2^{-(\ell-1)}\),
\(s_\ell=1,\)
\(L_{\para}=20.
\)
For \(\ell\leq20\), \(L_f=2\) and \(L_{\bm u}=20\); for \(\ell>20\),
\(L_f=4\) and \(L_{\bm u}=30\). Terminal local refinement uses L-BFGS-B
with \texttt{maxit = 100}, initialised at the terminal ESA parameter state.

For the proposal distribution of \(\para\) given \(f_t\), we use a mixture of
Metropolis-adjusted Langevin algorithm
\citep{roberts1996exponential,roberts1998optimal} and random-walk proposals:
\begin{equation*}
    q_{\para}(\para'\mid\para)
=
0.8q^\ell_{\rm MALA}(\para'\mid\para)
+
0.2q^\ell_{\rm RW}(\para'\mid\para),
\end{equation*}
where
\begin{equation*}
    q^\ell_{\rm MALA}(\para'\mid\para)
=
\mathcal N\left\{
\para';
\para+\frac{0.25^2}{2}A_{\para}\nabla f_t(\para),
0.25^2\tau_\ell A_{\para}+\xi I
\right\},
\end{equation*}
and
\begin{equation*}
    q^\ell_{\rm RW}(\para'\mid\para)
=
\mathcal N\left\{
\para';
\para,
\tau_\ell\operatorname{Cov}(\mathcal X_{r(t)})+\xi I
\right\}.
\end{equation*}
Here \(A_{\para}\) is a stabilised diagonal preconditioner. It is initialised
from
\(
\operatorname{diag}\{\operatorname{Cov}(\mathcal X_{r(t)})\}
\), following the standard use of empirical covariance information in
adaptive Metropolis proposals
\citep{haario2001adaptive,roberts2009examples}. After accepted MALA moves, it
is updated using clipped coordinatewise secant ratios, a
quasi-Newton-style scaling heuristic \citep{nocedal2006numerical}, and is
reset whenever \(f_t\) changes. The jitter \(\xi=10^{-8}\) prevents numerical
degeneracy. The proposal distribution for the auxiliary variable \(\bm u\)
in the exchange step is the same as \(q^\ell_{\rm RW}\).

The proposal for \(f_t\) is specified through the joint latent state
\(
\bm z=(\bm w^\top,\bm\varepsilon^\top)^\top
\)
of the decoupled RFF--Matheron path in the main article. We use independent
preconditioned Crank--Nicolson updates for its two Gaussian blocks:
\begin{equation*}
    \bm w'
=
\sqrt{1-\rho^2}\,\bm w+\rho\bm z_{\bm w},
\qquad
\bm\varepsilon'
=
\sqrt{1-\rho^2}\,\bm\varepsilon+\rho\bm z_{\bm\varepsilon},
\end{equation*}
where
\(
\bm z_{\bm w}\sim\mathcal N(\bm0,I_m)
\),
\(
\bm z_{\bm\varepsilon}
\sim\mathcal N(\bm0,\varsigma^2I_{n_t})
\),
independently, and \(\rho=0.3\). These updates are reversible with respect
to the product Gaussian reference law of
\((\bm w,\bm\varepsilon)\).

All numerical settings are fixed across the benchmark, simulation, and
applied examples.

\section{GP Fitting, Adaptive Refresh, and Stopping}
\label{oa:sale_workflow_settings}
This section records the common GP-fitting, refresh, and stopping conventions used in the experiments.
\subsection{Initial Design and GP Fitting}
The implementation rescales \(\Omega\) to \([0,1]^d\) and uses an automatic
relevance determination squared-exponential kernel in every example. At each
scheduled refit, the signal variance, length-scales, and nugget
\begin{equation*}
    \bm\vartheta
=
(\sigma_f^2,\ell_1,\ldots,\ell_d,\varsigma^2),
\end{equation*}
are estimated by
\begin{equation*}
    \widehat{\bm\vartheta}_{\mathrm{MAP}}^t
\in
\argmax_{\bm\vartheta}
\left\{
\log p(\bm y_t\mid\bm\Theta_t,\bm\vartheta)
+
\log p(\bm\vartheta)
\right\},
\end{equation*}
where \(p(\bm y_t\mid\bm\Theta_t,\bm\vartheta)\) is the GP marginal
likelihood and \(p(\bm\vartheta)\) is the prior for \(\bm\vartheta\).
Specifically,
\(
\log\sigma_f\sim\mathcal N(0,0.5^2),
\)
\(\log\sigma_f\in[-10,10],
\)
where \(\sigma_f\) is standardised by the observed response scale;
\(
\log\ell_j
\sim
\mathcal N\{\log(0.25\sqrt d),0.25^2\},
\)
\(\ell_j\geq0.01;
\)
and
\(
\varsigma^2\sim\operatorname{InvGamma}(2,10^{-3}),
\)
\(\varsigma^2\in[10^{-4},10^{-2}].
\)
\(\widehat{\bm\vartheta}_{\mathrm{MAP}}^t\) is held fixed until the next
refit. Objective evaluations are exact or nearly exact, so \(\varsigma^2\)
is used operationally as a numerical nugget.

\subsection{Adaptive Refresh Schedule}
The refresh interval is adjusted using a moving calibration diagnostic based
on one-step-ahead GP prediction errors. Let \(W_t\) denote a moving window of
recent queried indices. For \(i\in W_t\), define the standardised predictive
residual
\begin{equation*}
    r_i
=
\frac{
y_i-\mu_{i-1}(\para_i)
}{
\{\sigma_{i-1}^2(\para_i)+\varsigma^2\}^{1/2}
}.
\end{equation*}
The calibration diagnostic monitors whether recent residuals are
approximately centred, have approximately unit variance, and do not have
excessive large-tail events. Specifically, write
\begin{equation*}
    \bar r_t
=
\frac{1}{|W_t|}\sum_{i\in W_t}r_i,
\qquad
\bar s_t^2
=
\frac{1}{|W_t|-1}
\sum_{i\in W_t}(r_i-\bar r_t)^2,
\qquad
\bar q_t
=
\frac{1}{|W_t|}
\sum_{i\in W_t}\mathbf 1\{|r_i|>z_{\rm cal}\}.
\end{equation*}
The surrogate is treated as poorly calibrated when
\begin{equation*}
    |\bar r_t|>\epsilon_{\rm mean}
\quad\text{or}\quad
|\bar s_t^2-1|>\epsilon_{\rm var}
\quad\text{or}\quad
\bar q_t>\epsilon_{\rm tail}.
\end{equation*}
In the experiments, we set
\(
|W_t|=20\),
\(z_{\rm cal}=2\),
\(\epsilon_{\rm mean}=0.25\),
\(\epsilon_{\rm var}=0.5\),
\(\epsilon_{\rm tail}=0.2.
\)
A refresh is triggered only after the poor-calibration condition persists for
three consecutive checks and the current cooldown interval has elapsed. The
cooldown interval is shortened by a shrinkage factor of \(0.91\) when the
calibration condition is consistently poor, subject to the lower bound
\(T_{\min}=20\). When the surrogate remains well calibrated, the interval is
increased by a growth factor of \(1.15\) up to \(T_{\max}=200\).

\subsection{Stopping Rule}
At a refresh time, we treat two successive ssMCMC samples as parallel chains,
compute a coordinatewise potential scale reduction factor
\(\widehat R_j\), and set
\(
\widehat R_{\max}
=
\max_{1\leq j\leq d}\widehat R_j.
\)
With \(p_{\rm BO}=0.1\) and \(\varepsilon_R=0.01\), SALE stops when
\(
\widehat p_{t,k}<p_{\rm BO}\) and 
\(\widehat R_{\max}<1+\varepsilon_R,
\)
or when the evaluation budget is exhausted.

\section{Implementation and Problem-Specific Details}\label{oa:implementation_details}
This section records the problem-specific settings and the protocols used for the supplementary numerical comparisons.

\subsection{Computational Protocol for Numerical Comparisons}
\label{oa:comparison_protocol}

The internal SALE variants differ only in the BO strategy: AO/ESA, TS, or
UCB. For TS and UCB, the acquisition function is optimised over
\(\Omega_{r(t)}\) by eight-start L-BFGS-B with \texttt{maxit = 100} and
analytic gradients. The starts are eight distinct random points from
\(\mathcal X_{r(t)}\). All remaining SALE components are fixed as specified
in the main article and \Cref{oa:esa_settings,oa:sale_workflow_settings}.

For the external baselines, acquisition optimisation uses the same
eight-start L-BFGS-B routine, with starts drawn from an eight-point Latin
hypercube design over \(\Omega\). The GP surrogate class and
hyperparameter-refitting rule are matched to SALE.

The external baselines have the following acquisition rules:
\begin{itemize}
    \item GP-UCB; \citep{gutmann2016bolfi, li2026boss}:
    \begin{equation*}
        \para_{n_t + 1} = \arg\max_{\para \in \Omega} \mu_t(\para) + \sqrt{\beta_t}\sigma_t(\para), \quad \beta_t = 2\log\left(n_t^2\pi^2|\Omega|/6\delta\right), \quad \delta = 0.01.
    \end{equation*}
    \item GP-UCB+; \citep{kim2025enhancing}:
    \begin{align*}
        \para_{n_t + 1}^{\rm BO} = \arg\max_{\para \in \Omega} \mu_t(\para) + \sqrt{\beta_t}\sigma_t(\para), & \quad \beta_t = 2\log\left(n_t^2\pi^2|\Omega|/6\delta\right), \quad \delta = 0.01, \\
        \para_{n_t + 1}^{\rm rand} \sim &{\rm Unif}(\Omega).
    \end{align*}
    \item Log-GP Entropy; \citep{wang2018adaptive}:
    \begin{equation*}
        \para_{n_t + 1} = \arg\max_{\para \in \Omega} \mu_t(\para) + \log\sigma_t(\para).
    \end{equation*}
    \item Log-GP Variance; \citep{kandasamy2017query}:
    \begin{equation*}
        \para_{n_t + 1} = \arg\max_{\para \in \Omega} 2\mu_t(\para) + \sigma_t^2(\para) + \log[\exp\{\sigma^2_t(\para)\} - 1].
    \end{equation*}\end{itemize}
Note that for GP-UCB+, at each active time \(t\), there are two evaluations
based on the GP-UCB query \(\para_{n_t + 1}^{\rm BO}\) and the random
exploration query \(\para_{n_t + 1}^{\rm rand}\). Hence, the total number of
iterations for GP-UCB+ is half that of the other methods, while the total
number of evaluations is kept equal.

Within each replicate, all methods use the same 20-point Latin hypercube
initial design. For the analytic benchmarks, a given dimension and replicate
also share that design across targets; designs vary across replicates. The
UDG comparison likewise uses one common 20-point design for SALE and GP-UCB.

\subsection{Simulation and Applied-Example Details}
\label{oa:example_details}
We first give the parameterisations, priors, and numerical constructions for the simulated and applied examples.
\subsubsection{GRF--Mat\'ern}
\label{oa:matern}

The GRF--Mat\'ern data are generated with
\(\mu=0\), \(\sigma_f=1\), \(\ell=0.12\), \(\delta=0.10\), and
\(\nu=1.5\). The priors are
\(
\mu\sim\mathcal N(0,2^2)\),
\(\log\sigma_f\sim\mathcal N(\log1,0.5^2)\),
\(\log\ell\sim\mathcal N(\log0.15,0.5^2),
\)
\(
\log\delta\sim\mathcal N(\log0.10,0.5^2)\),
\(
\nu\sim\operatorname{Gamma}(3,2),
\)
restricted to
\(
\mu\in[-4,4]\),
\(\sigma_f\in[0.1,4]\),
\(\ell\in[0.01,1]\),
\(\delta\in[0.01,2]\),
\(\nu\in[0.5,5].
\)

\subsubsection{Lotka--Volterra Model}
\label{oa:lv}

For inference on the Lotka--Volterra model, we use
\begin{equation*}
    \para
=
(
m_{\alpha\gamma},
\log\beta,
d_{\alpha\gamma},
\log\delta,
\log x_0,
\log y_0,
\log\sigma
)\in\mathbb R^7,
\end{equation*}
where
\(
m_{\alpha\gamma}=\frac12\log(\alpha\gamma)
\)
and
\(
d_{\alpha\gamma}=\frac12\log(\alpha/\gamma)
\),
which reduce posterior dependence between \(\alpha\) and \(\gamma\). The
data-generating value on the natural scale is
\(
\para_{\mathrm{nat},0}
=
(1,0.05,1.5,0.075,30,4,0.25).
\)
The prior is placed on
\(
\log\para_{\rm nat}
:=
(\log\alpha,\log\beta,\log\gamma,\log\delta,
\log x_0,\log y_0,\log\sigma),
\)
with
\(
\log\para_{\rm nat}
\sim
\mathcal N(\log\para_{\rm nat,0},\Sigma)\),
\(
\Sigma
=
\operatorname{diag}
(0.75^2,1,0.75^2,1,0.75^2,0.75^2,0.75^2).
\)
The search domain is
\(
m_{\alpha\gamma}\in[0.1,0.3]\),
\(
\log\beta\in[-4,-2],
\)
\(d_{\alpha\gamma}\in[-0.7,0.5]\),
\(\log\delta\in[-3.5,-1.5],
\)
\(\log x_0 \in[2.5,4]\),
\(\log y_0\in[0.5,2],
\)
\(\log\sigma\in[-2.5,-0.5].
\)

\subsubsection{Bus-Engine Replacement Example}
\label{oa:rust}

Let \(q_j(\mu,r,\omega)\), \(j\in[J_{\max}]_0\), denote the probabilities
of the discretised mileage increment \(J\). Specifically,
\(
q_0
=
\omega+(1-\omega)\PP_{\rm NB}(0;\mu,r)\),
\(q_j
=
(1-\omega)\PP_{\rm NB}(j;\mu,r)\),
\( j\in[J_{\max}-1],
\)
and the final category collects the upper tail:
\(
q_{J_{\max}}
=
(1-\omega)\PP_{\rm NB}(J\geq J_{\max};\mu,r),
\) \(J_{\max}=24.\)
If the engine is kept, mileage accumulates; replacement resets the mileage
state before the next increment. The transition probability from state
\(x\) to \(x'\), under action \(a\in\{0,1\}\), is
\begin{equation*}
P_{a,\para}(x,x')
=
\sum_{j=0}^{J_{\max}}
q_j(\mu,r,\omega)
\mathbf 1\!\left[
x'=\min\{(1-a)x+j,K-1\}
\right].
\end{equation*}

Define
\(
u_{0,\para}(x)=-C(x)\),
\(
u_{1,\para}(x)=-R_C.
\)
For fixed discount factor \(\beta=0.9999\), the choice-specific function
\(Q_{a,\para}\) and integrated value function \(V_{\para}\) satisfy
\begin{equation*}
    Q_{a,\para}(x)
=
u_{a,\para}(x)
+
\beta\sum_{x'=0}^{K-1}
P_{a,\para}(x,x')V_{\para}(x'), \quad
V_{\para}(x)=
\log\!\left[
\exp\{Q_{0,\para}(x)\}
+
\exp\{Q_{1,\para}(x)\}
\right].
\end{equation*}

The resulting conditional choice probability is
\begin{equation*}
    p_{\para}(a\mid x)
=
\frac{\exp\{Q_{a,\para}(x)\}}
{\exp\{Q_{0,\para}(x)\}+\exp\{Q_{1,\para}(x)\}}.
\end{equation*}
The log-likelihood combines the observed choice and transition
contributions:
\begin{equation*}
\ell(\para)
=
\sum_{i,h}
\left\{
\log p_{\para}(a_{ih}\mid x_{ih})
+
\log P_{a_{ih},\para}(x_{ih},x_{i,h+1})
\right\}.
\end{equation*}

For the prior distribution, we use independent Gaussian priors on
the log scale. The parameter vector
\(
\para=(R_C,c_1,c_2,c_3,\mu,r,\omega)
\)
has \(\log\para\) truncated to \(\log\Omega\), with mean and
standard-deviation
\(
m=(\log 10,\log 5,\log 2,0,\log 2,\log 5,-1),
\)
\(s=(1,1.5,1.5,1.5,1,1,2)\).
The search box \(\Omega\) is
\(
R_C\in[2,50]\),
\(c_1\in[0.01,50],\)
\(c_2\in[0.0025,20],\)
\(c_3\in[0.0025,10]\),
\(\mu\in[0.2,15],\)
\(r\in[0.5,120],\)
\(\omega\in[3\times10^{-4},0.95].
\)

\subsubsection{UDG Quadrature-Rule Construction}
\label{oa:udg}

To propagate uncertainty in \(\para\) to \(\mathcal U(s)\), we use the SALE
posterior mode and \(2^4=16\) additional points along sign-combination
directions in the local eigen-coordinate system. For each quadrature point of
\(\para\), we fit the conditional latent Gaussian model by INLA and average
the resulting \(\mathcal U(s)\) using the quadrature weights. A full
tensor-product adaptive Gauss--Hermite rule of \citet{aghqtheory} requires
fitting and storing at least \(3^4=81\) conditional-INLA models, which is
memory-prohibitive here.

\subsection{Reference-Posterior Samples for TV Computation}
\label{oa:reference_samples}

For each analytic benchmark, we generate \(50{,}000\) draws with the
No-U-Turn Sampler in \texttt{rstan}; after burn-in and thinning,
\(M=5{,}000\) draws are retained. For each simulated example, adaptive
Metropolis produces \(20{,}000\) draws, again reduced to \(M=5{,}000\).
The retained chains have near-unity potential scale reduction factors.

\section{Additional Empirical Results}\label{oa:additional_results}
This section provides the supplementary empirical comparisons and runtime audit.
\subsection{Simple-Regret Comparison}
\label{oa:simple_regret}

\begin{figure}[t]
    \centering
    \includegraphics[width=0.92\linewidth]
    {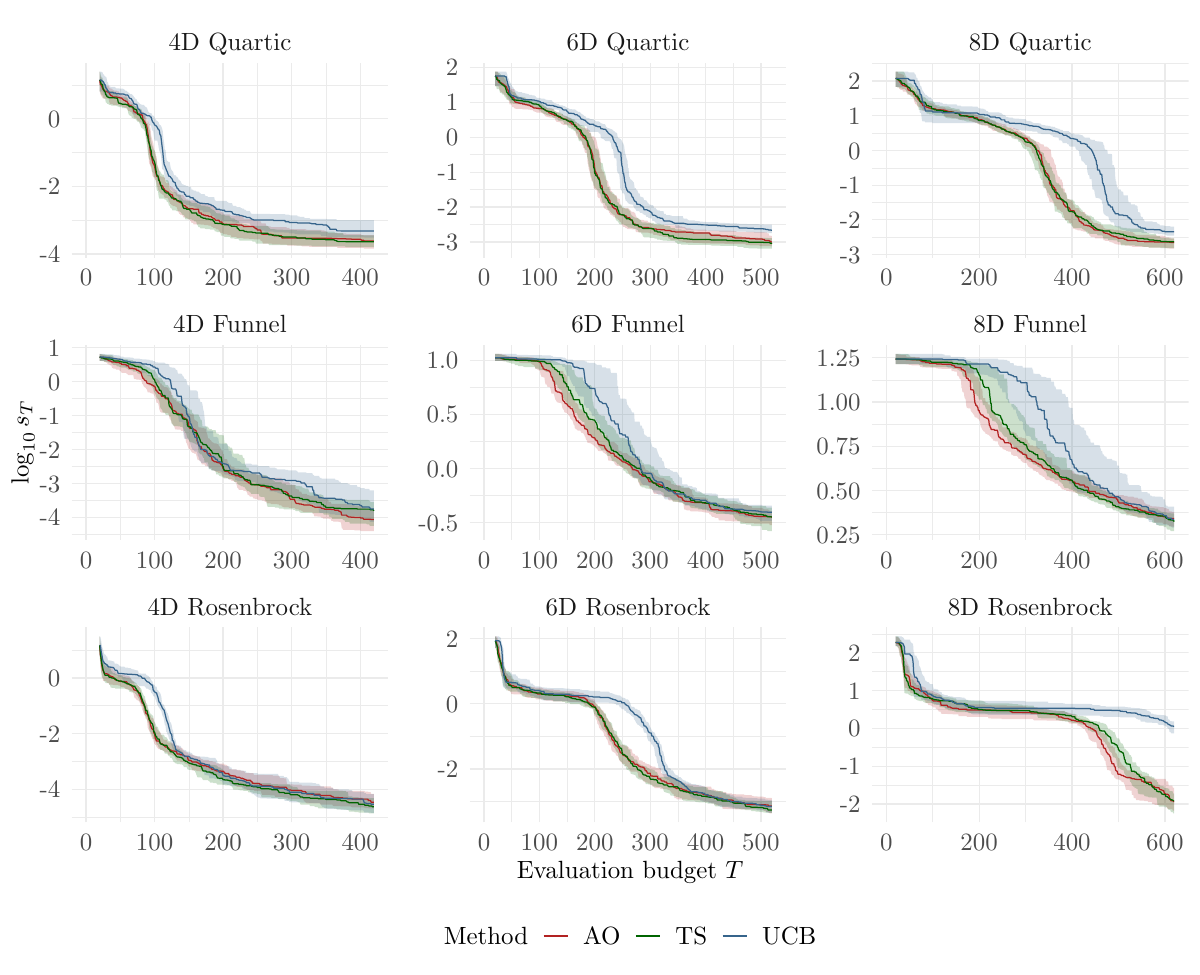}
    \caption{Log simple regret \(\log_{10}s_T\) under AO, TS, and UCB for
    the three unimodal benchmark targets. Curves show medians over \(50\)
    runs; shading shows IQRs. Lower is better.}
    \label{fig:oa_simple_regret}
\end{figure}

AO and TS attain similar simple-regret performance in
\Cref{fig:oa_simple_regret}, despite AO's lower average BO-branch regret in
the main article. Thus, both methods identify competitive high-value points
at broadly similar budgets, while AO makes lower-gap BO queries more
consistently along the full trajectory. This distinction is consistent with,
but not implied by, the perturbation-stability result in the main article.

\subsection{Internal SALE Inference Comparisons}
\label{oa:internal_inference}

\Cref{fig:oa_tv_benchmark_internal,fig:oa_tv_simulation_internal}
report the internal SALE comparisons for the analytic benchmarks and
simulated examples, respectively.

On the analytic benchmarks, the SALE variants are similar in the easier
settings, but AO generally begins reducing TV earlier on the
higher-dimensional Funnel and Rosenbrock targets. UCB reduces TV later on
these targets, although it can finish slightly lower on the Bimodal examples.

The separation is clearer in the simulated examples: AO reduces TV earlier
and more consistently, TS transitions later and with wider variability, and
UCB makes little progress within the available budgets. These trajectories
support AO as the default BO branch for the tested problems while preserving
the target-dependent exceptions visible in the benchmarks.

\begin{figure}[t]
    \centering
    \includegraphics[width=0.92\linewidth]
    {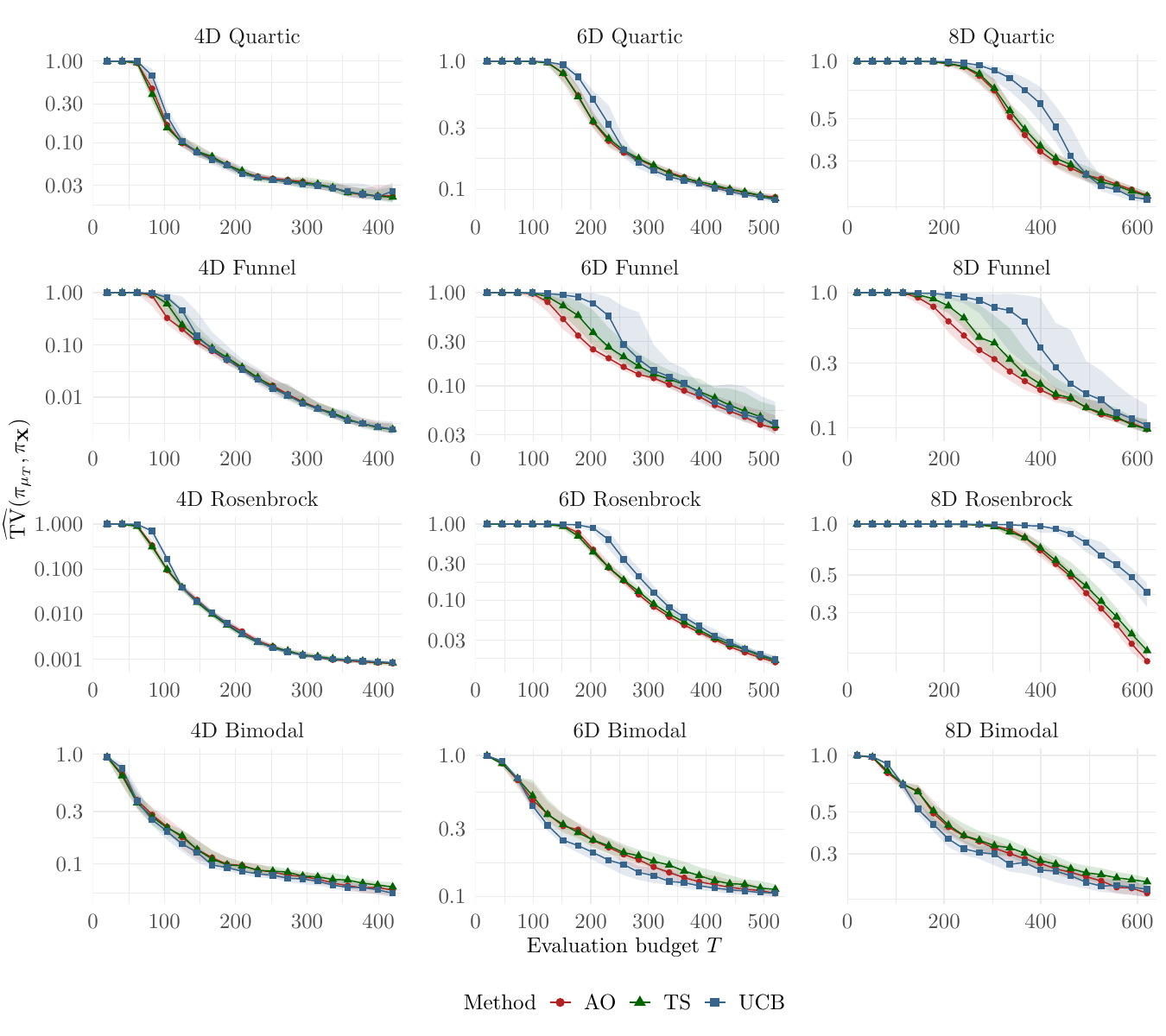}
    \caption{TV estimates
    \(\widehat{\mathrm{TV}}(\pi_{\mu_T},\pi_\dataseq)\) for SALE-AO,
    SALE-TS, and SALE-UCB across benchmark targets and dimensions. Curves show
    medians over \(50\) runs; shading shows IQRs. Lower is better.}
    \label{fig:oa_tv_benchmark_internal}
\end{figure}

\begin{figure}[t]
    \centering
    \includegraphics[width=0.92\linewidth]
    {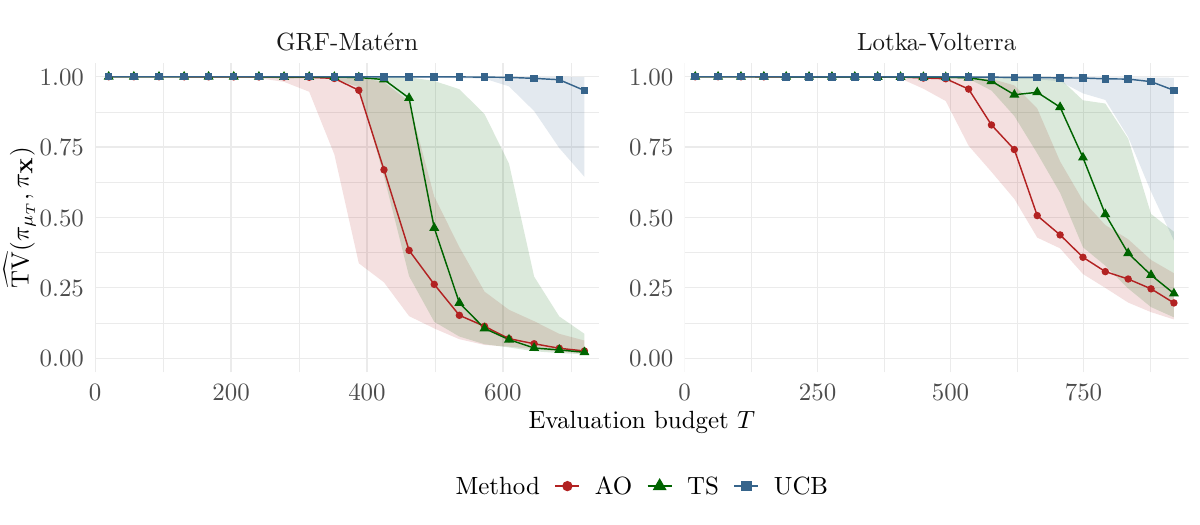}
    \caption{TV estimates
    \(\widehat{\mathrm{TV}}(\pi_{\mu_T},\pi_\dataseq)\) for SALE-AO,
    SALE-TS, and SALE-UCB across simulated examples. Curves show medians over
    \(30\) runs; shading shows IQRs. Lower is better.}
    \label{fig:oa_tv_simulation_internal}
\end{figure}

\subsection{Representative Run-Time Audit}
\label{oa:runtime}

\begin{table}[t]
\centering
\vspace{0.08in}
{\small
\renewcommand{\arraystretch}{1.08}
\begin{tabular}{llrrrr}
\toprule
Target & Method & \(n\) & Total budget \(T\) & Median [IQR] & Rel. median \\
\midrule
\multirow{2}{*}{Quartic 4D}
& SALE-AO          & 50 & 420 & \(138.2\,[122.5,143.9]\) & \(1.00\) \\
& Log-GP entropy   & 50 & 420 & \(74.4\,[69.5,81.9]\)    & \(0.54\) \\
\addlinespace
\multirow{2}{*}{Quartic 6D}
& SALE-AO          & 50 & 520 & \(201.6\,[174.7,213.4]\) & \(1.00\) \\
& Log-GP entropy   & 50 & 520 & \(205.9\,[186.7,216.3]\) & \(1.02\) \\
\addlinespace
\multirow{2}{*}{Quartic 8D}
& SALE-AO          & 50 & 620 & \(313.5\,[288.7,345.0]\) & \(1.00\) \\
& Log-GP entropy   & 50 & 620 & \(462.3\,[411.5,497.9]\) & \(1.47\) \\
\addlinespace
\multirow{4}{*}{GRF--Mat\'ern}
& SALE-AO          & 30 & 720 & \(416.8\,[388.8,450.2]\) & \(1.00\) \\
& SALE-TS          & 30 & 720 & \(423.2\,[390.3,435.3]\) & \(1.02\) \\
& SALE-UCB         & 30 & 720 & \(484.5\,[465.8,505.0]\) & \(1.16\) \\
& Log-GP entropy   & 30 & 720 & \(361.4\,[285.3,471.1]\) & \(0.87\) \\
\addlinespace
\multirow{4}{*}{Lotka--Volterra}
& SALE-AO          & 30 & 920 & \(950.8\,[909.4,1010.0]\) & \(1.00\) \\
& SALE-TS          & 30 & 920 & \(1000.9\,[946.5,1067.0]\) & \(1.05\) \\
& SALE-UCB         & 30 & 920 & \(1175.8\,[1065.2,1266.7]\) & \(1.24\) \\
& Log-GP entropy   & 30 & 920 & \(895.7\,[812.3,985.9]\) & \(0.94\) \\
\bottomrule
\end{tabular}
}
\caption{End-to-end wall-clock runtime audit. Timings are in seconds and
summarised over \(n\) independent replicates. The relative median is computed
within each target, using the default SALE implementation as the reference.}
\label{tab:oa_runtime_audit}
\end{table}

\Cref{tab:oa_runtime_audit} reports wall-clock time for log-GP entropy,
the closest external competitor in the TV comparisons. Relative to it,
SALE-AO is slower for the 4D Quartic target, comparable for 6D, and faster for 8D. Log-GP
entropy is \(13\%\) and \(6\%\) faster in median runtime on the GRF--Mat\'ern
and Lotka--Volterra examples, respectively. SALE-TS is within \(5\%\) of
SALE-AO on both simulated examples, whereas SALE-UCB is \(16\%\) and \(24\%\)
slower.

These are implementation-specific, end-to-end audits rather than
language-normalised complexity comparisons. The primary empirical comparison
remains posterior accuracy under a fixed budget of objective evaluations.

\clearpage
\section{Auxiliary RFF--ESA Convergence Results}
\label{oa:rff_esa_convergence}

This section records the auxiliary convergence analysis for the
finite-dimensional RFF--Matheron implementation of ESA. It is not needed for
the principal AO regret or UR calibration results in the main article. We
adopt the standing GP setup and notation of the main article, fix a realised
design \(\mathcal D_t\), and state the additional conditions used only in this
section.

\subsection{RFF and ESA Assumptions}

\begin{assumption}
\label{ass:oa_rff}
For each fixed realised design, assume
\(\bm K_t+\varsigma^2\bm I_{n_t}\) is positive definite after redundant exact
observations have been removed, and
\(
\EE_\rho\|\omega\|_2^2<\infty.
\)
\end{assumption}

\begin{assumption}
\label{ass:oa_esa}
For fixed \((t,m,\tau)\), assume:
\begin{enumerate}
\item[(E1)] \(\mathcal W\subset\mathbb R^{m+n_t}\) is compact,
\(\nu_t^{(m)}\) has full support on \(\mathcal W\), and
\(h_t^{(m)}\) is continuous on \(\Omega\times\mathcal W\). Write
\(
M_t^{(m)}
=
\sup_{\para\in\Omega,\bm z\in\mathcal W}
\lvert h_t^{(m)}(\para,\bm z)\rvert
<\infty.
\)

\item[(E2)] Let
\(
\bar\lambda=\lambda/\lambda(\Omega).
\)
The parameter proposal has density \(q_{\para}(\para'\mid\para)\) with respect to
\(\bar\lambda\), and the latent proposal has density
\(q_{\bm z}(\bm z'\mid\bm z)\) with respect to \(\nu_t^{(m)}\). There exist
finite positive constants
\(\underline q_{\para},\bar q_{\para},\underline q_{\bm z},\bar q_{\bm z}\) such
that
\begin{equation*}
    0<\underline q_{\para}
\le
q_{\para}(\para'\mid\para)
\le
\bar q_{\para}
<\infty,
\qquad
\forall \para,\para'\in\Omega,
\end{equation*}
and
\begin{equation*}
    0<\underline q_{\bm z}
\le
q_{\bm z}(\bm z'\mid\bm z)
\le
\bar q_{\bm z}
<\infty,
\qquad
\forall \bm z,\bm z'\in\mathcal W .
\end{equation*}

\item[(E3)] The exchange auxiliary variable is drawn exactly from
\(\pi_\tau\bigl(\cdot\mid h_t^{(m)}(\cdot,\bm z')\bigr)\), and each ESA sweep
contains at least one \(\para\)-update and one \(\bm z\)-exchange update.
\end{enumerate}
\end{assumption}
\subsection{Levelwise Convergence of ESA under the RFF--Matheron Representation}\label{oa:esa_rff_exact}

Fix \(t\), condition on \(\mathcal D_t\), and fix \(\tau>0\). For a stationary
kernel \(\kappa\), draw \(\omega_j\sim\rho\) from its spectral distribution
and \(b_j\sim\operatorname{Unif}(0,2\pi)\), independently, and define
\begin{equation*}
    \bm\phi_m(\para)
=
\sqrt{\frac{2\kappa(0)}{m}}
\left[
\cos(\omega_1^\top\para+b_1),\ldots,
\cos(\omega_m^\top\para+b_m)
\right]^\top .
\end{equation*}
Let
\(
\bm\Phi_{t,m}
=
\left[
\bm\phi_m(\para_1),\ldots,\bm\phi_m(\para_{n_t})
\right]^\top
\)
and
\(
\bm A_t
=
(\bm K_t+\varsigma^2\bm I_{n_t})^{-1}.
\)
Conditional on
\(\Xi_m=\{(\omega_j,b_j)\}_{j=1}^m\), the decoupled RFF--Matheron path is
\begin{equation*}
f_{t,m}^{\rm dec}(\para;\bm w,\bm\varepsilon)
=
\bm\phi_m(\para)^\top\bm w
+
\bm k_t(\para)^\top\bm A_t
\left\{
\bm y_t-\bm\Phi_{t,m}\bm w-\bm\varepsilon
\right\},
\end{equation*}
where
\(
\bm w\sim\mathcal N(\bm0,\bm I_m)
\)
and
\(
\bm\varepsilon
\sim
\mathcal N(\bm0,\varsigma^2\bm I_{n_t})
\)
independently. This is the construction used in \Cref{app:esa} of the main
paper. For brevity, write
\(
f_{t,m}^{\rm dec}
=
f_{t,m}^{\rm dec}(\cdot;\bm w,\bm\varepsilon)
\)
below, and define
\begin{equation*}
    \varphi_{\omega,b}(\para)
=
\sqrt{2\kappa(0)}\cos(\omega^\top\para+b),
\qquad
\kappa_m(\para,\para')
=
\bm\phi_m(\para)^\top\bm\phi_m(\para').
\end{equation*}
Set
\(
\bm K_{t,m}
=
\bm\Phi_{t,m}\bm\Phi_{t,m}^\top
\),
\(
\bm k_{t,m}(\para)
=
\bm\Phi_{t,m}\bm\phi_m(\para)
\),
\(
\bm G_t
=
\bm K_t+\varsigma^2\bm I_{n_t},\)
and \(\bm a_t(\para)
=
\bm G_t^{-1}\bm k_t(\para).
\)

Conditional on \((\mathcal D_t,\Xi_m)\), \(f_{t,m}^{\rm dec}\) is Gaussian with
\(
\mu_{t,m}(\para)=\mu_t(\para)
\)
and covariance
\begin{align}
\label{eq:decoupled_rff_covariance}
\kappa_{t,m}(\para,\para')
={}&
\kappa_m(\para,\para')
-\bm a_t(\para)^\top\bm k_{t,m}(\para')
-\bm a_t(\para')^\top\bm k_{t,m}(\para)
\notag\\
&+\bm a_t(\para)^\top
(\bm K_{t,m}+\varsigma^2\bm I_{n_t})
\bm a_t(\para').
\end{align}
Let
\(
\pi_{\tau,t}^{(m)}(\dee\para\mid\Xi_m)
=
\EE\{\pi_\tau(\dee\para\mid f_{t,m}^{\rm dec})
\mid\mathcal D_t,\Xi_m\}
\)
denote the corresponding RFF-induced AO law. Define
\(
\varepsilon_m
=
\|\kappa_m-\kappa\|_{L^\infty(\Omega\times\Omega)}.
\)
The notation \(O_{\Pr_{\Xi_m}}(\cdot)\) refers to stochastic order with
respect to the RFF draw \(\Xi_m\), conditionally on \(\mathcal D_t\). For a
probability measure on \(\Omega\times\mathcal W\), the subscript
\((\cdot)_{\para}\) denotes its \(\para\)-marginal.

\begin{proposition}
\label{prop:oa_rff_ao}
Under the standing GP setup of the main article and \Cref{ass:oa_rff}, for fixed
\((t,\mathcal D_t,\tau)\), there exists a finite constant
\(C_{\tau,t}<\infty\), depending on \(\tau\), \(\mathcal D_t\),
\(\varsigma^2\), and the kernel, such that
\(
\TV\!\left\{
\pi_{\tau,t}^{(m)}(\cdot\mid\Xi_m),
\pi_{\tau,t}
\right\}
\le
C_{\tau,t}\varepsilon_m .
\)
Moreover,
\(
\varepsilon_m
=
O_{\Pr_{\Xi_m}}(m^{-1/2})\),
\(
\TV\!\left\{
\pi_{\tau,t}^{(m)}(\cdot\mid\Xi_m),
\pi_{\tau,t}
\right\}
=
O_{\Pr_{\Xi_m}}(m^{-1/2})
\).
\end{proposition}

\begin{proof}
The uniform RFF kernel approximation
\(\varepsilon_m=O_{\Pr_{\Xi_m}}(m^{-1/2})\) follows from a standard
empirical-process argument for the random-feature class \citep{sutherland2015error}. Indeed, write
\(
g_{\para,\para'}(\omega,b)
=
\varphi_{\omega,b}(\para)\varphi_{\omega,b}(\para')\),
\(
(\para,\para')\in\Omega\times\Omega .
\)
The class
\(
\mathcal G
=
\{g_{\para,\para'}:(\para,\para')\in\Omega\times\Omega\}
\)
has a bounded envelope because
\(
|\varphi_{\omega,b}(\para)|\le \sqrt{2\kappa(0)}
\).
Moreover, under the finite spectral second-moment condition in
\Cref{ass:oa_rff}, it has finite \(L^2\)-entropy integral as a
finite-dimensional Lipschitz-parametric class over compact
\(\Omega\times\Omega\). Hence \(\mathcal G\) is \(\PP_{\omega,b}\)-Donsker, where \(\PP_{\omega,b}\) is the law of \((\omega,b)\), and
\begin{equation*}
    \sqrt m
\sup_{\para,\para'\in\Omega}
\left|
m^{-1}\sum_{j=1}^m
g_{\para,\para'}(\omega_j,b_j)
-
\EE g_{\para,\para'}(\omega,b)
\right|
=
O_{\Pr_{\Xi_m}}(1).
\end{equation*}
Equivalently,
\(
\varepsilon_m=O_{\Pr_{\Xi_m}}(m^{-1/2})
\).

Let
\(
\delta\kappa_m
=
\kappa_m-\kappa,\)
\(\delta\bm k_{t,m}(\para)
=
\bm k_{t,m}(\para)-\bm k_t(\para)\),
\(
\delta\bm K_{t,m}
=
\bm K_{t,m}-\bm K_t.
\)
Expanding the exact posterior covariance in the same form as
\Cref{eq:decoupled_rff_covariance} gives
\begin{equation*}
    \kappa_{t,m}(\para,\para')-\kappa_t(\para,\para')
=
\delta\kappa_m(\para,\para')
-\bm a_t(\para)^\top
\delta\bm k_{t,m}(\para')
-\bm a_t(\para')^\top
\delta\bm k_{t,m}(\para)+
\bm a_t(\para)^\top
\delta\bm K_{t,m}
\bm a_t(\para').
\end{equation*}
Let
\(
M_{a,t}
=
\sup_{\para\in\Omega}
\|\bm a_t(\para)\|_2.
\)
If \(n_t=0\), set \(M_{a,t}=0\). If \(n_t\geq1\),
\Cref{ass:oa_rff} ensures that
\(\bm G_t=\bm K_t+\varsigma^2\bm I_{n_t}\) is positive definite. Therefore,
\begin{equation*}
    M_{a,t}
\leq
\|\bm G_t^{-1}\|_{\mathrm{op}}
\sup_{\para\in\Omega}\|\bm k_t(\para)\|_2
\leq
\frac{
\sqrt{n_t}\,
\sup_{\para,\bm u\in\Omega}
|\kappa(\para,\bm u)|
}{
\lambda_{\min}(\bm G_t)
}
<\infty.
\end{equation*}
The final inequality follows because \(n_t\) is finite,
\(\lambda_{\min}(\bm G_t)>0\), and \(\kappa\) is bounded on compact
\(\Omega\). Moreover,
\begin{equation*}
    \sup_{\para\in\Omega}
\|\delta\bm k_{t,m}(\para)\|_2
\leq
\sqrt{n_t}\,\varepsilon_m,
\qquad
\|\delta\bm K_{t,m}\|_{\rm op}
\leq
n_t\varepsilon_m.
\end{equation*}
Consequently,
\begin{equation*}
    \|\mu_{t,m}-\mu_t\|_{L^\infty(\Omega)}
+
\|\kappa_{t,m}-\kappa_t\|_{L^\infty(\Omega\times\Omega)}
\leq
C_{\mathrm{RFF},t}\varepsilon_m,
\end{equation*}
where
\(
C_{\mathrm{RFF},t}
=
(1+\sqrt{n_t}\,M_{a,t})^2.
\)

It remains to transfer this posterior mean-covariance perturbation to the AO
law. For a finite Borel partition
\(\mathcal Q=\{C_1,\ldots,C_N\}\) of $\Omega$, choose \(\bm \xi_i\in C_i\) and write the
discretised Gibbs mass of a measurable \(A\subseteq\Omega\) as
\begin{equation*}
    H_{A,\mathcal Q}(\bm x)
=
\frac{
\sum_{i=1}^N \lambda(A\cap C_i)\exp(x_i/\tau)
}{
\sum_{i=1}^N \lambda(C_i)\exp(x_i/\tau)
}, \quad \bm x= (x_1, \dots, x_N) \in \mathbb R^N.
\end{equation*}
The sums of the absolute first and second partial derivatives of
\(H_{A,\mathcal Q}\) are bounded by constants depending only on \(\tau\),
uniformly in \(A\), \(\mathcal Q\), and \(N\). Indeed, if
\begin{equation*}
    w_i(\bm x)=\frac{\lambda(C_i)\exp(x_i/\tau)}
{\sum_{\ell=1}^N\lambda(C_\ell)\exp(x_\ell/\tau)},
\end{equation*}
then
\(0\le H_{A,\mathcal Q}\le1\) and
\(
\sum_{i=1}^N
|\partial_i H_{A,\mathcal Q}(\bm x)|
\le 1/\tau,\)
\(\sum_{i,j=1}^N
|\partial_{ij} H_{A,\mathcal Q}(\bm x)|
\le C/\tau^2
\)
for a universal constant \(C<\infty\).

Let
\(
\bm Z_m
=
\{f_{t,m}^{\rm dec}(\bm \xi_1),\ldots,
f_{t,m}^{\rm dec}(\bm \xi_N)\}^\top
\)
and
\(
\bm Z=\{f_t(\bm \xi_1),\ldots,f_t(\bm \xi_N)\}^\top.
\)
Conditionally on \((\mathcal D_t,\Xi_m)\), \(\bm Z_m\) is Gaussian with mean
\(\bm\mu_m\) and covariance \(\bm\Sigma_m\). Conditionally on
\(\mathcal D_t\), \(\bm Z\) is Gaussian with mean \(\bm\mu\) and covariance
\(\bm\Sigma\). The previous display implies
\begin{equation*}
    \max_i|(\bm\mu_m-\bm\mu)_i|
+
\max_{i,j}|(\bm\Sigma_m-\bm\Sigma)_{ij}|
\le
C_{\mathrm{RFF},t}\varepsilon_m .
\end{equation*}
For \(s\in[0,1]\), let \(\bm Z_s\) be Gaussian with mean
\(\bm\mu_s=\bm\mu+s(\bm\mu_m-\bm\mu)\) and covariance
\(\bm\Sigma_s=\bm\Sigma+s(\bm\Sigma_m-\bm\Sigma)\). Gaussian interpolation gives
\begin{equation*}
    \frac{\dee}{\dee s}\EE \left\{H_{A,\mathcal Q}(\bm Z_s)\right\}
=
\sum_{i=1}^N(\bm\mu_m-\bm\mu)_i
\EE\{\partial_iH_{A,\mathcal Q}(\bm Z_s)\} +
\frac12
\sum_{i,j=1}^N(\bm\Sigma_m-\bm\Sigma)_{ij}
\EE\{\partial_{ij}H_{A,\mathcal Q}(\bm Z_s)\}.
\end{equation*}
Using the derivative bounds above yields
\begin{equation*}
    \left|
\EE \left\{H_{A,\mathcal Q}(\bm Z_m)\right\}
-
\EE \left\{H_{A,\mathcal Q}(\bm Z)\right\}
\right|
\le
C_{\tau,t}\varepsilon_m,
\end{equation*}
where \(C_{\tau,t}<\infty\) is independent of \(A\), \(\mathcal Q\), and \(N\).

Finally, take a sequence of finite Borel partitions
\(\mathcal Q_n=\{C_{n,1},\ldots,C_{n,N_n}\}\) of $\Omega$ such that
\(\max_{1\le i\le N_n}\operatorname{diam}(C_{n,i})\to0\).
For each continuous sample path \(h\), the corresponding piecewise-constant
approximations converge uniformly to \(h\) on compact \(\Omega\). The Gibbs
probability of \(A\) under the piecewise-constant approximation therefore
converges to \(\pi_\tau(A\mid h)\), and the convergence is dominated by one.
Applying this to \(f_{t,m}^{\rm dec}\) and \(f_t\), then using dominated
convergence, gives for every measurable \(A\subseteq\Omega\),
\begin{equation*}
    \left|
\pi_{\tau,t}^{(m)}(A\mid\Xi_m)
-
\pi_{\tau,t}(A)
\right|
\le
C_{\tau,t}\varepsilon_m .
\end{equation*}
Since the bound is uniform in \(A\), taking the supremum over measurable
\(A\subseteq\Omega\) proves the total-variation bound. Combining this with the
RFF rate proves the final statement.
\end{proof}

Retain the augmented latent state
\(
\bm z=(\bm w^\top,\bm\varepsilon^\top)^\top
\in\mathcal W\subset\mathbb R^{m+n_t}
\)
from \Cref{app:esa}, and write
\begin{equation*}
h_t^{(m)}(\para,\bm z)
:=
f_{t,m}^{\rm dec}(\para;\bm w,\bm\varepsilon).
\end{equation*}
Let \(\nu_t^{(m)}\) be the compactly supported reference law on
\(\mathcal W\) stipulated in \Cref{ass:oa_esa}. Define
\begin{equation*}
\Pi_{\tau,t}^{(m)}(\dee\para,\dee\bm z)
=
\pi_\tau\bigl(
\dee\para\mid h_t^{(m)}(\cdot,\bm z)
\bigr)
\nu_t^{(m)}(\dee\bm z),
\end{equation*}
using the AO notation of the main article. Let
\begin{equation*}
\pi_{\tau,t}^{(m),\mathcal W}
:=
\{\Pi_{\tau,t}^{(m)}\}_{\para},
\qquad
\eta_{m,\mathcal W}
:=
\TV\!\left\{
\pi_{\tau,t}^{(m),\mathcal W},
\pi_{\tau,t}^{(m)}(\cdot\mid\Xi_m)
\right\}.
\end{equation*}
Let \(K_{\tau,t}^{(m)}\) denote one sweep of the exact-auxiliary version of
the ESA algorithm in the main article, applied to this finite-dimensional
state.

\begin{theorem}
\label{thm:oa_esa}
Under the standing GP setup of the main article and \Cref{ass:oa_esa}, for
every fixed \(\tau>0\) and \(m\ge1\),
\(K_{\tau,t}^{(m)}\) leaves \(\Pi_{\tau,t}^{(m)}\) invariant and is uniformly
ergodic. In particular, there exists \(\delta_{\tau,t,m}>0\), depending on the
fixed \((t,\mathcal D_t,\Xi_m)\), such that
\(
K_{\tau,t}^{(m)}(\bm x,\cdot)
\ge
\delta_{\tau,t,m}\Pi_{\tau,t}^{(m)}(\cdot)\),
\(\forall \bm x\in\Omega\times\mathcal W,
\)
and therefore
\begin{equation*}
    \sup_{\bm x\in\Omega\times\mathcal W}
\TV\!\left\{
\{K_{\tau,t}^{(m)}\}^{s}(\bm x,\cdot),
\Pi_{\tau,t}^{(m)}
\right\}
\le
(1-\delta_{\tau,t,m})^s,
\qquad
s\ge1 .
\end{equation*}
\end{theorem}

\begin{proof}
The \(\para\)-move is an ordinary Metropolis-Hastings update and is reversible
with respect to the conditional law
\(\pi_\tau\bigl(\cdot\mid h_t^{(m)}(\cdot,\bm z)\bigr)\). For the
\(\bm z\)-move, consider the augmented proposal obtained by adding
\((\bm z',\bm u)\), where
\(
\bm z'\sim q_{\bm z}(\cdot\mid\bm z)\nu_t^{(m)}(\dee\bm z')\),
\(
\bm u\sim\pi_\tau\bigl(\cdot\mid h_t^{(m)}(\cdot,\bm z')\bigr).
\)
Swapping \(\bm z\) and \(\bm z'\) gives exactly the exchange ratio in the ESA
algorithm of the main article; the factors
\(\mathcal Z_\tau\bigl(h_t^{(m)}(\cdot,\bm z)\bigr)\) and
\(\mathcal Z_\tau\bigl(h_t^{(m)}(\cdot,\bm z')\bigr)\) cancel. This is the standard
exchange argument \citep{murray2006mcmc}. Hence the \(\bm z\)-move is
reversible with
respect to the correct conditional law, and the composition leaves
\(\Pi_{\tau,t}^{(m)}\) invariant.

It remains to prove uniform ergodicity. Let
\(
\mu_{\rm ref}^{(m)}:=\bar\lambda\otimes\nu_t^{(m)}.
\)
By \Cref{ass:oa_esa},
\begin{equation*}
    |h_t^{(m)}(\para,\bm z)|\le M_t^{(m)},
\quad
\forall(\para,\bm z)\in\Omega\times\mathcal W .
\end{equation*}
For the \(\para\)-move, the conditional density of
\(\pi_\tau\bigl(\cdot\mid h_t^{(m)}(\cdot,\bm z)\bigr)\) with respect to
\(\bar\lambda\) is bounded
above and below by \(\exp\{2M_t^{(m)}/\tau\}\) and
\(\exp\{-2M_t^{(m)}/\tau\}\), respectively. Hence the one-step \(\para\)-move kernel minorises \(\bar\lambda\) with constant
\begin{equation*}
    \varepsilon_{\para,\tau,t,m}
=
\underline q_{\para}^2
\exp\{-4M_t^{(m)}/\tau\}/\bar q_{\para}
>0.
\end{equation*}
For the exchange move, the one-step exchange kernel minorises
\(\nu_t^{(m)}\) with constant
\begin{equation*}
    \varepsilon_{\bm z,\tau,t,m}
=
\underline q_{\bm z}^2
\exp\{-4M_t^{(m)}/\tau\}/\bar q_{\bm z}
>0.
\end{equation*}
Thus one ESA sweep minorises \(\mu_{\rm ref}^{(m)}\) with constant
\(
\varepsilon_{\tau,t,m}
=
\varepsilon_{\para,\tau,t,m}
\varepsilon_{\bm z,\tau,t,m}.
\)

The density of \(\Pi_{\tau,t}^{(m)}\) with respect to \(\mu_{\rm ref}^{(m)}\) is bounded
above by \(\exp\{2M_t^{(m)}/\tau\}\), because
\begin{equation*}
    \mathcal Z_\tau\bigl(h_t^{(m)}(\cdot,\bm z)\bigr)
\ge
\lambda(\Omega)\exp\{-M_t^{(m)}/\tau\}.
\end{equation*}
Therefore,
\(
\mu_{\rm ref}^{(m)}(A)
\ge
\exp\{-2M_t^{(m)}/\tau\}\Pi_{\tau,t}^{(m)}(A)
\)
for every measurable \(A\subseteq\Omega\times\mathcal W\), and hence
\begin{equation*}
    K_{\tau,t}^{(m)}(\bm x,A)
\ge
\delta_{\tau,t,m}\Pi_{\tau,t}^{(m)}(A),
\qquad
\delta_{\tau,t,m}
=
\varepsilon_{\tau,t,m}
\exp\{-2M_t^{(m)}/\tau\}
>0.
\end{equation*}
Doeblin's condition gives the stated TV convergence
\citep{meyn2009markov}.
\end{proof}

\begin{corollary}
\label{cor:oa_esa_to_exact_ao}
Under the standing GP setup of the main article and
\Cref{ass:oa_rff,ass:oa_esa}, for fixed
\((t,\mathcal D_t,\tau)\), any initial state
\(\bm x\in\Omega\times\mathcal W\), and any \(s\ge1\),
\begin{equation*}
\TV\!\left\{
\{K_{\tau,t}^{(m)}\}^{s}(\bm x,\cdot)_{\para},
\pi_{\tau,t}
\right\}
\le
(1-\delta_{\tau,t,m})^s
+
\eta_{m,\mathcal W}
+
\TV\!\left\{
\pi_{\tau,t}^{(m)}(\cdot\mid\Xi_m),
\pi_{\tau,t}
\right\}.
\end{equation*}
Consequently,
\begin{equation*}
\TV\!\left\{
\{K_{\tau,t}^{(m)}\}^{s}(\bm x,\cdot)_{\para},
\pi_{\tau,t}
\right\}
\le
(1-\delta_{\tau,t,m})^s
+
\eta_{m,\mathcal W}
+
C_{\tau,t}\varepsilon_m .
\end{equation*}
Since \(\varepsilon_m=O_{\Pr_{\Xi_m}}(m^{-1/2})\), the final term is
\(O_{\Pr_{\Xi_m}}(m^{-1/2})\).
\end{corollary}

\begin{proof}
By the triangle inequality,
\begin{align*}
\TV\!\left\{
\{K_{\tau,t}^{(m)}\}^{s}(\bm x,\cdot)_{\para},
\pi_{\tau,t}
\right\}
&\le
\TV\!\left\{
\{K_{\tau,t}^{(m)}\}^{s}(\bm x,\cdot)_{\para},
\pi_{\tau,t}^{(m),\mathcal W}
\right\}
+
\TV\!\left\{
\pi_{\tau,t}^{(m),\mathcal W},
\pi_{\tau,t}^{(m)}(\cdot\mid\Xi_m)
\right\}
\\
&\quad+
\TV\!\left\{
\pi_{\tau,t}^{(m)}(\cdot\mid\Xi_m),
\pi_{\tau,t}
\right\}.
\end{align*}
Taking marginals cannot increase total variation distance, so the first term is
bounded by \((1-\delta_{\tau,t,m})^s\) by
\Cref{thm:oa_esa}. The second term is
\(\eta_{m,\mathcal W}\) by definition. The third term is controlled by
\Cref{prop:oa_rff_ao}.
\end{proof}

\begin{remark}
\label{rem:oa_esa_compact_idealization}
The result is levelwise and assumes an exact exchange auxiliary draw. The
compact \(\mathcal W\) condition is a proof device for Doeblin
minorisation; for a compact truncation of the Gaussian latent law,
\(\eta_{m,\mathcal W}\) records the corresponding truncation error. The
practical noncompact latent state and finite inner chain introduce errors
not controlled by this result.
\end{remark}

\vskip 0.2in
\bibliography{bibliography}

@book{wendland2004scattered,
  author = {Wendland, Holger},
  title = {{Scattered Data Approximation}},
  year = {2004},
  volume = {17},
  publisher = {Cambridge University Press},
}

@article{srinivas2012information,
  author = {Srinivas, Niranjan and Krause, Andreas and Kakade, Sham M. and Seeger, Matthias W.},
  title = {{Information-Theoretic Regret Bounds for Gaussian Process Optimization in the Bandit Setting}},
  journal = {IEEE Trans. Inf. Theory},
  year = {2012},
  volume = {58},
  number = {5},
  pages = {3250--3265},
}

@article{inla,
  author = {H{\aa}vard Rue and Sara Martino and Nicolas Chopin},
  title = {{Approximate Bayesian Inference for Latent Gaussian Models by Using Integrated Nested Laplace Approximations}},
  journal = {J. R. Stat. Soc. Ser. B},
  year = {2009},
  volume = {71},
  number = {2},
  pages = {319--392},
}

@article{aghqtheory,
  author = {Bilodeau, Blair and Stringer, Alex and Tang, Yanbo},
  title = {{Stochastic Convergence Rates and Applications of Adaptive Quadrature in Bayesian Inference}},
  journal = {J. Am. Stat. Assoc.},
  year = {2024},
  volume = {119},
  number = {545},
  pages = {690--700},
}

@article{bivand2014approximate,
  author = {Bivand, Roger S and G{\'o}mez-Rubio, Virgilio and Rue, H{\aa}vard},
  title = {{Approximate Bayesian Inference for Spatial Econometrics Models}},
  journal = {Spat. Stat.},
  year = {2014},
  volume = {9},
  pages = {146--165},
}

@article{gomez2018markov,
  author = {G{\'o}mez-Rubio, Virgilio and Rue, H{\aa}vard},
  title = {{Markov Chain Monte Carlo with the Integrated Nested Laplace Approximation}},
  journal = {Stat. Comput.},
  year = {2018},
  volume = {28},
  pages = {1033--1051},
}

@article{Bachl_2019,
  author = {Fabian E. Bachl and Finn Lindgren and David L. Borchers and Janine B. Illian},
  title = {{{inlabru}: An {R} Package for {Bayesian} Spatial Modelling from Ecological Survey Data}},
  journal = {Methods Ecol. Evol.},
  year = {2019},
  volume = {10},
  pages = {760--766},
}

@article{Li_2022,
  author = {Dayi Li and Gwendolyn M. Eadie and Roberto Abraham and Patrick E. Brown and William E. Harris and Steven R. Janssens and Aaron J. Romanowsky and Pieter van Dokkum and Shany Danieli},
  title = {{Light from the Darkness: Detecting Ultra-Diffuse Galaxies in the Perseus Cluster through Over-Densities of Globular Clusters with a Log-Gaussian Cox Process}},
  journal = {{ApJ}},
  year = {2022},
  volume = {935},
  number = {1},
  pages = {3},
}

@article{VanDokkum2015,
  author = {{van Dokkum}, Pieter and Abraham, Roberto and Merritt, Allison and Zhang, Jielai and Geha, Marla and Conroy, Charlie},
  title = {{Forty-Seven Milky Way-Sized, Extremely Diffuse Galaxies in the Coma Cluster}},
  journal = {{ApJL}},
  year = {2015},
  volume = {798},
  number = {2},
  pages = {L45},
}

@article{Abraham2014,
  author = {Abraham, Roberto G. and van Dokkum, Pieter},
  title = {{Ultra--Low Surface Brightness Imaging with the Dragonfly Telephoto Array}},
  journal = {{PASP}},
  year = {2014},
  volume = {126},
  number = {935},
  pages = {55},
}

@article{russo2014learning,
  author = {Russo, Daniel and Van Roy, Benjamin},
  title = {{Learning to Optimize via Posterior Sampling}},
  journal = {Math. Oper. Res.},
  year = {2014},
  volume = {39},
  number = {4},
  pages = {1221--1243},
}

@article{russo2018tutorial,
  author = {Russo, Daniel J. and Van Roy, Benjamin and Kazerouni, Abbas and Osband, Ian and Wen, Zheng},
  title = {{A Tutorial on Thompson Sampling}},
  journal = {Found. Trends Mach. Learn.},
  year = {2018},
  volume = {11},
  number = {1},
  pages = {1--96},
}

@book{rasmussen2006gaussian,
  author = {Rasmussen, Carl Edward and Williams, Christopher K. I.},
  title = {{Gaussian Processes for Machine Learning}},
  year = {2006},
  publisher = {MIT Press},
  address = {Cambridge, MA},
}

@article{dudley1967sizes,
    author = {Dudley, Richard M.},
    title = {{The Sizes of Compact Subsets of Hilbert Space and Continuity of Gaussian Processes}},
    journal = {J. Funct. Anal.},
    year = {1967},
    volume = {1},
    number = {3},
    pages = {290--330},
}

@book{adler2007random,
  author = {Adler, Robert J. and Taylor, Jonathan E.},
  title = {{Random Fields and Geometry}},
  year = {2007},
  publisher = {Springer},
  address = {New York},
}

@article{li2026boss,
  author = {Li, Dayi and Zhang, Ziang},
  title = {{Bayesian Optimization Sequential Surrogate ({BOSS}) Algorithm: Fast Bayesian Inference for a Broad Class of Bayesian Hierarchical Models}},
  journal = {Comput. Stat. Data Anal.},
  year = {2026},
  volume = {213},
  pages = {108253},
}

@article{wang2018adaptive,
  author = {Wang, Hongqiao and Li, Jinglai},
  title = {{Adaptive Gaussian Process Approximation for Bayesian Inference with Expensive Likelihood Functions}},
  journal = {Neural Comput.},
  year = {2018},
  volume = {30},
  number = {11},
  pages = {3072--3094},
}

@article{kim2025enhancing,
  author = {Kim, Hwanwoo and Sanz-Alonso, Daniel},
  title = {{Enhancing Gaussian Process Surrogates for Optimization and Posterior Approximation via Random Exploration}},
  journal = {SIAM/ASA J. Uncertain. Quantif.},
  year = {2025},
  volume = {13},
  number = {3},
  pages = {1054--1084},
}

@article{pellejeroibanez2020cosmological,
  author = {Pellejero-Iba{\~n}ez, Marcos and Angulo, Raul E. and Aric{\`o}, Giovanni and Zennaro, Matteo and Contreras, Sergio and St{\"u}cker, Jens},
  title = {{Cosmological Parameter Estimation via Iterative Emulation of Likelihoods}},
  journal = {{MNRAS}},
  year = {2020},
  volume = {499},
  number = {4},
  pages = {5257--5268},
}

@article{matsubara2024generalized,
  author = {Matsubara, Takuo and Knoblauch, Jeremias and Briol, Fran\c{c}ois-Xavier and Oates, Chris J.},
  title = {{Generalized Bayesian Inference for Discrete Intractable Likelihood}},
  journal = {J. Am. Stat. Assoc.},
  year = {2024},
  volume = {119},
  number = {547},
  pages = {2345--2355},
}

@article{bingham2024inverse,
  author = {Bingham, Derek and Butler, Troy and Estep, Don},
  title = {{Inverse Problems for Physics-Based Process Models}},
  journal = {Annu. Rev. Stat. Appl.},
  year = {2024},
  volume = {11},
  number = {1},
  pages = {461--482},
}

@article{lan2023spatiotemporal,
  author = {Lan, Shiwei and Li, Shuyi and Pasha, Mirjeta},
  title = {{Bayesian Spatiotemporal Modeling for Inverse Problems}},
  journal = {Stat. Comput.},
  year = {2023},
  volume = {33},
  number = {4},
  pages = {89},
}

@article{zeghal2025sbiwl,
  author = {Zeghal, Justine and Lanzieri, Denise and Lanusse, Fran{\c{c}}ois and Boucaud, Alexandre and Louppe, Gilles and Aubourg, Eric and Bayer, Adrian E. and The LSST Dark Energy Science Collaboration},
  title = {{Simulation-Based Inference Benchmark for Weak Lensing Cosmology}},
  journal = {{A\&A}},
  year = {2025},
  volume = {699},
  pages = {A327},
}

@article{lovell2025learning,
  author = {Lovell, Christopher C. and others},
  title = {{Learning the Universe: Cosmological and Astrophysical Parameter Inference with Galaxy Luminosity Functions and Colours}},
  journal = {{MNRAS}},
  year = {2025},
  volume = {544},
  number = {4},
  pages = {3949--3979},
}

@article{park2018bayesian,
  author = {Park, Jaewoo and Haran, Murali},
  title = {{Bayesian Inference in the Presence of Intractable Normalizing Functions}},
  journal = {J. Am. Stat. Assoc.},
  year = {2018},
  volume = {113},
  number = {523},
  pages = {1372--1390},
}

@article{park2020function,
  author = {Park, Jaewoo and Haran, Murali},
  title = {{A Function Emulation Approach for Doubly Intractable Distributions}},
  journal = {J. Comput. Graph. Stat.},
  year = {2020},
  volume = {29},
  number = {1},
  pages = {66--77},
}

@article{jarvenpaa2021parallel,
  author = {J{\"a}rvenp{\"a}{\"a}, Marko and Gutmann, Michael U. and Vehtari, Aki and Marttinen, Pekka},
  title = {{Parallel Gaussian Process Surrogate Bayesian Inference with Noisy Likelihood Evaluations}},
  journal = {Bayesian Anal.},
  year = {2021},
  volume = {16},
  number = {1},
  pages = {147--178},
}

@article{elgammal2023gpry,
  author = {El Gammal, Jonas and Sch{\"o}neberg, Nils and Torrado, Jes{\'u}s and Fidler, Christian},
  title = {{Fast and Robust Bayesian Inference Using Gaussian Processes with GPry}},
  journal = {{JCAP}},
  year = {2023},
  volume = {2023},
  number = {10},
  pages = {021},
}

@article{gutmann2016bolfi,
  author = {Gutmann, Michael U. and Corander, Jukka},
  title = {{Bayesian Optimization for Likelihood-Free Inference of Simulator-Based Statistical Models}},
  journal = {J. Mach. Learn. Res.},
  year = {2016},
  volume = {17},
  number = {125},
  pages = {1--47},
}

@inproceedings{acerbi2018vbmc,
  author = {Acerbi, Luigi},
  title = {Variational {Bayesian} {Monte Carlo}},
  booktitle = {{NeurIPS 2018}},
  year = {2018},
  pages = {8213--8223},
  editor = {Bengio, Samy and Wallach, Hanna and Larochelle, Hugo and Grauman, Kristen and Cesa-Bianchi, Nicol{\`o} and Garnett, Roman},
  publisher = {Curran Associates, Inc.},
}

@book{settles2012active,
  author = {Settles, Burr},
  title = {{Active Learning}},
  year = {2012},
  volume = {6},
  series = {Synthesis Lectures on Artificial Intelligence and Machine Learning},
  publisher = {Morgan \& Claypool Publishers},
}

@article{surer2024sequential,
  author = {S{\"u}rer, {\"O}zge and Plumlee, Matthew and Wild, Stefan M.},
  title = {{Sequential Bayesian Experimental Design for Calibration of Expensive Simulation Models}},
  journal = {Technometrics},
  year = {2024},
  volume = {66},
  number = {2},
  pages = {157--171},
}

@article{shahriari2016review,
  author = {Shahriari, Bobak and Swersky, Kevin and Wang, Ziyu and Adams, Ryan P. and de Freitas, Nando},
  title = {{Taking the Human Out of the Loop: A Review of Bayesian Optimization}},
  journal = {Proc. IEEE},
  year = {2016},
  volume = {104},
  number = {1},
  pages = {148--175},
}

@article{lartaud2024sequential,
  author = {Paul Lartaud and Philippe Humbert and Josselin Garnier},
  title = {{Solving Bayesian Inverse Problems Using Gaussian Process Regression with Goal-Oriented Active Learning}},
  journal = {Technometrics},
  year = {2026},
  volume = {68},
  number = {1},
  pages = {172--185},
}

@inproceedings{seo2000gaussian,
  author = {Seo, Sambu and Wallat, Marko and Graepel, Thore and Obermayer, Klaus},
  title = {{Gaussian} process regression: active data selection and test point rejection},
  booktitle = {{IJCNN 2000}},
  year = {2000},
  volume = {3},
  pages = {241--246},
}

@book{gramacy2020surrogates,
  author = {Gramacy, Robert B.},
  title = {{Surrogates: Gaussian Process Modeling, Design, and Optimization for the Applied Sciences}},
  year = {2020},
  publisher = {Chapman \& Hall/CRC},
}

@article{kandasamy2017query,
  author = {Kandasamy, Kirthevasan and Schneider, Jeff and P{\'o}czos, Barnab{\'a}s},
  title = {{Query Efficient Posterior Estimation in Scientific Experiments via Bayesian Active Learning}},
  journal = {Artif. Intell.},
  year = {2017},
  volume = {243},
  pages = {45--56},
}

@article{jarvenpaa2019efficient,
  author = {J{\"a}rvenp{\"a}{\"a}, Marko and Gutmann, Michael U. and Pleska, Arijus and Vehtari, Aki and Marttinen, Pekka},
  title = {{Efficient Acquisition Rules for Model-Based Approximate Bayesian Computation}},
  journal = {Bayesian Anal.},
  year = {2019},
  volume = {14},
  number = {2},
  pages = {595--622},
}

@inproceedings{acerbi2019exploration,
  author = {Acerbi, Luigi},
  title = {An exploration of acquisition and mean functions in variational {Bayesian} {Monte Carlo}},
  booktitle = {{AABI 2019}},
  year = {2019},
  volume = {96},
  pages = {1--10},
}

@article{roberts2026surrogatebased,
  author = {Roberts, Andrew Gerard and Dietze, Michael C. and Huggins, Jonathan H.},
  title = {{Surrogate-Based Bayesian Inference: Uncertainty Quantification and Active Learning}},
  journal = {arXiv preprint arXiv:2603.13646},
  year = {2026},
  eprint = {2603.13646},
  archiveprefix = {arXiv},
  primaryclass = {stat.ME},
}

@article{reiser2025uncertainty,
  author = {Reiser, Philipp and Aguilar, Javier Enrique and Guthke, Anneli and B{\"u}rkner, Paul-Christian},
  title = {{Uncertainty Quantification and Propagation in Surrogate-Based Bayesian Inference}},
  journal = {Stat. Comput.},
  year = {2025},
  volume = {35},
  number = {3},
  pages = {66},
}

@article{roberts2026propagating,
  author = {Roberts, Andrew Gerard and Dietze, Michael and Huggins, Jonathan H.},
  title = {{Propagating Surrogate Uncertainty in Bayesian Inverse Problems}},
  journal = {arXiv preprint arXiv:2601.03532},
  year = {2026},
  eprint = {2601.03532},
  archiveprefix = {arXiv},
  primaryclass = {stat.ME},
}

@article{li2025_poisson,
  author = {Li, Dayi and Stringer, Alex and Brown, Patrick E. and Eadie, Gwendolyn M. and Abraham, Roberto G.},
  title = {{Poisson Cluster Process Models for Detecting Ultra-Diffuse Galaxies}},
  journal = {Ann. Appl. Stat.},
  year = {2025},
  volume = {19},
  number = {1},
  pages = {261--285},
}

@article{li2025CDG2,
  author = {Li, Dayi and Liu, Qing and Eadie, Gwendolyn M. and Abraham, Roberto G. and Marleau, Francine R. and Harris, William E. and van Dokkum, Pieter and Romanowsky, Aaron J. and Danieli, Shany and Brown, Patrick E. and Stringer, Alex},
  title = {{Candidate Dark Galaxy-2: Validation and Analysis of an Almost Dark Galaxy in the Perseus Cluster}},
  journal = {{ApJL}},
  year = {2025},
  volume = {986},
  number = {2},
  pages = {L18},
}

@article{li2025mathpop,
  author = {Li, Dayi and Eadie, Gwendolyn M. and Brown, Patrick E. and Harris, William E. and Abraham, Roberto G. and van Dokkum, Pieter and Janssens, Steven R. and Berek, Samantha C. and Danieli, Shany and Romanowsky, Aaron J. and Speagle, Joshua S.},
  title = {{Discovery of Two Ultra-Diffuse Galaxies with Unusually Bright Globular Cluster Luminosity Functions via a Mark-Dependently Thinned Point Process (MATHPOP)}},
  journal = {{ApJ}},
  year = {2025},
  volume = {984},
  number = {2},
  pages = {147},
}

@article{moller1998log,
  author = {M{\o}ller, Jesper and Syversveen, Anne Randi and Waagepetersen, Rasmus Plenge},
  title = {{Log Gaussian Cox Processes}},
  journal = {Scand. J. Stat.},
  year = {1998},
  volume = {25},
  number = {3},
  pages = {451--482},
}

@article{harris2020piper,
  author = {Harris, William E. and Brown, Rachel A. and Durrell, Patrick R. and Romanowsky, Aaron J. and Blakeslee, John and Brodie, Jean and Janssens, Steven and Lisker, Thorsten and Okamoto, Sakurako and Wittmann, Carolin},
  title = {{The PIPER Survey. I. An Initial Look at the Intergalactic Globular Cluster Population in the Perseus Cluster}},
  journal = {{ApJ}},
  year = {2020},
  volume = {890},
  number = {2},
  pages = {105},
}

@article{roberts2009examples,
  author = {Roberts, Gareth O. and Rosenthal, Jeffrey S.},
  title = {{Examples of Adaptive MCMC}},
  journal = {J. Comput. Graph. Stat.},
  year = {2009},
  volume = {18},
  number = {2},
  pages = {349--367},
}

@incollection{delahaye2019simulated,
  author = {Delahaye, Daniel and Chaimatanan, Supatcha and Mongeau, Marcel},
  title = {{Simulated Annealing: From Basics to Applications}},
  booktitle = {Handbook of Metaheuristics},
  year = {2019},
  volume = {272},
  pages = {1--35},
  editor = {Gendreau, Michel and Potvin, Jean-Yves},
  series = {International Series in Operations Research \& Management Science},
  publisher = {Springer},
  address = {Cham},
  edition = {3},
}

@article{van_Dokkum_2024,
  author = {Pieter van Dokkum and Dayi Li and Roberto Abraham and Shany Danieli and Gwendolyn M. Eadie and William E. Harris and Aaron J. Romanowsky},
  title = {{Deep HST/UVIS Imaging of the Candidate Dark Galaxy CDG-1}},
  journal = {{RNAAS}},
  year = {2024},
  volume = {8},
  number = {5},
  pages = {135},
}

@article{berild2022importance,
  author = {Berild, Martin Outzen and Martino, Sara and G{\'o}mez-Rubio, Virgilio and Rue, H{\aa}vard},
  title = {{Importance Sampling with the Integrated Nested Laplace Approximation}},
  journal = {J. Comput. Graph. Stat.},
  year = {2022},
  volume = {31},
  number = {4},
  pages = {1225--1237},
}

@inproceedings{rahimi2007random,
  author = {Rahimi, Ali and Recht, Benjamin},
  title = {Random features for large-scale kernel machines},
  booktitle = {{NeurIPS 2007}},
  year = {2007},
  pages = {1177--1184},
  publisher = {Curran Associates, Inc.},
}

@inproceedings{sutherland2015error,
  author = {Sutherland, Dougal J. and Schneider, Jeff},
  title = {On the error of random {Fourier} features},
  booktitle = {{UAI 2015}},
  year = {2015},
  pages = {862--871},
  publisher = {AUAI Press},
}

@article{matheron1973intrinsic,
  author = {Matheron, Georges},
  title = {{The Intrinsic Random Functions and Their Applications}},
  journal = {Adv. Appl. Probab.},
  year = {1973},
  volume = {5},
  number = {3},
  pages = {439--468},
}

@article{wilson2021pathwise,
  author = {Wilson, James T. and Borovitskiy, Viacheslav and Terenin, Alexander and Mostowsky, Peter and Deisenroth, Marc Peter},
  title = {{Pathwise Conditioning of Gaussian Processes}},
  journal = {J. Mach. Learn. Res.},
  year = {2021},
  volume = {22},
  number = {105},
  pages = {1--47},
}

@inproceedings{murray2006mcmc,
  author = {Murray, Iain and Ghahramani, Zoubin and MacKay, David J. C.},
  title = {{MCMC} for doubly-intractable distributions},
  booktitle = {{UAI 2006}},
  year = {2006},
  pages = {359--366},
  publisher = {AUAI Press},
}

@book{meyn2009markov,
  author = {Meyn, Sean P. and Tweedie, Richard L.},
  title = {{Markov Chains and Stochastic Stability}},
  year = {2009},
  publisher = {Cambridge University Press},
  edition = {2},
}

@article{picheny2010adaptive,
  author = {Picheny, Victor and Ginsbourger, David and Roustant, Olivier and Haftka, Raphael T. and Kim, Nam-Ho},
  title = {{Adaptive Designs of Experiments for Accurate Approximation of a Target Region}},
  journal = {J. Mech. Des.},
  year = {2010},
  volume = {132},
  number = {7},
  pages = {071008},
}

@article{bect2019supermartingale,
  author = {Bect, Julien and Bachoc, Fran{\c{c}}ois and Ginsbourger, David},
  title = {{A Supermartingale Approach to Gaussian Process Based Sequential Design of Experiments}},
  journal = {Bernoulli},
  year = {2019},
  volume = {25},
  number = {4A},
  pages = {2883--2919},
}

@article{Rust1987,
  author = {Rust, John},
  title = {{Optimal Replacement of {GMC} Bus Engines: An Empirical Model of Harold Zurcher}},
  journal = {Econometrica},
  year = {1987},
  volume = {55},
  number = {5},
  pages = {999--1033},
}

@article{ImaiJainChing2009,
  author = {Imai, Susumu and Jain, Neelam and Ching, Andrew T.},
  title = {{Bayesian Estimation of Dynamic Discrete Choice Models}},
  journal = {Econometrica},
  year = {2009},
  volume = {77},
  number = {6},
  pages = {1865--1899},
}

@article{Norets2012,
  author = {Norets, Andriy},
  title = {{Estimation of Dynamic Discrete Choice Models Using Artificial Neural Network Approximations}},
  journal = {Econom. Rev.},
  year = {2012},
  volume = {31},
  number = {1},
  pages = {84--106},
}

@article{NoretsShimizu2024,
  author = {Norets, Andriy and Shimizu, Kenichi},
  title = {{Semiparametric Bayesian Estimation of Dynamic Discrete Choice Models}},
  journal = {J. Econom.},
  year = {2024},
  volume = {238},
  number = {2},
  pages = {105642},
}

@article{Norets2009,
  author = {Norets, Andriy},
  title = {{Inference in Dynamic Discrete Choice Models with Serially Correlated Unobserved State Variables}},
  journal = {Econometrica},
  year = {2009},
  volume = {77},
  number = {5},
  pages = {1665--1682},
}

@article{Chen2025,
  author = {Chen, Ertian},
  title = {{Model-Adaptive Approach to Dynamic Discrete Choice Models with Large State Spaces}},
  journal = {arXiv preprint arXiv:2501.18746},
  year = {2025},
  eprint = {2501.18746},
  archiveprefix = {arXiv},
  primaryclass = {econ.EM},
}

@incollection{Rust1994,
  author = {Rust, John},
  title = {{Structural Estimation of Markov Decision Processes}},
  booktitle = {Handbook of Econometrics},
  year = {1994},
  volume = {4},
  pages = {3081--3143},
  editor = {Engle, Robert F. and McFadden, Daniel L.},
  publisher = {Elsevier},
  address = {Amsterdam},
  chapter = {51},
}

@article{AguirregabiriaMira2010,
  author = {Aguirregabiria, Victor and Mira, Pedro},
  title = {{Dynamic Discrete Choice Structural Models: A Survey}},
  journal = {J. Econom.},
  year = {2010},
  volume = {156},
  number = {1},
  pages = {38--67},
}

@article{whitehouse2023sublinear,
  author = {Whitehouse, Justin and Wu, Zhiwei Steven and Ramdas, Aaditya},
  title = {{On the Sublinear Regret of {GP-UCB}}},
  journal = {arXiv preprint arXiv:2307.07539},
  year = {2023},
  eprint = {2307.07539},
  archiveprefix = {arXiv},
  primaryclass = {cs.LG},
}

@article{wilson2024stopping,
  author = {Wilson, James T.},
  title = {{Stopping Bayesian Optimization with Probabilistic Regret Bounds}},
  journal = {arXiv preprint arXiv:2402.16811},
  year = {2024},
}

@article{roberts1996exponential,
  author = {Roberts, Gareth O. and Tweedie, Richard L.},
  title = {{Exponential convergence of Langevin distributions and their discrete approximations}},
  journal = {Bernoulli},
  year = {1996},
  volume = {2},
  number = {4},
  pages = {341--363},
}

@article{roberts1998optimal,
  author = {Roberts, Gareth O. and Rosenthal, Jeffrey S.},
  title = {{Optimal scaling of discrete approximations to Langevin diffusions}},
  journal = {J. R. Stat. Soc. Ser. B},
  year = {1998},
  volume = {60},
  number = {1},
  pages = {255--268},
}

@article{haario2001adaptive,
  author = {Haario, Heikki and Saksman, Eero and Tamminen, Johanna},
  title = {{An adaptive Metropolis algorithm}},
  journal = {Bernoulli},
  year = {2001},
  volume = {7},
  number = {2},
  pages = {223--242},
}

@book{nocedal2006numerical,
  author = {Nocedal, Jorge and Wright, Stephen J.},
  title = {{Numerical Optimization}},
  year = {2006},
  publisher = {Springer},
  edition = {2},
}

@article{jarvenpaa2024approximate,
  author = {J{\"a}rvenp{\"a}{\"a}, Marko and Corander, Jukka},
  title = {{Approximate Bayesian Inference from Noisy Likelihoods with Gaussian Process Emulated MCMC}},
  journal = {J. Mach. Learn. Res.},
  year = {2024},
  volume = {25},
  number = {366},
  pages = {1--55},
}

@article{hwang1980laplace,
  author = {Hwang, Chii-Ruey},
  title = {{Laplace's Method Revisited: Weak Convergence of Probability Measures}},
  journal = {Ann. Probab.},
  year = {1980},
  volume = {8},
  number = {6},
  pages = {1177--1182},
}

@inproceedings{iwazaki2025improvedregret,
  author = {Iwazaki, Shogo},
  title = {Improved regret bounds for {Gaussian} process upper confidence bound in {Bayesian} optimization},
  booktitle = {{NeurIPS 2025}},
  year = {2025},
  volume = {38},
}

@article{santin2016approximation,
  author = {Santin, Gabriele and Schaback, Robert},
  title = {{Approximation of Eigenfunctions in Kernel-Based Spaces}},
  journal = {Adv. Comput. Math.},
  year = {2016},
  volume = {42},
  number = {4},
  pages = {973--993},
}

@inproceedings{oliveira2021noregret,
  author = {Oliveira, Rafael and Ott, Lionel and Ramos, Fabio},
  title = {No-regret approximate inference via {Bayesian} optimisation},
  booktitle = {{UAI 2021}},
  year = {2021},
  volume = {161},
  pages = {2082--2092},
  series = {{PMLR}},
  publisher = {PMLR},
}

@article{villani2024posterior,
  author = {Villani, Paolo and Andr{\'e}s-Arcones, Daniel and Unger, J{\"o}rg F. and Weiser, Martin},
  title = {{Posterior Sampling with Adaptive Gaussian Processes in Bayesian Parameter Identification}},
  journal = {arXiv preprint arXiv:2411.17858},
  year = {2024},
}

@inproceedings{helin2024introduction,
  author = {Helin, Tapio and Stuart, Andrew M. and Teckentrup, Aretha L. and Zygalakis, Konstantinos C.},
  title = {Introduction to {Gaussian} process regression in {Bayesian} inverse problems, with new results on experimental design for weighted error measures},
  booktitle = {{MCQMC 2024}},
  year = {2024},
  volume = {460},
  pages = {49--79},
  series = {Springer Proceedings in Mathematics \& Statistics},
  publisher = {Springer},
}

@article{llorente2020adaptive,
  author = {Llorente, Fernando and Martino, Luca and Elvira, V{\'i}ctor and Delgado-G{\'o}mez, David and L{\'o}pez-Santiago, Javier},
  title = {{Adaptive Quadrature Schemes for Bayesian Inference via Active Learning}},
  journal = {IEEE Access},
  year = {2020},
  volume = {8},
  pages = {208462--208483},
}

@article{vehtari2021rank,
  author  = {Vehtari, Aki and Gelman, Andrew and Simpson, Daniel and
             Carpenter, Bob and B{\"u}rkner, Paul-Christian},
  title   = {Rank-Normalization, Folding, and Localization:
             An Improved {$\widehat{R}$} for Assessing Convergence
             of {MCMC} (with Discussion)},
  journal = {Bayesian Anal.},
  year    = {2021},
  volume  = {16},
  number  = {2},
  pages   = {667--718},
}

\end{document}